\documentclass[runningheads]{llncs}

\usepackage{graphicx}

\usepackage{todonotes}
\usepackage[american]{babel}
\usepackage{wrapfig}
\usepackage{listings}
\usepackage[section]{placeins}
\usepackage{amsmath}
\usepackage{amssymb}
\usepackage{todonotes}
\usepackage{mathtools}
\usepackage[noend,noline,ruled,linesnumbered]{algorithm2e}
\usepackage{dsfont}
\usepackage{xcolor}
\usepackage{skull}
\usepackage{tikz}
\usepackage{pgfplots}

\usetikzlibrary{fit} 
\usetikzlibrary{backgrounds} 
\usetikzlibrary{patterns,calc}
\usetikzlibrary{decorations.pathreplacing}
\usetikzlibrary{positioning,fit, backgrounds, calc}
\usetikzlibrary{decorations.pathmorphing}
\usetikzlibrary{decorations.markings}
\tikzstyle{background}=[rectangle,fill=gray!10, inner sep=0.1cm, rounded corners=0mm]
\usepgflibrary{shapes}
\usetikzlibrary{automata,arrows}
\usetikzlibrary{shadows}
\tikzstyle{nloc}=[draw, text badly centered, rectangle, rounded corners, minimum size=2em,inner sep=0.5em]
\tikzstyle{background}=[rectangle,fill=gray!10, inner sep=0.1cm, rounded corners=0mm]
\tikzstyle{loc}=[draw,rectangle,minimum size=1.4em,inner sep=0em]
\tikzset{
    gluon/.style={decorate,draw=black,
        decoration={coil,amplitude=1pt, segment length=5pt}} 
}
\tikzset{
    gluonew/.style={decorate,draw=black,
        decoration={coil,amplitude=1pt, segment length=2pt}} 
}

\tikzset{
    gluon1/.style={decorate,draw=black,
        decoration={coil,amplitude=3pt, segment length=3pt}} 
}

\newcommand{\mtpda}[1]{MTPDA}
\newcommand{\mpda}{\ensuremath{\mathsf{MPDA}}\xspace}
\newcommand{\tmpda}{\ensuremath{\mathsf{TMPDA}}\xspace}

\newcommand{\ws}{\mathsf{WS}}
\newcommand{\remove}{\mathsf{remove}}
\newcommand{\append}{\mathsf{append}}
\newcommand{\trunc}{\mathsf{trunc}}
\newcommand{\last}{\mathsf{last}}
\newcommand{\maxi}{\mathsf{max}}
\newcommand{\hole}{\mathsf{hole}}
\newcommand{\pushhole}{\texttt{hole}}

\newcommand{\textsub}{\textsubscript}
\newcommand{\textsup}{\textsuperscript}
\newcommand{\nop}{\ensuremath{\mathsf{nop}}\xspace}
\newcommand{\push}{}
\newcommand{\pop}{}
\newcommand{\op}{\ensuremath{\mathsf{op}}\xspace}
\newcommand{\src}{\ensuremath{\mathsf{src}}\xspace}
\newcommand{\tgt}{\ensuremath{\mathsf{tgt}}\xspace}
\newcommand{\trim}{\ensuremath{\mathsf{BHIM}}}
\newcommand{\cmax}{c\textsubscript{max}}
\newcommand{\sol}{\texttt{SetOfLists}}
\newcommand{\wsr}{\texttt{WR}}
\newcommand{\wnr}{\mathsf{WR}}

\newcommand{\wsrt}{\texttt{WRT}}
\newcommand{\atcrit}{\ensuremath{\mathcal{A}_{Tcrit}^{\emptyset}}}
\newcommand{\acrit}{\ensuremath{\mathcal{A}_{crit}^{\emptyset}}}
\newcommand{\acritn}{\ensuremath{\mathcal{A}_{crit}}}

\newcommand{\aprodp}{\ensuremath{\mathcal{A}_{prodcon}^{(M,N)}}}

\usepackage{booktabs}

\makeatletter
 \RequirePackage[bookmarks,unicode,colorlinks=true]{hyperref}%

   \def\@citecolor{blue}%
   \def\@urlcolor{blue}%
   \def\@linkcolor{blue}%

\def\orcidID#1{\smash{\href{http://orcid.org/#1}{\protect\raisebox{-1.25pt}{\protect\includegraphics{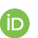}}}}}
\makeatother

\begin{document}
\title{Revisiting Underapproximate Reachability for Multipushdown
  Systems\thanks{Partly supported by UMI ReLaX, DST/CEFIPRA/INRIA
    project EQuaVe \& TCS.}}

\author{S. Akshay\inst{1} \and
Paul Gastin\inst{2} \and
S Krishna\inst{1}\and Sparsa Roychowdhury\inst{1}\orcidID{0000-0003-3583-7612}}
\authorrunning{S. Akshay, P. Gastin, S. Krishna, S. Roychowdhury}

\institute{IIT Bombay, Mumbai, India \and
ENS Paris-Saclay, Paris, France\\}
\maketitle             
\begin{abstract}
Boolean programs with multiple recursive threads can be captured as
  pushdown automata with multiple stacks.  This model is Turing complete, and
  hence, one is often interested in analyzing a restricted class that still
  captures useful behaviors.  In this paper, we propose a new class of bounded
  underapproximations for multi-pushdown systems, which subsumes most existing
  classes.  We develop an efficient algorithm for solving the under-approximate
  reachability problem, which is based on efficient fix-point computations.  We
  implement it in our tool \trim{}  and illustrate its applicability by
  generating a set of relevant benchmarks and examining its performance.  As an
  additional takeaway \trim{} solves the binary reachability problem in pushdown
  automata.  To show the versatility of our approach, we then extend our
  algorithm to the timed setting and provide the first implementation that can
  handle timed multi-pushdown automata with closed guards.

\end{abstract}
\section{Introduction}
The reachability problem for pushdown systems with multiple stacks is known to be undecidable. However, multi-stack pushdown automata (MPDA hereafter) represent a theoretically concise and analytically useful model of multi-threaded recursive programs with shared memory. As a result, several previous works in the literature have proposed different under-approximate classes of behaviors of MPDA that can be analyzed effectively, such as \emph{ Round Bounded}, \emph{Scope Bounded}, \emph{Context Bounded} and \emph{Phase Bounded}~\cite{la2007robust,la2010language,torre2016scope,cyriac2012mso,la2011reachability,la2012scope}. From a practical point of view, these underapproximations have led to efficient tools including, GetaFix~\cite{la2009analyzing}, SPADE~\cite{patin2007spade}. It has also been argued (e.g., see~\cite{Qadeer08}) that such bounded underapproximations suffice to find several bugs in practice.  In many such tools efficient fix-point techniques are used to speed-up computations.

We extend known fix-point based approaches by developing a new algorithm that
can handle a larger class of bounded underapproximations than the well-known bounded context
and bounded scope underapproximations for multi-pushdown systems while remaining efficiently implementable.  Our algorithm works for a new class of underapproximate behaviors called \emph{hole bounded} behaviors, which subsumes context/scope bounded underapproximations, and is orthogonal to phase bounded underapproximations. A ``hole'' is a maximal sequence of push operations of a fixed stack, interspersed with well-nested sequences of any stack. Thus, in a sequence $\alpha=\beta \gamma$ where 
 $\beta=[push_1(\textcolor{red}{push_2}\textcolor{blue}{push_3}$
$\textcolor{blue}{pop_3}\textcolor{red}{pop_2})push_1(\textcolor{blue}{push_3}\textcolor{blue}{pop_3})]^{10}$
and {\small $\gamma=\textcolor{red}{push_2} push_1 \textcolor{red}{pop_2} pop_1(pop_1)^{20}$,
$\beta$} is a hole with respect to stack 1. 
The suffix $\gamma$ has 2 holes (the $push_2$ and the $push_1$).  
Thus we say that $\alpha$ is 3-hole bounded. On the other hand, the number of context switches (and scope bound) in $\alpha$ is $> 50$.  A ($k$-)hole bounded sequence  is one such, where,  at any point of the computation, the number of ``open'' holes are bounded at this point (by $k$). We show that the class of hole bounded sequences subsumes most of the previously defined classes of underapproximations  and is, in fact, contained in the very generic class of tree-width bounded  sequences. This immediately shows decidability of the reachability problem for our class.

  Analyzing the more generic class  of tree-width bounded sequences is often much more difficult; for instance, building bottom-up tree automata for this purpose does not
 scale very well as it explores a large  (and often useless) state
 space. Our technique is radically  different from using tree automata. Under the hole bounded assumption, we pre-compute information regarding well-nested sequences and holes using fix-point computations and use them in our algorithm. Using efficient data structures to implement this approach, we develop a tool (\trim{}) for  \textcolor{red}{B}ounded \textcolor{red}{H}ole reachability \textcolor{red}{i}n \textcolor{red}{M}ulti-stack pushdown systems.
 
\noindent{\bf Highlights of \trim{}}.

\noindent $\bullet$ Two significant aspects  of the fix-point approach in \trim{} are: (i) we efficiently  solve the binary reachability problem for   pushdown automata. i.e.,  \trim{} computes all pairs of states $(s,t)$ such that $t$ is reachable from $s$ with empty stacks. This allows us to go beyond reachability and handle some liveness questions; (ii) we pre-compute the set of pairs of states that are endpoints of holes. This allows us to greatly limit the search for an accepting run.

\noindent $\bullet$  While the fix-point approach solves (binary) reachability efficiently, it does not a priori produce a witness of reachability.  We remedy this situation by proposing a backtracking algorithm, which cleverly 
uses the computations done in the fix-point algorithm,  to generate a witness efficiently.

\noindent $\bullet$ \trim{} is parametrized with respect to the hole bound:  if
non-emptiness can be checked or witnessed by a well-nested sequence (this is an
easy witness and \trim{} looks for easy witnesses first, then
gradually increases complexity, if no easy witness is found), then it is sufficient to have the hole bound 0. Increasing this complexity measure as required to certify non-emptiness gives an efficient implementation, in the sense that we search for harder witnesses only when no easier witnesses (w.r.t this complexity measure) exist. In examples described in the experimental section, a small (less than 4) bound suffices and we expect this to be the case for most practical examples. 

\noindent $\bullet$ Finally, we extend our approach to handle timed multi-stack pushdown systems. This shows the versatility of our approach and also requires us to solve several technical challenges which are specific to the timed setting.  Implementing this approach in \trim{} makes it, to the best of our knowledge, the first tool that can analyze timed multi-stack pushdown automata with closed guards.
 
 We analyze the performance of \trim{} in practice, by considering
benchmarks from the literature, and generating timed variants of some of
them. One of our benchmarks is a variant of the Bluetooth example \cite{chaki2006verifying,patin2007spade}, where \trim{} was able to catch a known race detection error.
Another interesting benchmark is a model of a parameterized multiple producer consumer example, having parameters $M, N$ on the quantities of two items $A, B$ produced.  Here, \trim{} could detect bugs by finding witnesses having just 2 holes,
while, it is unlikely that existing tools working on scope/context bounded underapproximations can handle them as the number of scope/context switches is dependent on $M,N$ (in fact, it is twice the least common multiple of $M$ and $N$). In the timed setting, one of the main challenges has been the unavailability of timed benchmarks; even in the untimed setting, many benchmarks were unavailable due to their proprietary nature. Due to lack of space, proofs, technical details and parametric plots of experiments are in~\cite{supp}.

   \noindent{\bf Related Work}.  
   Among other under-approximations, scope bounded~\cite{torre2016scope} subsumes context and round bounded underapproximations, and it also paves path for
   GetaFix~\cite{la2009analyzing}, a tool to analyze recursive (and multi-threaded) boolean programs. As mentioned earlier hole boundedness strictly subsumes scope boundedness. On the other hand, GetaFix uses symbolic approaches via BDDs, which is orthogonal to the improvements made in this paper. Indeed, our next step would be to build a symbolic version of \trim{} which extends the hole-bounded approach to work with symbolic methods. Given that \trim{} can already handle synthetic examples with 12-13 holes (see~\cite{supp}), we expect this to lead to even more drastic improvements and applicability.
   For sequential programs, a summary-based
   algorithm is used in~\cite{la2009analyzing}; summaries are like our
   well-nested sequences, except that well-nested sequences admit contexts from
   different stacks unlike summaries.  As a result, our class of bounded hole
   behaviors generalizes summaries.  Many other different theoretical
   results like phase bounded~\cite{la2007robust}, order
   bounded~\cite{atig2012model} which gives interesting underapproximations of
   \mpda{}, are subsumed in tree-width bounded behaviors, but they do not seem to have practical implementations. 
   Adding real-time information to pushdown automata by using clocks or timed stacks has been considered, both in the discrete and dense-timed settings. Recently, there has been a flurry of theoretical results in the topic~\cite{bouajjani1994automatic,AbdullaAS12,DBLP:conf/lata/AbdullaAS12,AGK16,AkshayGKS17}. However, to the best of our knowledge none of these algorithms have been successfully implemented (except~\cite{AkshayGKS17} which implements a tree-automata based technique for single-stack timed systems) for multi-stack systems. One reason is that these algorithms do not employ scalable fix-point based techniques, but instead depend on region automaton-based search or tree automata-based search techniques.

\section{Underapproximations in \mpda}\label{sec:2}
\vspace*{-2mm}

    A multi-stack pushdown automaton (\mpda)  is a tuple 
    $M=(\mathcal{S},\Delta, s_0,\mathcal{S}_f,$ $n,\Sigma,\Gamma)$ 
        where, 
        $\mathcal{S} $ is a finite non-empty set of locations,
        $\Delta$ is a finite set of transitions,
        $s_0\in \mathcal{S}$ is the initial location,
        $\mathcal{S}_f \subseteq \mathcal{S}$ is a set of final locations,
        $n\in \mathbb{N}$ is the number of stacks,
        $\Sigma $ is a finite input  alphabet,
        and $\Gamma$ is a finite stack alphabet which contains $\bot$. 
        A transition $t \in \Delta$ can be represented as a tuple
        $(s,\text{\op{}},a,s')$, where, $s,s' \in \mathcal{S}$ are
        respectively, the source and destination locations of the
        transition $t$, $a \in \Sigma$ is the label of the transition,
       and \op{} is one of the following
        operations
            (1) \nop, or no stack operation,
        (2) $\push(\downarrow_i{\alpha})$
        which pushes $\alpha \in \Gamma$
            onto stack $i \in \{1,2,\ldots,n\}$,
            (3) $\pop(\uparrow_{i}{\alpha})$ which pops stack $i$ if the top
            of stack $i$ is $\alpha \in \Gamma$.
For a transition $t=(s,\op,a,s')$ we write $\src(t)=s,\tgt(t)=s'$ and $\op(t)=\op$. At the moment we ignore the action label $a$ but this will be useful later when we go beyond reachability to model checking. A \emph{configuration} of the \mpda{} is a tuple
$(s,\lambda_1,\lambda_2,\ldots,\lambda_n)$ such that, $s \in \mathcal{S}$ is
the current location and $\lambda_i \in \Gamma^*$
represents the current content of $i^{th}$ stack. 
The semantics of the \mpda{}  is  defined as follows: a run is accepting if it starts from the initial state and reaches a final state with all stacks empty.
The language accepted by a \mpda{} is  defined as the set of 
 words generated by the accepting runs of the \mpda.
Since the reachability problem for \mpda{} is  
Turing complete, we consider under-approximate reachability.   
 A sequence of transitions is called \textbf{complete} if each push
in that sequence has a 
matching pop and vice versa.
A \textbf{well-nested} sequence denoted $ws$ is defined inductively as follows: 
a possibly empty sequence of \nop-transitions is $ws$, and so is the
sequence $t ~ws  ~t'$ where 
$\op(t)=\push({\downarrow}_i{\alpha})$ and
$\op(t')=\pop({\uparrow}_i{\alpha})$ are a  matching pair of push and pop operations of stack $i,\forall i \in \{1\ldots n\}$. Finally the concatenation of two well-nested sequences is a well-nested sequence, i.e., they are closed under concatenation. The set of all well-nested sequences defined by an MPDA is denoted $\ws$. If we visualize this by drawing edges between pushes and their corresponding pops, well-nested sequences have no crossing edges,
as in \includegraphics[scale=0.2]{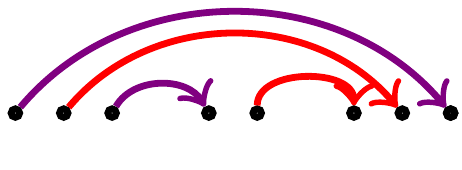} and
\includegraphics[scale=0.2]{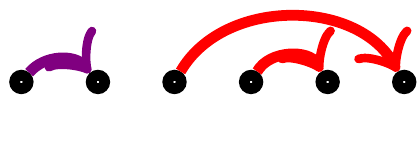}, where we have two stacks,
depicted with red and violet edges.  We emphasize that a well-nested
sequence can have well-nested edges from any stack.  In a sequence
$\sigma$, a push (pop) is called a \textbf{pending} push (pop)  if
its matching pop (push) is not in the same sequence $\sigma$.

\noindent{\bf{Bounded Underapproximations}}. 
As mentioned in the introduction, different bounded under-approximations  have
been considered in the literature to get around the Turing
completeness of \mpda{}. 
  During a computation, a context is a sequence of
    transitions where only one stack or no stack is used. In \emph{ context bounded} computations the  number of contexts are bounded~\cite{qadeer2004kiss}.
    A \emph{round} is a sequence of (possibly empty) contexts for stacks 
      $1,2,\ldots,n$.  \emph{Round bounded} computations restrict the total number of rounds
      allowed~\cite{la2010language,AGK16,AkshayGKS17}.
      \emph{Scope bounded} computations generalize bounded
      context computations. Here, the context changes within any push and
      its corresponding pop is bounded~\cite{la2010language,la2011reachability,la2012scope}.
      A \emph{phase} is a contiguous sequence
      of transitions in a computation,  where we restrict pop to only
      one stack, but there are no restrictions on the pushes~\cite{la2007robust}. 
      A  phase bounded computation is one where the number of phase changes is bounded.

\noindent{\bf Tree-width}. A generic way of looking at them is to consider classes which have a bound on the tree-width~\cite{madhusudan2011tree}.  In fact, the notions of split-width/clique-width/tree-width of communicating finite state machines/timed push down systems has been explored in~\cite{DBLP:journals/corr/abs-1903-03773},~\cite{cyriac-thesis}. The behaviors of the underlying system are then represented as graphs.   It has been shown in these references that if the  
 family of  graphs arising from the behaviours of the underlying
 system (say $S$) have a bounded tree-width,  then the reachability
 problem is decidable  for $S$ via, tree-automata.  However, this does not immediately give rise to an efficient implementation. The tree-automata approach usually gives non-deterministic or bottom-up tree automata, which when implemented in practice (see~\cite{AkshayGKS17}) tend to blow up in size and explore a large and useless space. Hence there is  a need for efficient algorithms, which exist for more specific underapproximations such as context-bounded (leading to fix-point algorithms and their implementations~\cite{la2009analyzing}).

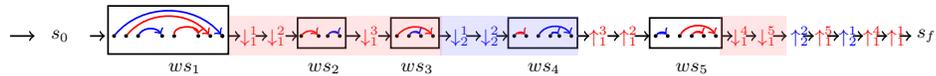
\begin{figure}[b]
\scalebox{.8}{
\begin{tikzpicture}[->,thick]
\tikzset
  {dest/.style={circle,draw,minimum width=0.005mm,inner sep=0mm}}
 \node[state, draw=white,initial,initial text={}] (s) at (-1.8,0) {$s_0$};
 \draw (-1,-0.3) rectangle (1,0.5);
 \draw[->] (-1.3,0) -- (-1.05,0);
 \node[dest] (s0) at (-.9,0) {};
  \node[dest] (s1) at (-.7,0) {};
 \node[dest] (s2) at (-0.5,0) {};
  \node[dest] (s3) at (-0.1,0) {};
 \node[dest] (s6) at (0.5,0) {};
 \node[dest] (s7) at (.1,0) {};
  \node[dest] (s4) at (0.7,0) {};
 \node[dest] (s5) at (0.9,0) {};
 \node[dest] (s5) at (0.9,0) {};
\put(0.1,-16.9){$ws_1$};
 \path(s1) edge[draw=red,bend left=50] node[above] {} node{}(s4);
\path(s2) edge[draw=blue,bend left=60] node[above] {} node{}(s3);
\path(s0) edge[draw=blue,bend left=50] node[above] {} node{}(s5);
 \path(s7) edge[draw=red,bend left=90] node[above] {} node{}(s6);
\draw[->] (1,0) -- (1.2,0);
\node[circle,dashed,minimum width=0.05mm,inner sep=0.5mm] (p1) at (1.35,0) {$\textcolor{red}{\downarrow^{1}_1}$};
\node[circle,dashed,minimum width=0.05mm, inner sep=0.5mm] (p2) at (1.8,0) {$\textcolor{red}{\downarrow^2_1}$};
\draw[->] (1.45,0) -- (1.65,0);

 \draw (2.15,-0.2) rectangle (2.95,0.3);
 \draw[->] (1.95,0) -- (2.15,0);
 \node[dest] (t0) at (2.21,0) {};
  \node[dest] (t1) at (2.5,0) {};
 \node[dest] (t2) at (2.65,0) {};
  \node[dest] (t3) at (2.85,0) {};
  \path(t0) edge[draw=red,bend left=50] node[above] {} node{}(t1);
\path(t2) edge[draw=blue,bend left=60] node[above] {} node{}(t3);
\node[circle,dashed,minimum width=0.05mm, inner sep=0.5mm] (p3) at (3.35,0) {$\textcolor{red}{\downarrow^3_1}$};
 \node[dest,draw=white] (dum2) at (4.5,0) {};
\put(66,-16.9){$ws_2$};
\put(110,-16.9){$ws_3$};
\put(170,-16.9){$ws_4$};

\draw[->] (2.95,0) -- (3.2,0);
\draw (3.7,-0.2) rectangle (4.5,0.3);
 \draw[->] (3.49,0) -- (3.7,0);
 \node[dest] (r0) at (3.8,0) {};
  \node[dest] (r1) at (4,0) {};
 \node[dest] (r2) at (4.2,0) {};
  \node[dest] (r3) at (4.4,0) {};
  \path(r0) edge[draw=red,bend left=50] node[above] {} node{}(r3);
\path(r1) edge[draw=blue,bend left=60] node[above] {} node{}(r2);

\node[circle,dashed,minimum width=0.05mm, inner sep=0.5mm] (p4) at (4.85,0) {$\textcolor{blue}{\downarrow^1_2}$};
\node[circle,dashed,minimum width=0.05mm, inner sep=0.5mm] (p5) at (5.35,0) {$\textcolor{blue}{\downarrow^2_2}$};
\draw[->] (4.5,0) -- (4.7,0);
 \draw[->] (4.95,0) -- (5.2,0);
 
\draw (5.65,-0.2) rectangle (6.8,0.3);
 \draw[->] (5.45,0) -- (5.65,0);
 \node[dest] (q0) at (5.75,0) {};
  \node[dest] (q1) at (5.95,0) {};
 \node[dest] (q2) at (6.15,0) {};
 \node[dest] (q3) at (6.35,0) {};
   \node[dest] (q4) at (6.5,0) {};
 \node[dest] (q5) at (6.7,0) {};
   \path(q0) edge[draw=red,bend left=50] node[above] {} node{}(q1);
\path(q2) edge[draw=blue,bend left=60] node[above] {} node{}(q5);
\path(q3) edge[draw=blue,bend left=60] node[above] {} node{}(q4);

\node[circle,dashed,minimum width=0.05mm, inner sep=0.5mm] (p9) at (7.15,0) {$\textcolor{red}{\uparrow^3_1}$};
\node[circle,dashed,minimum width=0.05mm, inner sep=0.5mm] (p9) at (7.65,0) {$\textcolor{red}{\uparrow^2_1}$};
\put(240,-16.9){$ws_5$};
\node[dest,draw=white] (dum5) at (6.95,0) {};
\node[dest,draw=white] (dum6) at (9.1,0) {};

 \draw[->] (6.8,0) -- (7,0);
\draw[->] (7.3,0) -- (7.5,0);

\draw (8,-0.2) rectangle (9.2,0.3);
 \draw[->] (7.75,0) -- (8,0);
 \node[dest] (d0) at (8.1,0) {};
  \node[dest] (d1) at (8.3,0) {};
 \node[dest] (d2) at (8.5,0) {};
 \node[dest] (d3) at (8.7,0) {};
   \node[dest] (d4) at (8.9,0) {};
 \node[dest] (d5) at (9.1,0) {};
   \path(d0) edge[draw=blue,bend left=50] node[above] {} node{}(d1);
\path(d2) edge[draw=red,bend left=60] node[above] {} node{}(d5);
\path(d3) edge[draw=red,bend left=60] node[above] {} node{}(d4);

\node[circle,dashed,minimum width=0.05mm, inner sep=0.5mm] (p6) at (9.5,0) {$\textcolor{red}{\downarrow^4_1}$};
\node[circle,dashed,minimum width=0.05mm, inner sep=0.5mm] (p7) at (9.95,0) {$\textcolor{red}{\downarrow^5_1}$};

\node[circle,dashed,minimum width=0.05mm, inner sep=0.5mm] (p8) at (10.5,0) {$\textcolor{blue}{\uparrow^2_2}$};
\node[circle,dashed,minimum width=0.05mm, inner sep=0.5mm] (p9) at (10.9,0) {$\textcolor{red}{\uparrow^5_1}$};
\node[circle,dashed,minimum width=0.05mm, inner sep=0.5mm] (p10) at (11.3,0) {$\textcolor{blue}{\uparrow^1_2}$};
\node[circle,dashed,minimum width=0.05mm, inner sep=0.5mm] (p9) at (11.7,0) {$\textcolor{red}{\uparrow^4_1}$};
\node[circle,dashed,minimum width=0.05mm, inner sep=0.5mm] (p9) at (12.1,0) {$\textcolor{red}{\uparrow^1_1}$};
\node[circle,dashed,minimum width=0.05mm, inner sep=0.5mm] (p10) at (12.6,0) {$s_f$};
\node[dest,draw=white] (dum7) at (11.6,-.2) {};
\node[dest,draw=white] (dum8) at (12.21,0.2) {};

\draw[->] (9.2,0)--(9.38,0);
\draw[->] (9.6,0)--(9.78,0);
\draw[->] (10.05,0)--(10.28,0);
 \draw[->] (10.55,0)--(10.73,0);
 \draw[->] (10.95,0)--(11.15,0);
\draw[->] (11.38,0)--(11.55,0);
\draw[->] (11.77,0)--(11.95,0);
\draw[->] (12.15,0)--(12.35,0);

\node[fill=red!50, rectangle, inner sep=0.5, fill opacity=0.2,fit= (p1)(dum2)](H1){};
\node[dest,draw=white] (dum3) at (6.76,0) {};

\node[fill=blue!50, rectangle, inner sep=0.5, fill opacity=0.2,fit= (p4)(dum3)](H1){};

\node[dest,draw=white] (dum4) at (9.8,0) {};

\node[fill=red!50, rectangle, inner sep=0.5, fill opacity=0.2,fit=
(p6)(p7)](H1){};

\end{tikzpicture}
}
\caption{A run $\sigma$ with 2 holes (2 red patches) of the red
  stack and 1 hole (one blue  patch) of the blue stack.
}\label{run:holes1}
\end{figure}

\subsection{ A new class of under-approximations}\label{sec:lbh}
Our goal is to bridge the gap between having practically efficient algorithms and handling more expressive classes of under-approximations for reachability of multi-stack pushdown systems. To do so, we define a bounded approximation which is expressive enough to cover previously defined practically interesting classes (such as context bounded etc), while at the same time allowing efficient decidable reachability tests, as we will see in the next section.

\begin{definition}(Holes).\label{defn:hole}
  Let  $\sigma$ be  complete sequence of transitions, of length $n$ in a \mpda, and let $ws$ be a  well-nested sequence. 
  \begin{itemize}
\item  A \textbf{hole} of stack $i$  is a maximal factor of $\sigma$ of the form $(\downarrow_{i} ws)^{+}$, where $ws\in \ws$.       
The maximality of the hole of stack $i$ follows from the fact that any possible extension ceases to be a hole of stack $i$; that is, the only possible events following a maximal hole of stack $i$ are a push $\downarrow_j$ of some stack $j \neq i$, or a pop of some stack $j \neq i$. In general, whenever we speak about a hole, the underlying stack is clear. 
 \item A push $\downarrow_i$ in a hole (of stack $i$)  is called
   a pending push at (i.e., just before) a position $x\leq n$, if its matching pop occurs in $\sigma$ at a position $z > x$. 
   
\item A hole (of stack $i$) is said to be \textbf{open} at a position $x\leq n$, if there is a pending push $\downarrow_i$ of the hole at $x$.  Let $\#_x(\hole)$ denote the number of open holes at position $x$.  The \textbf{hole bound} of $\sigma$ is defined as $\maxi_{1 \leq x \leq |\sigma|} \#_x(\hole)$.
\item A \emph{hole segment} of stack $i$ is a prefix of a  hole of stack $i$, ending in a $ws$, while an  \emph{atomic hole segment} of stack $i$ is just the segment of the form $\downarrow_i  ws$.

 \end{itemize} 
\end{definition}

As an example, consider the sequence $\sigma$ in
Figure~\ref{run:holes1} of transitions of a \mpda{} having  stacks 1,2 (denoted respectively red and  blue). We use superscripts for each push, pop of each stack to distinguish the $i$th push, $j$th pop and so on of each stack. 
 There are two holes of stack 1 (red stack) denoted by the red  patches, and one hole of stack 2 (blue stack) denoted by the blue patch. The 
subsequence  \textcolor{red}{$\downarrow^1_1\downarrow^2_1ws_2$} of the first hole is not a maximal factor, since it can be extended by \textcolor{red}{$\downarrow^3_1ws_3$} in the run $\sigma$, extending the hole. 
Consider the position in $\sigma$ marked with $\textcolor{blue}{\downarrow^1_2}$. At this position, there is an open hole of the red stack (the first red patch), 
and there is an open hole of the blue stack (the blue patch). Likewise, at the position
  $\textcolor{red}{\uparrow^5_1}$, there are 2 open holes of the red stack  (2 red patches) 
and one open hole of the blue  stack 2 (the blue patch). The
hole bound of $\sigma$ is 3. The green patch consisting of $\textcolor{red}{\uparrow^{3}_1}$, 
$\textcolor{red}{\uparrow^{2}_1}$ and $ws_5$ is a pop-hole of stack 1. Likewise, the pops
$\textcolor{blue}{\uparrow^{2}_2}$, $\textcolor{red}{\uparrow^{5}_1}$, 
$\textcolor{blue}{\uparrow^{1}_2}$ are all pop-holes (of length 1) of stacks 
2,1,2 respectively.

\begin{definition}\textsc{(Hole Bounded Reachability Problem)} Given a
  \mpda{} and  $K \in \mathbb{N}$, the $K$-hole bounded reachability
 problem is the following\@: Does there exist a  $K$-hole
  bounded accepting run of the \mpda{}? 
\end{definition}
\begin{proposition}\label{prop:bh}
   The tree-width of $K$-hole bounded \mpda{} behaviors is at most $(2K+3)$.  
  \end{proposition}
 With this, from~\cite{madhusudan2011tree}\cite{AGK16}\cite{AkshayGKS17}, decidability and complexity follow. Thus,
\begin{corollary}
  \label{cor:main}
The $K$-hole bounded reachability problem for \mpda{} is decidable in
$\mathcal{O}(|\mathcal{M}|^{2K+3})$ where, $\mathcal{M}$ is the size of the underlying \mpda{}.
\end{corollary}
 
 Next, we turn to the expressiveness of this class with respect to the classical underapproximations of \mpda: first, the \textbf{hole} bounded class  strictly subsumes
\textbf{scope} bounded which already subsumes \textbf{context} bounded and \textbf{round} bounded classes. Also \textbf{hole} bounded \mpda{} and \textbf{phase} bounded \mpda{} are orthogonal.
\begin{proposition}\label{prop:first}
  Consider a \mpda{} $M$.  For any $K$, let $L_K$ denote a set of sequences
  accepted by $M$ which have number of rounds or number of contexts or scope
  bounded by $K$.  Then there exists $K'\leq K$ such that $L_K$ is $K'$
  hole bounded.
  Moreover, there exist languages which are $K$ hole bounded for some constant $K$, which are not $K'$ round or context or scope bounded for any $K'$.  Finally, there exists a language which is accepted by phase bounded \mpda{} but
  not accepted by hole bounded \mpda{} and vice versa. \label{prop:one}
\end{proposition}

\begin{proof}
  We first recall that if a language $L$ is 
  $K$-round, or $K$-context bounded, then it is also $K'$-scope bounded for some  $K'\leq K$~\cite{la2011reachability,la2010language}. 
  Hence, we only show that scope bounded systems are
  subsumed by hole bounded systems.

Let $L$ be a  $K$-scope bounded language, and let $M$ be a MPDA accepting $L$.  
Consider a run $\rho$ of $w \in L$ in $M$. Assume that at any point $i$ in the
run $\rho$,  $\#_i(\pushhole s)=k'$, and towards a contradiction, let,
$k' > K$. Consider the leftmost open hole in $\rho$ 
which has a pending push $\downarrow_p$ whose pop $\uparrow_p$ is to the right of $i$.  Since $k'>K$ is the number of open holes at $i$, there are at least $k'>K$ context changes in between $\downarrow_p$  and $\uparrow_p$. 
 This contradicts the $K$-scope bounded assumption, and hence $k' \leq K$.

\noindent To show the strict containment, consider the visibly pushdown 
language~\cite{alur2004visibly} given by 
$L^{bh} = \{a^{\textcolor{red}{n}}b^{\textcolor{red}{n}}(a^{p_1}c^{p_1+1}b^{p'_1}d^{p'_1+1}\cdots 
a^{p_{\textcolor{red}{n}}}c^{p_{\textcolor{red}{n}}+1}b^{p'_{\textcolor{red}{n}}}d^{p'_{\textcolor{red}{n}}+1})\mid n ,p_1,p'_1, \dots, p_n, p'_n \in \mathbb{N}\}$. A possible word $w \in L^{bh}$ is $a^3b^3~a^2c^3b^2d^3~a^2c^3bd^2~ac^2bd^2$ with 
$a,b $  representing push in stack 1,2 
respectively and $c,d$ representing  the corresponding matching pop from stack 1,2.
A run $\rho$ accepting the word $w \in L^{bh}$ will start with a sequence of pushes of stack 1 followed
by another sequence of pushes of stack 2. Note that, the number of the pushes $n$ 
is same in both stacks. Then there is a group $G$ consisting of a 
well-nested sequence of
stack 1 (equal $a$ and $c$) followed by a pop of the stack 1 (an extra $c$),   another well-nested
sequence of stack 2 (equal $b$ and $d$) and a pop of the stack 2 (an extra $d$), 
repeated $n$ times.  From the definition of the
$\pushhole$, the total number of holes required in $G$ is
0. But, we need 1 hole for the sequence of $a$'s and another for the sequence of  $b$'s at the
beginning of the run, which creates at most 2 holes during the
run. Thus, the hole bound for any accepting run $\rho$ 
is 2, and the language $L^{bh}$ is 2-hole bounded.

However, $L^{bh}$ is not $k$-scope bounded for any $k$.  Indeed, for each
$m\geq1$, consider the word $w_m=a^m b^m(ac^{2}bd^{2})^m\in L^{bh}$.  It is easy
to see that $w_m$ is $2m$-scope bounded (the matching $c,d$ 
of each $a,b$ happens $2m$ context switches later) 
but not $k$-scope bounded for $k<2m$. It can be seen that $L^{bh}$ is not $k$-phase bounded either. Finally, $L'=\{(ab)^nc^nd^n \mid n \in \mathbb{N}\}$ with $a,b$ and $c,d$ respectively being push and pop of stack 1,2 is not hole-bounded but 2-phase bounded. 
\qed\end{proof}
\section{A Fix-point Algorithm for Hole Bounded Reachability}\label{sec:algo}
In the previous section, we showed that hole-bounded underapproximations are a decidable subclass for reachability, by showing that this class  has a bounded tree-width. However, as explained in the introduction, this does not immediately give a fix-point based algorithm, which has been shown to be much more efficient for other more restricted sub-classes, e.g., context-bounded. 
In this section, we provide such a fix-point based algorithm for the hole-bounded class and explain its advantages. Later we discuss its versatility by showing extensions and evaluating its performance on a suite of benchmarks.

We describe the algorithm in two steps: first we give a simple fix-point 
based algorithm for the problem of 0-hole or \emph{well-nested reachability}, i.e, reachability by a well-nested sequence without any holes. For the 0-hole case,  our algorithm computes the \emph{reachability relation}, also called the \emph{binary reachability problem~\cite{dang2000binary}}. That is, we accept all pairs of states $(s,s')$ such that there is a well-nested run from $s$ with empty stack to $s'$ with empty stack.
Subsequently, we combine this binary reachability for well-nested sequences with an efficient graph search to obtain an  algorithm for $K$-hole bounded reachability.

\noindent{\bf{Binary well-nested reachability for \mpda}}. 
Note that single stack PDA are a special case, since all runs are indeed well-nested. 
\begin{enumerate}
 \item {\bf Transitive Closure}:
   Let $\mathcal{R}$ be the set of tuples of the form $(s_i,s_j)$ representing  that state $s_j$ is   reachable from state $s_i$ via a \nop{} discrete transition. Such a sequence from $s_i$ to $s_j$ is trivially \emph{well-nested}.  We take the $\mathsf{TransitiveClosure}$ of $\mathcal{R}$ 
   using Floyd-Warshall algorithm~\cite{cormen2009introduction}.
   The resulting set $\mathcal{R}_c$ of tuples
   answers the
   binary reachability for finite state automata (no stacks).
 \item {\bf Push-Pop Closure}: 
   For stack operations, consider a push transition on
   some stack (say stack $i$) of symbol $\gamma$, enabled from a state
   $s_1$, reaching state $s_2$. If there is a matching pop transition
   from a state $s_3$ to $s_4$, which pops the same stack symbol
   $\gamma$ from the stack $i$
   and if we have $(s_2,s_3) \in \mathcal{R}_c$, then we can add the tuple
   $(s_1,s_4)$ to $\mathcal{R}_c$.  
   The function $\mathsf{WellNestedReach}$ (Algorithm~\ref{alg:pushpopc}) repeats this process 
   and the transitive closure described above
   until a fix-point is reached.  Let us denote the resulting set of tuples by
   $\wnr$. Thus,
\end{enumerate}   

   \begin{lemma}
   $(s_1,s_2)\in \wnr $ iff $\exists$ a well-nested run
   in the \mpda{} from $s_1$ to $s_2$.  
   \label{lemm:wellnested}
 \end{lemma}
\begin{algorithm}[t]
  
  {
    \scriptsize
        \SetKwProg{Return}{return}{\string;}{}
        \SetKwProg{Fn}{Function}{\string:}{}
        \SetKwFunction{NumberOfStacks}{NumberOfStacks}
        \SetKwFunction{NumberOfHoles}{NumberOfHoles}
        \SetKwFunction{pushpopclosure}{WellNestedReach}
        \SetKwFunction{concatclosure}{ConcatClosure}
        \SetKwFunction{addpop}{AddPop$_i$}
        \SetKwFunction{addhole}{AddHole$_i$}
        \SetKwFunction{pushclosure}{AtomicHoleSegment$_i$}
        \SetKwFunction{pushclosureC}{HoleSegment$_i$}
        \SetKwFunction{addcomplete}{AddWellNested}
        \SetKwFunction{Transitive}{TransitiveClosure}
        \SetKwFunction{isempty}{IsEmpty}
        \SetKwComment{Comment}{\textbackslash\textbackslash}{}
        \SetKwRepeat{Do}{do}{while}
        \Fn{\isempty{$M=(\mathcal{S},\Delta, s_0,\mathcal{S}_f,$ $n,\Sigma,\Gamma) ,K$}}{
        \KwResult{True or False}
            \wsr{} := \pushpopclosure{$M$};
            \Comment{\textcolor{red}{Solves binary reachability for
                pushdown system}}
            
            \If{some $(s_0,s_1)\in$ \wsr{} with $s_1 \in \mathcal{S}_f$ }{
                \Return{False}{}
              }
              \ForAll{$ i \in [n]$}{
               $AHS_i := \emptyset$; $Set_i := \emptyset$\;
 \ForAll{$(s,{\downarrow}_{i}(\alpha),a,s_1)\in \Delta$ and $ (s_1,s') \in$ \wsr}{
    $AHS_i:= AHS_i \cup \{(i, s,\alpha,s')\}$; $Set_i := Set_i\cup\{(s,s')\}$\;
 }
 $HS_i := \{(i,s,s')\mid (s,s') \in \Transitive{$Set_i$} \}$\;
 }
 $\mu:= [s_0]$; $\mu.\NumberOfHoles := 0$\;
 $\sol_{new} := \{\mu\}$; $\sol:=\emptyset$\;
 \Do{$\sol_{new}\neq \emptyset$}{
    $\sol:= \sol \cup \sol_{new}$\;
    $\sol_{todo} := \sol_{new}$; $\sol_{new} := \emptyset$\;
    \ForAll{ $\mu' \in \sol_{todo}$}{
       \If{$\mu'.\NumberOfHoles < K$}{
          \ForAll{$i \in[n]$}{
             \Comment*[h]{\textcolor{red}{~Add hole for stack i}}\\
             $\sol_{h} := \addhole(\mu',HS_i) \setminus \sol$\;
             $\sol_{new} := \sol_{new} \cup \sol_h$\;
          }
       }

       \If{$\mu'.\NumberOfHoles > 0$}{
          \ForAll{$i \in [n]$}{
             \Comment*[h]{\textcolor{red}{~Add pop for stack i}}\\
             $\sol_p := \addpop(\mu', M, AHS_i, HS_i,\wsr{}) \setminus \sol$\;
             $\sol_{new} := \sol_{new} \cup \sol_p$\;
             \ForAll{$\mu_3 \in \sol_p $}{
                \If{$\mu_3.last \in \mathcal{S}_f$ and $\mu_3.\NumberOfHoles = 0$}
                   {\Return(\Comment*[h]{\textcolor{red}{If reached destination state}}){False}{}}
             }
          }
       }
    }
 }
 \Return{$True$}{}
 } 
 }  
 \caption{Algorithm for Emptiness Checking of hole bounded \mpda{} }\label{alg:isempty}
    
\end{algorithm}

 \noindent{\bf{Beyond well-nested reachability}}. 
 A naive algorithm for $K$-hole bounded reachability for $K>0$ is to start from
 the initial state $s_0$, and do a Breadth First Search (BFS),
 nondeterministically choosing between extending with a well-nested segment, creating
 hole segments (with a pending push) and closing hole segments (using pops).  We
 accept when there are no open hole segments and reach a final state; this gives
 an exponential time algorithm.  Given the exponential dependence on the
 hole-bound $K$ (Corollary~\ref{cor:main}), this exponential blowup is
 unavoidable in the worst case, but we can do much better in practice.  In
 particular, the naive algorithm makes arbitrary non-deterministic choices
 resulting in a blind exploration of the BFS tree.

 In this section, we use the binary well-nested reachability algorithm
 as an efficient subroutine to limit the search in BFS to its
 reachable part (note that this is quite different from DFS as well since we do not just go down one path). The crux is that at any point, we create a new
 hole for stack $i$, \emph{only} when (i) we know that we cannot reach
 the final state without creating this hole and (ii) we know that we can
 close all such holes which have been created.  Checking (i) is easy, since   we just use the
 $\wnr$ relation for this. Checking (ii) blindly would correspond to
 doing a DFS; however, we precompute this information and simply look
 it up, resulting in a constant time operation after the precomputation. 
 \medskip

 \noindent {\bf Precomputing hole information.} Recall that a \emph{hole} of stack $i$ is a maximal sequence of the form $(\downarrow_i ws)^+$, where $ws$ is a well-nested sequence and $\downarrow_i$ represents a push of stack $i$. A \emph{hole segment} of stack $i$ is a prefix of a  hole of stack $i$, ending in a $ws$, while an  \emph{atomic hole segment} of stack $i$ is just the segment of the form $\downarrow_i  ws$.
A \emph{hole-segment} of stack $i$ which starts from state $s$ in the MPDA and ends in state $s'$, can be represented by the triple $(i,s, s')$, that we call a \emph{hole triple}.
We compute the set $HS_i$ of all hole triples  $(i,s,s')$ such that starting at 
$s$, there is a hole segment of stack $i$ which ends at state $s'$, as detailed in lines (5-9) of Algorithm~\ref{alg:isempty}. In doing so, we also compute the set $AHS_i$ of  all atomic hole segments of stack $i$ and store them as tuples of the form $(i,s_p, \alpha, s_q)$ such that $s_p$ and $s_q $ are the \mpda{} states respectively at the left and right end points of an atomic hole  segment of stack $i$, and $\alpha$ is the symbol pushed on stack $i$  ($s_p \xrightarrow{\downarrow_i(\alpha)ws}s_q$).
\medskip

\noindent{\bf A guided BFS exploration.} We start with a list $\mu_0=[s_0]$ consisting of the initial state 
and construct a BFS exploration tree whose nodes are lists of bounded length. 
A list is a 
 sequence of states and hole triples
 representing a 
 $K$-hole bounded run 
 in a concise form. If $H_i$ represents a hole triple for stack $i$, then 
 a list is a sequence of the form $[s, H_i, H_j, H_k, H_i, \dots, H_{\ell}, s']$. 
  The simplest kind of list is a single state $s$.   
 For example, a list with 3 holes 
of stacks $i, j, k$ is
 $\mu = [s_0$,\textcolor{red}{$(i,s,s')$},\textcolor{red}{$(j,r,r')$},\textcolor{red}{$(k,t,t')$},$t'']$. 
  The hole triples (in red) denote open  holes in the list.
The maximum number of open holes in a list is bounded, making the length of the list also bounded.  Let $\last(\mu)$ represent the last element of the list $\mu$. This is always a state. For a node $v$ storing list $\mu$ in the BFS tree, if $v_1,\ldots v_k$ are its children, then the corresponding lists $\mu_1,\ldots \mu_k$ are obtained by extending the list $\mu$ by one of the following operations:
\begin{enumerate}
\item {\bf Extend $\mu$  with a hole}. 
  Assume there is a hole of some stack $i$, which starts at
  $\last(\mu)=s$, and ends at $s'$. 
    If the list at the parent node $v$ is
  $\mu=[\dots, s]$, then for all $(i,s, s') \in HS_i$, we obtain the list
  $\trunc(\mu)\cdot \append[(i,s, s'), s']$
  at the child node (i.e., we remove the last element $s$ of $\mu$, then append to this list the hole triple  $(i,s, s')$, followed by $s'$).   

\item {\bf Extend $\mu$  with a pop}.
 Suppose there is a transition  $t = (s_k, {\uparrow}_{i}(\alpha),a,
 s_k')$  from $\last(\mu)=s_k$, where $\mu$ is of the form {\small $[s_0,\dots, (h,u,v),$ $
  \textcolor{red}{(i,s, s')}$, $(j,t,t')
  \dots,s_k]$}, such that there is no hole triple of stack $i$ after $\textcolor{red}{(i,s, s')}$, we extend the run by matching this pop (with its push). However, to obtain the last pending push of stack $i$ corresponding to this hole, just $HS_i$ information is not enough since we also need to match the stack content. Instead,  we check if we can split the hole \textcolor{red}{$(i,s,s')$} into
   (1) a hole triple $(\textcolor{red}{i,s,s_a)}\in HS_i$, and
 (2)  a tuple $(i,s_a,\alpha,s') \in AHS_i$.
  If both (1) and (2)  are possible, then the pop transition $t$ corresponds to
  the last pending push of the hole \textcolor{red}{$(i,s,s')$}.  
   $t$ indeed  matches the  pending push recorded in the atomic hole $(i,s_a,\alpha,s')$ in
  $\mu$,  enabling the firing of transition  $t$ from
  the state $s_k$, reaching $s'_k$.  In this case, we add the child
  node with the list $\mu'$ obtained from $\mu$ as follows.  We replace
  (i) $s_k$ with $s_k'$, and (ii) $(i,s,s')$ with
  \textcolor{red}{$(i,s, s_a)$}, respectively signifying firing of the transition $t$ and 
  the ``shrinking'' of the hole, by shifting the  end point of the hole segment 
  to the left.  When we obtain the hole triple $(i,s,s)$ (the start and end
  points of the hole segment coincide), we may have
  uncovered the last pending push 
  and thereby ``closed'' the hole segment completely.  
  At this point, we may choose to remove $(i,s,s)$ from the list, obtaining 
  $[s_0,\dots,(h,u,v)$, $(j,t,t')
  \dots,s_k']$. For every such $\mu'= [s_0, \ldots, (h,u,v),
  (i,s,s_a), (j,t,t'), \ldots, s_k']$ and  all  $(s_k',s_m) \in WS$ we also extend 
  $\mu'$ to $\mu''= [s_0, \ldots, (h,u,v),
  (i,s,s_a), (j,t,t'), \ldots, s_m]$. 
  Notice that the size of the list in the child node obtained on a pop, is either the same as the list in the parent, or
  is smaller. The details can be found in Appendix~\ref{app:algo}.
 \end{enumerate}
 
The number of lists is bounded since the number of states and the
length of the lists are bounded.  The BFS exploration tree will thus
terminate. Combining the above steps gives us
Algorithm~\ref{alg:isempty}, whose correctness gives us:
\begin{theorem}
  \label{thm:mainalgo}
  Given a \mpda{} and a positive integer $K$,  Algorithm~\ref{alg:isempty} terminates and answers ``false'' iff there exists a $K$-hole bounded accepting run of the \mpda.
\end{theorem}

\noindent{\bf{Complexity of the Algorithm}}. 
The maximum number of states of the system is
$ |\mathcal{S}| $.  
The time complexity of transitive closure is $\mathcal{O}(|\mathcal{S}|^3)$,  using a 
Floyd-Warshall implementation. 
The time complexity of computing $\mathsf{WellNestedReach}$ 
which uses the transitive closure, is  $\mathcal{O}(|\mathcal{S}|^5) + \mathcal{O}(|\mathcal{S}|^2 \times (|\Delta| \times |\mathcal{S}|))$. To compute $AHS$ for $n$ stacks
the time complexity is $\mathcal{O}(n\times |\Delta| \times
|\mathcal{S}|^2)$ and to compute $HS$ for $n$ stacks the complexity is
$\mathcal{O}(n\times |\mathcal{S}|^2)$.
For multistack systems,  
each list keeps track of   (i) the number of hole segments($\leq K$),  
and (ii) information pertaining to holes (start, end points of holes, and which stack the hole corresponds to). 
    In the worst case, this will be $(2K +
2)$ possible states in a list, as we are keeping the states at the
start and  end points of all the hole segments and a stack
per hole.  
So, there are 
$\leq |\mathcal{S}|^{2K+3} \times n^{K+1}$ lists. 
In the worst case, when there is no $K$-hole bounded run,  we may end up generating all possible lists for a given
bound $K$ on the hole segments. The time complexity is thus bounded above by 
$\mathcal{O}(|\mathcal{S}|^{2K+3}\times n^{K+1}
+ |\mathcal{S}|^5 + |\mathcal{S}|^3 \times |\Delta|)$.  

\noindent{\bf{Beyond Reachability}}.
We can solve the usual safety questions in the (bounded-hole) underapproximate setting, by checking for underapproximate reachability on the product of the given system with the complement of the safe set.  
  Given the way Algorithm~\ref{alg:isempty} is designed,  the fix-point algorithm allows us to go beyond reachability.  In particular, we can solve several (increasingly difficult) variants of the repeated reachability problem, without much modification.

Consider the question : For a given state $s$ and \mpda{}, does there exist a run $\rho$ starting from $s_0$
        which visits $s$ infinitely often? This is decidable if we can decompose $\rho$ into a finite prefix 
        $\rho_1$ and an infinite suffix $\rho_2$ s.t. 
        (1) both $\rho_1, \rho_2$ are well-nested, or (2) $\rho_1$ is $K$-hole bounded complete (all stacks empty), and 
        $\rho_2$ is well-nested, or (3) $\rho_1$ is $K$-hole bounded, and $\rho_2=(\rho_3)^{\omega}$, where 
        $\rho_3$ is $K$-hole bounded.   
   It is easy to see that      
        (1) is solved by two calls to $\mathsf{WellNestedReach}$ and
        choosing non-empty runs. (2)
        is solved by a call to Algorithm~\ref{alg:isempty}, modified so that we reach $s$, and then calling $\mathsf{WellNestedReach}$. 
        Lastly, to solve (3), first modify  Algorithm~\ref{alg:isempty}
        to check reachability to $s$ with possibly non-empty stacks. Then 
        run the modified algorithm twice : first start from $s_0$ and reach $s$; second 
        start from $s$ and reach $s$ again.

\section{Generating a Witness}\label{sec:witness}
We next focus on the question of generating a witness for an accepting run when our algorithm guarantees non-emptiness. This question is important to address from the point of view of applicability: if our goal is to see if bad states are reachable, i.e., non-emptiness corresponds to presence of a bug, the witness run gives the trace of how the bug came about and hence points to what can be done to fix it (e.g., designing a controller). We remark that this question is difficult in general. While there are naive algorithms which can explore for the witness (thus also solving reachability), these do not use fix-point techniques and hence are not efficient. On the other hand, since we use fix-point computations  to speed up our reachability algorithm, finding a witness, i.e., an explicit run witnessing reachability, becomes non-trivial.
Generation of a witness in the case of well-nested runs is simpler than the 
case when the run has holes, and requires us to ``unroll'' pairs  $(s_0,s_f) \in \wsr$
recursively and generate the sequence of transitions responsible for $(s_0,s_f)$, as detailed in Algorithm~\ref{alg:uniquepath}. 

 \noindent {\bf Getting Witnesses from Holes}\label{para:witness}. Now we move on to the more complicated case of behaviours having holes.  
Recall that in BFS exploration we start from the states reachable from $s_0$ by well-nested sequences,  and explore subsequent states obtained either
from  (i) a hole creation, or (ii) a pop
operation on a stack. Proceeding in this manner, if we reach a final
configuration  (say $s_f$), with all holes closed (which implies empty stacks), then we declare non-emptiness. 
To generate a witness, we start from the final state $s_f$ reachable in
the run (a leaf node in the BFS exploration tree) and
\emph{backtrack} on the BFS
exploration tree till we reach the initial state $s_0$.  This results in generating a witness run in the reverse, from the right  to the left. 
 
 \noindent$\bullet$ Assume that the current node of the BFS tree was obtained using a pop operation.  There
  are two possibilities to consider here (see below) depending on whether this pop
  operation closed or shrunk some hole.  Recall that each hole has a left end
  point and a right end point and is of a specific stack $i$, depending on 
  the pending pushes $\downarrow_i$ it has.  So, if the \mpda{} has 
  $k$ stacks, then a list in the exploration tree can have $k$ kinds of
  holes.  The witness algorithm uses $k$ stacks called \emph{witness stacks} to correctly implement the
  backtracking procedure, to deal with $k$ kinds of holes. 
  Witness stacks should not be confused with the stacks of the \mpda{}.  
     
 \noindent$\bullet$    Assume that the current pop operation is closing a hole  \includegraphics[scale=0.2]{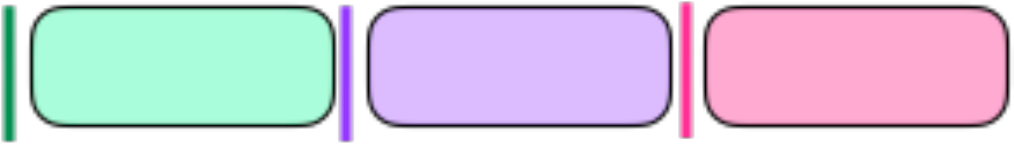} of kind $i$ as in Figure~\ref{spit-hole}. 
    This hole consists of the atomic holes 
    \includegraphics[scale=0.2]{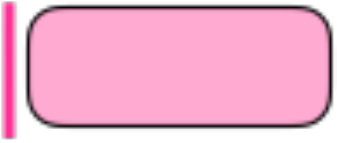},  \includegraphics[scale=0.2]{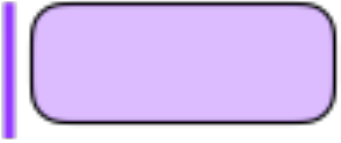} and 
     \includegraphics[scale=0.2]{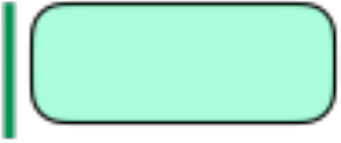}. The atomic hole  \includegraphics[scale=0.2]{ah3.pdf} consists 
     of the push \includegraphics[scale=0.3]{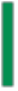} and the well-nested sequence 
     \includegraphics[scale=0.2]{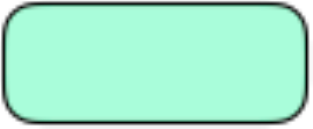} (same for the other two atomic holes).
    Searching among possible push transitions, we identify the matching push
    \includegraphics[scale=0.3]{green-push.pdf}  associated with the current pop, resulting in closing the hole.   On backtracking, this leads  to a
    parent node with the atomic hole \includegraphics[scale=0.2]{ah3.pdf} having as left end point, the push  
    \includegraphics[scale=0.3]{green-push.pdf}, and the right end
    point as the target of the $ws$ \includegraphics[scale=0.2]{green-ws.pdf}. 
     We push onto the witness stack $i$,
    a barrier (a delimiter symbol $\#$) followed by the matching
    push transition \includegraphics[scale=0.3]{green-push.pdf} 
    and then the $ws$,  \includegraphics[scale=0.2]{green-ws.pdf}.  The barrier segregates the contents of the
    witness stack when we have two pop transitions of the same
    stack in the reverse run, closing/shrinking two different holes.  

\begin{wrapfigure}[19]{l}{4.77cm}
  \vspace{-.65cm}
\includegraphics[scale=0.2]{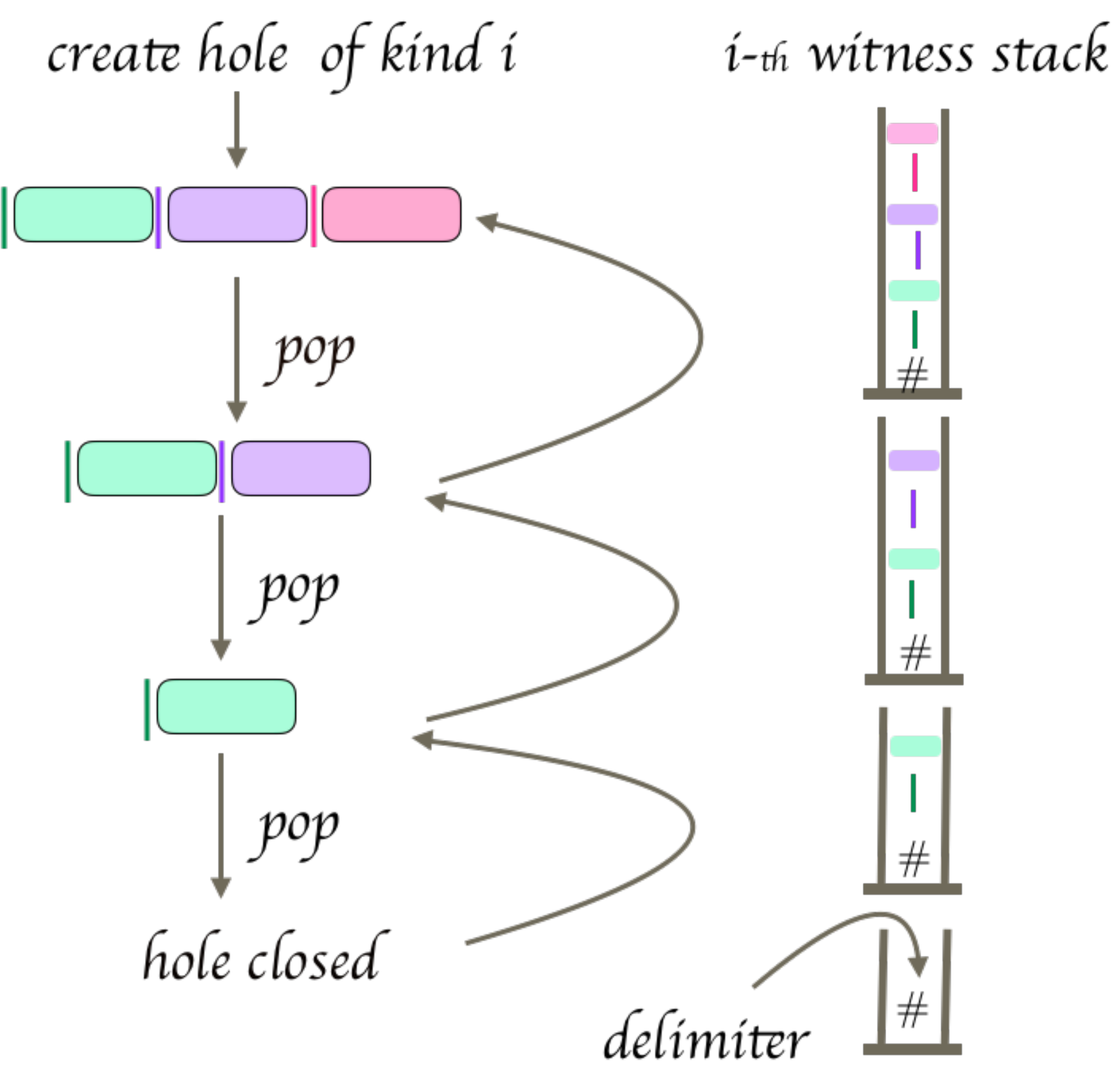}
\caption{Backtracking to spit  out the hole \protect\includegraphics[scale=0.1]{hole.pdf} 
in reverse. The transitions of the atomic hole \protect\includegraphics[scale=0.1]{ah1.pdf} are first written in the reverse order, 
followed by 
those of \protect\includegraphics[scale=0.1]{ah2.pdf} in reverse, and then
of  \protect\includegraphics[scale=0.1]{ah3.pdf} in reverse.
}\label{spit-hole}
\end{wrapfigure}

      \noindent$\bullet$
     Assume that the current pop operation is shrinking a hole of kind
     $i$.
 The list at the present node has this hole, and its parent will have a larger
     hole (see Figure~\ref{spit-hole}, where the parent node of \includegraphics[scale=0.2]{ah3.pdf}
      has \includegraphics[scale=0.2]{ah3.pdf} \includegraphics[scale=0.2]{ah2.pdf}).  
      As in the
     case above, we first identify the matching push transition, and check if it
     agrees with the push in the last atomic hole segment in the parent.  If so, we populate
     the witness stack $i$ with the rightmost atomic hole segment of the parent
     node (see Figure~\ref{spit-hole}, \includegraphics[scale=0.2]{ah2.pdf} is populated in
     the stack).  Each time we find a pop on backtracking the exploration tree,
     we find the rightmost atomic hole segment of the parent node, and keep
     pushing it on the stack, until we reach the node which is obtained as a
     result of a hole creation.  Now we have completely recovered the
     entire hole information by backtracking, and fill the witness stack with the reversed atomic hole segments which constituted this hole.
     Notice that when we finish processing a hole of kind $i$, then the
  witness   stack $i$ has the hole reversed inside it, followed by a
  barrier.  The next  hole of the same kind $i$ will be treated in the
  same manner.

 \noindent$\bullet$
  If the current node of the BFS tree is obtained by creating a hole of kind $i$ in the fix-point algorithm, then we pop the contents of witness stack $i$ till we reach a 
  barrier.  This spits out the atomic hole segments of the hole from the right
  to the left, giving us a sequence of push transitions, and the respective $ws$
  in between.  The transitions constituting the $ws$ are retrieved
  and added.  Notice that popping
  the witness stack $i$ till a barrier spits out the sequence of transitions in the correct reverse order while backtracking.

\section{Adding Time to Multi-pushdown systems}\label{sec:timed}
In this section, we briefly describe how the algorithms described in section~\ref{sec:algo} can be extended to work in the timed setting. Due to lack of space, we focus on some of the significant challenges and advances, leaving the formal details and algorithms to the supplement~\cite{supp}.
A \tmpda{} extends a \mpda{} $\mathcal{S}$ with a set $\mathcal{X}$ of clock variables. Transitions check constraints which are conjunctions/disjunctions  of constraints (called closed guards in the literature) of the form $x \leq c$ or $x \geq c$ for $c \in \mathbb{N}$ and $x$ any clock from   $\mathcal{X}$.  Symbols pushed on stacks ``age'' with time elapse; that os, they store the time elapsed since they were pushed onto the stack.  A pop is successful only when the age of the symbol lies within a certain interval. The acceptance condition is as in the case of \mpda{}. 

The first main challenge in adapting the algorithms in section~\ref{sec:algo} to the timed setting was to take care of all possible time elapses along with the operations defined in Algorithm~\ref{alg:isempty}. The usage of closed guards in \tmpda{} means that it suffices to explore all runs with integral time elapses (for a proof see e.g., Lemma 4.1 in~\cite{AGK16}). Thus configurations are pairs of states with valuations that are vectors of non-negative integers, each of which is bounded by the maximal constant in the system. Now, to check reachability we need to extend all the precomputations (transitive closure, well-nested reachability, as well as atomic and non-atomic hole segments) with the time elapse information. To do this, we use a weighted version of the Floyd-Warshall algorithm by  storing time elapses during precomputations. This allows us to use this precomputed \emph{timed} well-nested reachability information while performing the BFS tree exploration, thus ensuring that any explored state is indeed reachable by a timed run. In doing so, the most challenging part is extending the BFS tree wrt a pop.  Here, we not only have to find a split of a hole into an atomic hole-segment and a hole-segment as in Algorithm~\ref{alg:isempty}, but also need to keep track of possible partitions of time, making the 
algorithm quite challenging.  

\noindent{\textbf{Timed Witness:}}
  As in the untimed case, we generate a witness certifying non-emptiness of 
  \tmpda{}. But, producing a witness for the fix-point computation as discussed earlier requires unrolling. The fix-point computation generates a pre-computed set $\wsrt{}$  of  tuples $((s,\nu),t,(s',\nu'))$, where $s,s'$ are states   $t$ is time elapsed in the  well-nested sequence and $\nu,\nu' \in \mathbb{N}^{|\mathcal{X}|}$ are integral  valuations, i.e., integer values taken by clocks. This set of tuples does not have information about the intermediate transitions and time-elapses. To handle this, using the pre-computed information, we define a lexicographic progress measure which ensures termination of this search.  
      The main idea is as follows: the first progress measure is to check if there a time-elapse $t$ transition possible between $(s,\nu)$ and $(s',\nu')$ and if so, we print this out. If not, $\nu' \neq \nu+t$, and some set of clocks have been reset in the transition(s) from $(s, \nu)$ to $(s', \nu')$. 
  The second progress measure looks at the sequence of transitions from   $(s, \nu)$ to $(s', \nu')$, consisting of reset transitions (at most the number of clocks) that result in $\nu'$ from $\nu$. If neither the first nor the second progress measure apply, then $\nu=\nu'$, and we are left to explore the last progress measure, by exploring at most  $|\mathcal{S}|$ number of transitions from   $(s, \nu)$ to $(s', \nu')$.
 Using this progress measure, we can seamlessly extend the witness generation to the timed setting. The challenges involved therein, can be seen in the full version \cite{supp}. 

  \section{Implementation and Experiments}\label{sec:exp}
We implemented a tool \trim{} ({\bf{B}}ounded {\bf{H}}oles {\bf{I}}n {\bf{M}}PDA) in C++ based on
Algorithm~\ref{alg:isempty}, which takes an \mpda{} and a constant
$K$ as input and returns \emph{True} iff there exists a 
 $K$-hole bounded run from the start state to an accepting state of the
\mpda{}.  In case there is such an accepting run, \trim{} generates one such, with minimal number of holes.  For a given hole bound $K$, \trim{} first tries to produce a witness with 0 holes, and iteratively tries to obtain a witness by  increasing the bound 
on holes till $K$. In most cases, \trim{} found the witness before reaching the bound $K$. Whenever \trim{}'s witness had $K$ holes, it is guaranteed that there are no witnesses with a smaller number of holes. 

To evaluate the performance of \trim{},  we looked at some available benchmarks 
and modeled them as \mpda. We also added timing constraints to some examples such that they can be modeled as \tmpda{}.
Our tests were run on a GNU/Linux system with Intel\textsuperscript{\textregistered} Core\textsuperscript{TM} i7--4770K CPU @ 3.50GHz, and 16GB of RAM. Details of all examples here, as well as 
an additional example of a linux kernel bug
  can be found \cite{supp}.

\noindent{ $\bullet$ \bf Bluetooth Driver~\cite{qadeer2004kiss}}.\label{sbs:btdriver} The  Bluetooth device driver example~\cite{qadeer2004kiss},
 has an arbitrary number of  threads, working with a  shared memory.  
We model this  using a 2-stack pushdown system, where a system state
represents the current valuation of the global variables, and the stacks are used to
maintain the call-return between different functions, as well as to 
keep track of context switches between threads. 
  A known error as pointed out
in~\cite{qadeer2004kiss} is a race condition between two threads where one
thread tries to write to a global variable and the other thread tries to read from it.  \trim{} found this error, with a well-nested witness. A timed extension of this example was also  considered, where, a witness was  obtained again with hole bound 0.  

\noindent{$\bullet$  \bf Bluetooth Driver v2}~\cite{chaki2006verifying,patin2007spade}.\label{sbs:btdriver2}
A modified version of  \hyperref[sbs:btdriver]{Bluetooth driver}
is considered ~\cite{chaki2006verifying,patin2007spade}, where a
counter is maintained to count the number of threads actively using
the driver. We model this with a 
A two stack \mpda{}. With a well-nested witness, \trim{} found the error of  
 interrupted I/O, where  the stopping thread kills the driver while the
other thread is busy with I/O. 

\noindent{$\bullet$  \bf A Multi-threaded Producer Consumer Problem (OLD)}.\label{sbs:producerconsumer}
The Producer consumer problem (see e.g.,~\cite{silberschatz2018operating})
is a classic example of concurrency and synchronization. An
interesting scenario is when there are multiple producers and
consumers. Assume that two ingredients called `A' and `B' are produced in a production line
in batches, where a batch can produce arbitrarily many items, but it is
fixed for a day. Further, assume that 
 (1) two units of `A' and one unit of `B' make an item called `C'; 
(2) the production line starts by producing 
a batch of A's  and then in the rest of the day, it keeps producing B's 
in batches, one after the other. During the day, `C's are churned out 
using `A' and `B' in the proportion mentioned above  and, if we run
out of `A's, we obtain an error; there is no problem if `B' is
exhausted, since a fresh batch producing `B' is commenced. This idea
can be imagined as a real life scenario where item `A' represents an
item which is very expensive to produce but can be produced in large
amount but the item `B' can be produced frequently, but it has to be
consumed quickly, else it becomes useless.
For $m,n,k \in \mathbb{N}$, consider words of the form 
$a^m (b^{k} (c^2d)^{k})^n$ where,   
$a$ represents the production of one unit of `A', $b$ represents the production 
of one unit of `B', $c$ represents consumption of one unit of `A' and
$d$ represents consumption of one unit of `B'. `m' represents the
production capacity of `A' for the day and `$k$' represents production
capacity of `B'(per batch) for the day, `n' represents the number
batches of `B' produced in a day. Unless $m \geq 2nk$, we will obtain an error.    This is  easily modeled 
 using a 2 stack visibly multi pushdown automaton where $a,b$ are push symbols of stack 1, 2 respectively and $c,d$ are pop symbols of stack 1, 2 respectively. 
 Let $L_{m,k,n}$ be the set of words of the above form s.t. $
 2nk < m $.  It can be seen that $L_{m,k,n}$ does not have any
 well-nested word in it. The number of context switches(also, scope bound) in words of
 $L_{m,k,n}$ depends on the parameters $k$ and $n$. However, $L_{m,k,n}$ is
 2 hole-bounded : at any position of the word, 
 the open holes come from the unmatched sequences of $a$ and $b$ seen so far.  
 \trim{} checked for the non-emptiness of $L_{m,k,n}$ with a witness 
 of hole bound 2.

\noindent{$\bullet$  \bf A Multi-threaded Producer Consumer Problem}.\label{sbs:producerconsumer_par}
The producer consumer problem (see e.g.,~\cite{silberschatz2018operating}) is a classic example of concurrency and synchronization. An interesting scenario is when there are multiple producers and consumers. Assume that two ingredients called 'A' and 'B' are produced in a production line in batches (of $M$ and $N$ respectively). These parameters $M$ and $N$ are fixed for each day but may vary across days. There is another consumer machine that (1) consumes one unit of 'A' and one unit of 'B' in that order; (2)  repeats this process until all ingredients are consumed. In between if one of the ingredients runs out, then we non-deterministically produce more batches of the ingredient and then continue. To avoid wastage the factory aims to consume all ingredients produced in a day, hence the problem of interest is to check if all A's and B's produced in a day are consumed. We can model this factory using a two-stack pushdown system, one stack per product, $A, B$, where the sizes of the batches,  $M>0$ and $N>0$ respectively, are parameters. The production and consumption of the `A's and `B's are modeled using push and pop in the respective stack. For a given $M$ and $N$, the language accepted by the system is non-empty iff there is a run where all the produced `A's and `B's are consumed. The language accepted by the two-stack pushdown system is given by $L_{M,N}=((a^M+b^N)^+ (\bar{a}\bar{b})^+)^+$, where $a,b$ represent respectively,  the push on stack 1, 2 and $\bar{a}, \bar{b}$ represent the pop on stack 1, 2 and hence must happen equal number of times.

For any $M, N>0$, any accepting run of the two stack pushdown system cannot be well-nested. Further, in an accepting run, the minimum number of items produced (and hence its length) must be a multiple of $LCM(M,N)$. As the consumption of `A's and `B's happen in an order one by one i.e., in a sequence where consumption of `A' and `B' alternate, the minimum number of context changes (and the scope bound) required in an accepting run depends on $M$ and $N$ (in fact it is $O(2\times LCM(M,N))$. On the other hand, the shortest accepting run is 2-hole bounded: at any position of the word,  the open holes come from the unmatched sequences of $a$ and $b$ seen so far. Thus for any $M,N{>}0$, \trim{} was able to check for non-emptiness of $L_{M,N}$ with a witness of hole bound 2.

\noindent{ $\bullet$  \bf Critical time constraints~\cite{bhave2016perfect}}.\label{sbs:abcd} This is one of the timed examples, where we consider the language $L^{crit} = \{a^yb^zc^yd^z\mid y,z \geq 1\}$ with time constraints between occurrences of symbols.  The first $c$ must appear after 1 time-unit of the last $a$, the first $d$ must appear within 3 time-units  of the last $b$, and the last $b$ must appear within 2 time units from the start, and the last $d$ must appear at 4 time units.
$L^{crit}$ is accepted by a \tmpda{} with two timed stacks.
$L^{crit}$ has no well-nested word, is 4-context bounded, but only 2
hole-bounded.

\noindent{$\bullet$  \bf{Concurrent Insertions in Binary Search Trees}}.\label{sbs:bst}
Concurrent insertions in  binary search trees is a very important problem in database management systems. ~\cite{kung1980concurrent,chaki2006verifying} proposes an algorithm to solve this problem for concurrent implementations. However, incorrect implementation of  locks allows a thread to overwrite others. We modified the algorithm \cite{kung1980concurrent} to capture this bug, and modeled it as \mpda{}. \trim{} found the bug with a witness of hole-bound 2. 

\noindent{$\bullet$  \bf Maze Example}.\label{sbs:maze}
Finally we consider a robot navigating a maze, picking items; an extended (from single to multiple stack) version of the example from~\cite{AkshayGKS17}. 
In the untimed setting, a witness for non-emptiness was obtained with 
hole-bound 0, while in the extension with time, the witness had a hole-bound 2.

  \begin{table}[t]
    \resizebox{\textwidth}{!}{
    \centering
    
     \begin{tabular}{c| c| c| c| c| c| c|c } 
     \toprule
     Name & Locations & Transitions & Stacks &  Holes & Time
                                                        Empty (mili sec) &
                                                                     Time
                                                                          Witness
                                                                           (mili
                                                                          sec)
       & Memory(KB) \\ [0.5ex] 
       \hline
           
     \hyperref[sbs:btdriver]{Bluetooth} & 45&89 &2 &0 &149.3&0.241
 & 6934 \\
       \hline

       \hyperref[sbs:btdriver2]{Bluetooth v2} &47 &134 &2 &0 &92.2 &0.176
       & 5632\\
         \hline
       \hyperref[sbs:producerconsumer]{MultiProdCons(OLD)} & 11 & 18 & 2 & 2 &
                                                               11.1&0.1

                                                                   &
                                                                     1796
       \\

       \hline

        \hyperref[sbs:producerconsumer_par]{MultiProdCons(3,2)} & 7 & 11 & 2 & 2 &
                                                               126.529&0.281

                                                                   &
                                                                     5632
       \\
       \hline
       \hyperref[sbs:producerconsumer_par]{MultiProdCons(24,7)} & 32 & 34 & 2 & 2 &
                                                               1879.33&10.63

                                                                   &
                                                                     21836
       \\
        \hline
       \hyperref[sbs:dmtarget]{dm-target} & 22 & 27 & 2 & 2 &
                                                               26.483&0.279

                                                                   &
                                                                     6624
       \\
       \hline

        \hyperref[sbs:bst]{Binary Search Tree} &29 & 78 & 2 & 2 &
                                                               60.8&5.1
                                                                   &
                                                                     5143
       \\
         \hline
     \hyperref[sbs:abcd]{untimed-$L^{crit}$ } & 6 & 10 & 2 &2&14.9&0.7 &4692\\ 
       \hline
            \hyperref[sbs:maze]{untimed-Maze} &9 & 12 &2 &0 & 8.25&0.07 &5558\\
     \hline
    
       \hyperref[sec:2]{$L^{bh}$} (from Sec.~\ref{sec:lbh}) & 7 & 13 & 2 &2  & 22.2&0.6 &4404 \\

         \hline

       
     \bottomrule
     \end{tabular}
    }
     \caption{Experimental results: Time Empty and Time
       Witness column represents no. of milliseconds needed for emptiness checking
       and to generate witness respectively. }\label{tbl:experiment}
   \end{table}
    
\begin{table}[t]
    \centering
   \resizebox{\textwidth}{!}{
     \begin{tabular}{c| c| c| c| c| c| c| c| c| c|c} 
     \toprule
     Name & Locations & Transitions & Stacks & Clocks &
                                                        c\textsubscript{max}
       &Aged(Y/N)& Holes & Time Empty(mili sec)& Time Witness (mili sec)  & Memory(KB) \\ [0.5ex] 
       \hline
           
     \hyperref[sbs:btdriver]{Bluetooth} & 45 &89&2&0&2&Y&0& 152.8 & 0.119
&5568\\
     \hline

     \hyperref[sbs:abcd]{$L^{crit}$} & 6 & 10 & 2 &2&8 & Y&2 & 9965.2 &3.7 & 203396\\ 
       \hline
            \hyperref[sbs:maze]{Maze} &9 & 12 &2 &2 &5&Y &2 &349.3&0.31&11604\\

\hline

     \bottomrule
     \end{tabular}
    }

     \caption{Experimental results of timed examples. The column \cmax{} is defined as the
       maximum constant in the automaton, and Aged denotes if the
       stack is timed or not}\label{tbl:experimentTimed}
    \end{table}

\noindent{\bf{Results and Discussion}}. 
    The performance of \trim{}  is presented in Table~\ref{tbl:experiment} for untimed examples and in Table~\ref{tbl:experimentTimed} for timed examples. 
 
    Apart from the results in the tables, to check the robustness 
   of \trim{} wrt parameters like the number of locations, transitions, stacks, holes and  clocks (for \tmpda{}), we looked at examples with an empty language, by making accepting states non-accepting in the examples
    considered so far.   This forces \trim{} to explore all possible paths 
   in the BFS tree, generating the lists at all nodes.    The scalability 
   of \trim{} wrt all these parameters are in~\cite{supp}. 
    
\noindent{\bf \trim{} Vs. State of the art}. What makes \trim{} 
stand apart wrt the existing state of the art tools is that 
(i) none of the existing tools handle underapproximations  captured by bounded holes, (ii) none of the existing tools work with multiple stacks in the timed setting (even closed guards!). The state of the art research in underapproximations wrt untimed multistack pushdown systems has produced some robust tools like GetaFix which handles multi-threaded programs with bounded context switches.
    While we have adapted some of the examples from GetaFix, the latest available version of GetaFix has some issues in handling those examples\footnote{we did get in touch
      with one of the authors, who confirmed this.}. Likewise, SPADE, MAGIC and the counter implementation \cite{HagueL12}  are currently not maintained, resulting in a non-comparison of \trim{} and these tools. Most examples handled by \trim{} correspond to non-context bounded, or non-scope bounded, or timed languages which are beyond GetaFix : the 2-hole bounded witness found by \trim{} for the language   $L_{9, 5}$ for the multi producer consumer case cannot be found by GetaFix/MAGIC/SPADE with less than ~90 context switches. In the timed setting, the \hyperref[sbs:maze]{Maze} example which has a 2 hole-bounded witness where the robot visits certain locations equal number of times is beyond \cite{AkshayGKS17}, which can handle only single stack.
      
  \section{Future Work}
     As immediate future work, we are working on \trim{} {\bf{v2}} to be symbolic, inspired from GetaFix. The current avatar of \trim{} showcases the efficiency of fix-point techniques extended to larger bounded underapproximations; indeed going symbolic will make \trim{} much more robust and scalable. This version will  also include a parser to handle boolean programs, allowing us to evaluate larger repositories of available benchmarks. 

\noindent{\emph{Acknowledgements}.} We would like to thank Gennaro
Parlato for the discussions on GetaFix and for providing us benchmarks
and anonymous reviewers for more pointers. 

 \bibliographystyle{splncs04}
 \bibliography{biblio}

\vfill

{\small\medskip\noindent{\bf Open Access} This chapter is licensed under the terms of the Creative Commons\break Attribution 4.0 International License (\url{http://creativecommons.org/licenses/by/4.0/}), which permits use, sharing, adaptation, distribution and reproduction in any medium or format, as long as you give appropriate credit to the original author(s) and the source, provide a link to the Creative Commons license and indicate if changes were made.}

{\small \spaceskip .28em plus .1em minus .1em The images or other third party material in this chapter are included in the chapter's Creative Commons license, unless indicated otherwise in a credit line to the material.~If material is not included in the chapter's Creative Commons license and your intended\break use is not permitted by statutory regulation or exceeds the permitted use, you will need to obtain permission directly from the copyright holder.}

\medskip\noindent\includegraphics{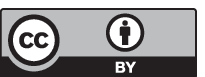}

 \newpage
\appendix
\centerline{\LARGE \textbf{Appendix}}
\section{Details for Section~\ref{sec:2}}\label{sec:proofs}
\subsection{Proposition~\ref{prop:bh}}
\label{app:bh}

We use the notion of  Tree Terms (TTs)~\cite{AkshayGKS17} to compute the tree-width of a
given graph. Where a minimal finite set of colors are used to color the
vertices and then partition the graph in two partitions such that the
cut vertices are colored. The aim of this approach is to decompose a
graph to ``atomic'' tree terms. We cannot use a color more than once
in a partition of graph, unless we \emph{forget} it. This can be
modeled as a game between two player, Adam and Eve. Where, Eve's goal
is to reach atomic terms with minimum finite number of colors, and
Adam's goal is to make Eve's life difficult by choosing a more
demanding partition.

To prove that the a model has bounded tree-width we will try to
capture the runs of the model in terms of graphs (Multiply nested
words~\cite{la2012scope}) and play the game mentioned above.
\subsection*{Tree Width of Hole Bounded Multistack Pushdown Automaton}
 We will capture the behaviour (any run $\rho$) of $K-$hole bounded
 multistack pushdown systems as a graph $G$ where, every node
 represents a transition $t \in \Delta$ and the edge between the nodes
 can be of the following types.

 \begin{itemize}
     \item Linear order $\preccurlyeq$ between the transitions gives the order in which the transitions are fired in the system. We will use
    $\preccurlyeq^+$ to represent transitive closure of $\preccurlyeq$.
       
     \item The other type of edges represent the push pop relation
       between two transitions. Which means, if a transition $t_1$ have
       a push operation in the stack $i$  and transition $t_2$
       has the corresponding pop of the stack $i$, matching the push
       on stack $i$ of transition $t_1$, then we
       have an edge $t_1 \curvearrowright^s t_2$  between them, which
       will represent the push-pop relation.
     \end{itemize}

     To prove that the tree width of the class of graph $G$ is bounded,
     we will use coloring game~\cite{AkshayGKS17} and show that we need bounded number of
     colors to split any graph $g \in G$ to atomic tree terms.

     Eve will start from the right most node of the graph by coloring
     it. The last node of the graph can be any one of the following,
     \begin{itemize}
     \item End point of a well-nested sequence
       \item Pop transition $t_{pp}$ of stack $i$ , such that, the push
         $t_{ps}$ is coming from nearest \emph{hole} of stack $i$.
       \end{itemize}
       \begin{enumerate}
\item If the endpoint colored is the end point of a well-nested sequence
then Eve can remove the well-nested sequence by adding another color to the first point of the well-nested sequence. 

If we look at the well-nested part, using just one more color we can
split it to atomic tree terms~\cite{AkshayGKS17}. 

But the other part still remains a graph of class $G$ so Adam will
choose this partition for Eve to  continue
the coloring game.

\item If the last point of the graph $G$ is a pop point $t_{pp}$ as discussed
  earlier, then the corresponding push $t_{ps}$ can come from a open
  \emph{hole} or a closed \emph{hole}.
  \begin{itemize}
  \item If it is coming from a closed
  \emph{hole} then,  Eve will add color to the corresponding push $t_{ps}$ along
  with the transition $t_{q}$ such that, $t_{ps} \preccurlyeq^+ t_{q}$
  and $t_{ps} \-- t_{q}$ is a well-nested sequence, which forms a atomic
  hole segment $(\uparrow ws)$ where, $\uparrow$ represents the push pop edge $t_{ps}
  \curvearrowright^s t_{pp}$ and $ws$ represents the well-nested sequence
  $t_{ps} \-- t_{q}$. This operation requires  $2$ colors. Please note that, the right end of the
  \emph{hole} which got colored after removal of $t_{ps}\--t_{q}$ is another
  push of the \emph{hole}, because \emph{hole} are defined as a sequence
  $(\uparrow ws)^+$. 
  \item If the push is coming from open \emph{hole} then the push
    transition $t_{ps}$ is already colored from previous operation as
    discussed above, hence Eve will add another color  $t_{q'}$ to mark
    the next well-nested sequence $t_{ps} \-- t_{q'}(ws')$ in the right of $t_{ps}$. Now, Eve can remove the stack edge
    $t_{pp} \curvearrowright t_{ps}$ along with the well-nested
    sequence $ws'$. This operation  widens the \emph{hole}. 
  \end{itemize}
  In both the above operations, the graph has two components
  one with a stack edge $t_{pp} \curvearrowright t_{ps}$ and another one
  with a well-nested sequence. Which require at most $1$ color extra to split into
  atomic tree terms. On the remaining part Eve
  will continue playing the game from right most point.
\end{enumerate}
Here, we claim that at any point of time of the coloring game, there will
be  $2K + 2 $ active
colors for $K \geq 1$ and $ K \in \mathbb{N}$. Every step of the game
splits the graph in two part, and one part always can be split into
atomic tree terms with at most $3$ colors. The remaining part will
require at most $2$
colors for every open \emph{hole} in the left of the right most point
of the graph. As the number of open \emph{hole} is bounded by $K$, so we can not have more
than $K$ open \emph{holes} in the left of any point. So, $2K $ colors to mark the
\emph{holes}. So, total number of colors needed to break any such
graph to atomic tree terms is $2K+4$.

\subsection{Proposition~\ref{prop:one}}
We describe the missing details in proposition \ref{prop:one}. 
\begin{enumerate}
\item \textcolor{red}{$L^{bh}$ cannot be accepted by any K-bounded phase
  \mpda{}}.

Recall that, $L^{bh} =  \{a^nb^n(a^{q_i}c^{q_i+1}b^{q'_j}d^{q'_j+1})^n| n ,q_i,q'_j \in \mathbb{N}~
  \forall i,j \in [n]\}$, and $a,b$ represents push in stack 1,2 respectively
  and $c,d$ represents the corresponding pops from stack 1,2.
For all $m$, consider the word $w_1 =
a^mb^m(a^lc^{l+1}b^{l'}d^{l'+1})^m$. Here, clearly the number of phases is
$K=2m$. Now if $w_1$  is accepted by some phase bounded \mpda{} M
then it must have $2m$ as the bound on the phases which will not be
sufficient to accept $w_2(a^{m+1}b^{m+1}(a^lc^{l+1}b^{l'}d^{l'+1})^{m+1})\in
L^{bh}$. 
\item \textcolor{red}{$L'=\{(ab)^nc^nd^n \mid n \in \mathbb{N}\}$ cannot be accepted by any
  K-hole bounded \mpda{}}. 
  
For any $m \in \mathbb{N}$ assume a word $w_1 =
(ab)^mc^md^m \in L'$, where $a,b$ represents push in stack 1,2
respectively and $c,d$ represents the corresponding pops from stack 1,2. Clearly, this can be accepted by a bounded \emph{hole} multistack
pushdown automata $M$ with bound = $2m$. Now if $L'$ is accepted by
$M$ then it must also accept, $w_2 = (ab)^{m+1}c^{m+1}d^{m+1}$. However, the number of \emph{holes} required to accept $w_2$ 
is  $2(m+1) > 2m$.  This contradicts 
the assumption that  $M$ accepts the language. 
\end{enumerate}

\section{Details for Section~\ref{sec:algo}}\label{app:algo}
In this section, we provide all the subroutines mentioned in Section~\ref{sec:algo} and used in Algorithm~\ref{alg:isempty} for \mpda{}. 
We start by presenting Algorithm~\ref{alg:pushpopc} which computes the well-nested reachability relation, i.e., it computes the set $\wnr$ of all pairs of states $(s,s')$ such that there is a well-nested sequence from $s$ to $s'$.
  \begin{algorithm}[]\label{alg:pushpopc}
\SetKwProg{Return}{return}{\string;}{}
\SetKwProg{Fn}{Function}{\string:}{}
\SetKwFunction{pushpopclosure}{WellNestedReach}
\SetKwFunction{closure}{ConcatClosure}
 \SetKwComment{Comment}{\textbackslash\textbackslash}{}
\SetKwFunction{Transitive}{TransitiveClosure}
 \Fn{\pushpopclosure{$M=(\mathcal{S},\Delta, s_0,\mathcal{S}_f,$ $n,\Sigma,\Gamma)$}}{
 \KwResult{ \wsr$ := \{(s,s')| s' $ is
   reachable from $s$ via a well-nested sequence $\}$}
     $\mathcal{R}_c := \{(s,s) | s \in
     \mathcal{S}\}$\;
       \ForAll{$(s_1,\text{\op{}},a,s_2) \in \Delta $ with $\text{\op{}}$ =
         \nop}{
           $\mathcal{R}_c := \mathcal{R}_c\cup \{(s_1,s_2) \}$;~\Comment*[h]{\textcolor{red}{Transitions with
         \nop{} operation}}\\
         }

    $\mathcal{R}_c$ :=
    \Transitive{$\mathcal{R}_c$};~\Comment*[h]{\textcolor{red}{Using
      Floyd-Warshall Algorithm}}\\
         
 \While{True}{
   \wsr $ :=  \mathcal{R}_c$\;
     \ForAll{$(s,\downarrow_{i}(\alpha),a,s_1)\in \Delta$ }{
         \ForAll{$ (s_1,s_2) \in $\wsr}{
             \ForAll{$(s_2,\uparrow_i(\alpha),a,s')\in \Delta$ }{
              $\mathcal{R}_c :=\mathcal{R}_c  \cup \{(s,s')\}$;~\Comment*[h]{\textcolor{red}{Wrap well-nested sequence
                with matching push-pop}}\\
         }
         }
     }
     $\mathcal{R}_c$:= \Transitive{$\mathcal{R}_c$}\;
     \If{ $\mathcal{R}_c \setminus $\wsr{} $ = \emptyset$}{
         break;~\Comment*[h]{\textcolor{red}{Break when no new
             well-nested sequence added}}\\
     }
 }
 }
 \Return{\wsr{}}{}
 \caption{Well Nested Reachability}
  \end{algorithm}
  The proof of correctness of this algorithm (and thus
  Lemma~\ref{lemm:wellnested}) is easy to see. First, line 5 the set
  $\mathcal{R}_c$ contains all pairs $(s,')$ such that $s'$ is
  reachable from $s$ in the MPDA without using the stack. Then for
  every push transition from a state $s$ we check in lines 8-11
  whether there is an (already computed) well-nested sequence that can
  reach a state $s'$ with a corresponding pop transition and if so we
  add $(s,s')$. We take the transitive closure and repeat this
  process, hence guaranteeing that at fixed point we will have all
  well-nested pairs, i.e., $\wnr$.
  
\begin{algorithm}[]\label{alg:addhole}
\SetKwProg{Return}{return}{\string;}{}
\SetKwProg{Fn}{Function}{\string:}{}
\SetKwFunction{NumberOfHoles}{NumberOfHoles}
\SetKwFunction{pushclosure}{PushClosure$_i$}
 \SetKwComment{Comment}{\textbackslash\textbackslash}{}
\SetKwFunction{pushclosureC}{\ensuremath{PushClosure^*}}
\SetKwFunction{addH}{AddHole$_i$}
\Fn{\addH{$\mu$, $HS_i$}}{
\KwResult{$Set$, a set of lists.}

$Set := \emptyset$\;
\ForAll{$(i,s,s')\in HS_i$ with $s=\last(\mu)$}{
  $\mu' := copy(\mu)$;~\Comment*[h]{\textcolor{red}{Create a copy of
      the list $\mu$}}\\
   $ \trunc(\mu')$;~\Comment*[h]{\textcolor{red}{ $\trunc(\mu)$ is defined as $\remove(\last(\mu))$)}}\\
   $\mu'.\append[(i,s, s'), s']$;~\Comment*[h]{\textcolor{red}{Append
       to the list $\mu'$}}\\
   $\mu'.\NumberOfHoles := \mu.\NumberOfHoles + 1$\;
   $Set := Set \cup \{\mu'\}$\;
}
}
\Return{$Set$}{}
\caption{ AddHole}
\end{algorithm}

\begin{algorithm}[h!]\label{alg:extendwithPop}
\SetKwProg{Return}{return}{\string;}{}
\SetKwProg{Fn}{Function}{\string:}{}
\SetKwFunction{pushcomplete}{PushComplete$_i$}
\SetKwFunction{pushclosure}{PushClosure$_i$}
\SetKwFunction{addpop}{AddPop$_i$}
 \SetKwComment{Comment}{\textbackslash\textbackslash}{}
\Fn{\addpop{$\mu,M=(\mathcal{S},\Delta, s_0,\mathcal{S}_f,$
     $n,\Sigma,\Gamma)$, $AHS_i$, $HS_i$,\wsr}}{
\KwResult{$Set$, a set of lists}
$Set := \emptyset$\;
$\textcolor{red}{(i,s_1,s_3)} :=
lastHole_i(\mu)$;~\Comment*[h]{\textcolor{red}{Get the last open hole
    of stack $i$}}\\
\ForAll{$ (i,s_1,s_2) \in HS_i$, 
    $(s_2,\alpha,s_3)\in AHS_i$,
    $(s,{\uparrow}_{i}(\alpha),s') \in \Delta$, $s = \last(\mu)$ and $(s',s'') \in $ \wsr}{
    $\mu' := copy(\mu)$\;
  $ \trunc(\mu')$\;
  $\mu'.\append(s'')$\;
  \If{$(s_1=s_2)$}{
    $\mu'' := copy(\mu)$\;
     $ \trunc(\mu'')$\;
     $\mu''.\append(s'')$\;
     $\mu''.\mathsf{remove}(\textcolor{red}{(i,s_1,s_3)})$;~\Comment*[h]{\textcolor{red}{Remove
       the hole $(i,s_1,s_2)$ from the list $\mu''$ }}\\
     $\mu''.\NumberOfHoles := \mu.\NumberOfHoles$-1\;
     $Set := Set \cup \{\mu''\}$\;
  }
  $\mu'.\mathsf{replace}(\textcolor{red}{(i, s_1,s_3)},$ by
        \mbox{\qquad}$\textcolor{blue}{(i,s_1,
          s_2)})$;~\Comment*[h]{\textcolor{red}{Replace bigger hole
            $(i,s_1,s_3)$ by new smaller hole $(i,s_1,s_2)$}}\\
  $Set := Set \cup \{\mu'\}$\;

}
}
\Return{$Set$}{}
\caption{Extend with a pop}
\end{algorithm}
\paragraph{Details of Algorithm~\ref{alg:addhole}} For a given list
$\mu$ Algorithm~\ref{alg:addhole} tries to extend the list $\mu$ by adding
a hole of a stack $i$. This is achieved by checking the last state $s_{last}$
the list $\mu$ and finding all possible hole  in $HS_i$ that start
with $s_{last}$ and appending the hole followed by a suitable
well-nested sequence to $\mu$.
\paragraph{Details of Algorithm~\ref{alg:extendwithPop}}
For a given list $\mu$ this algorithm tries to extend $\mu$ with a pop
operation. The algorithm starts with extracting the last hole($H_i$) of stack
$i$. Due to the well-nested property, the pop (which is not
part of a well-nested sequence) must be matched with the first pending push in the last hole of
stack $i$ in $\mu$.
Then the algorithm checks for all atomic hole-segments $AHS_i$
and hole-segments $HS_i$ s
of the stack $i$, such that, the hole $H_i$ can be partitioned in
$HS_i$ and $AHS_i$. Then the push in $AHS_i$ is matched with the
matched pop operation and the hole is now shrunk into $HS_i$. So, the
algorithm replaces $H_i$ with $HS_i$. If the $H_i$ is same as some
$AHS_i$ then, the hole can be closed and hence it removes the hole
from the list. In this case it also reduces the count of the number of
holes in the list. Note that without the pre-computation of $AHS_i$
and $HS_i$ this part of the algorithm is fairly difficult. Using the
pre-computation allow us to use simple table look ups when the states are
known, this takes only constant time.

\section{Details for Section~\ref{sec:witness}}\label{app:witness}
The algorithm for witness generation, as discussed in the  main part 
of the paper, does a backtracking on the BFS tree. When we encounter a 
node in the BFS tree extending the list with a pop, creating a hole, we use the last state in the list, the transition information from the node, and the witness stack 
for backtracking.  During the backtracking we also need to know the
sequence of transitions responsible for the well-nested sequences,
which can be generated using the Algorithm~\ref{alg:uniquepath}. The
backtracking Algorithm~\ref{alg:holeuniquepath} is discussed in the
following example.

\begin{algorithm}[!h]\label{alg:uniquepath}
 
\SetKwProg{Return}{return}{\string;}{}
\SetKwProg{Fn}{Function}{\string:}{}
\SetKwFunction{Witnss}{Witness}
\SetKwFunction{useless}{Witness}
\SetKwFunction{timeLapse}{TimeLapse}
\SetKwFunction{useful}{UseFulPath}
\SetKwFunction{firable}{Firable}
\SetKwFunction{usefulTr}{UsefulTransition}

\Fn{\useless{$s_1,s_2,M=(\mathcal{S},\Delta, s_0,\mathcal{S}_f,$ $n,\Sigma,\Gamma)$,\wsr{}}}{
\KwResult{ A sequence of transitions for a run resulting the
  well-nested sequence \wsr{}}
\If{$s_1 == s_2$}{
  \Return{$\epsilon$}{}
}
\If{$\exists  t= (s_1,\nop{},a,s_2)\in \Delta$}{
  \Return{$t$}{}
  }
  \ForAll{$s',s'' \in \mathcal{S}$}{
  
   \If{     $((s_1 \neq s') \vee (s'' \neq s_2)) \wedge (s',s'')\in$ \wsr
     $\wedge \exists t=(s_1,\downarrow_i(\alpha),a,s') \in \Delta \wedge$
     $\exists t_2 =(s'',\uparrow_i(\alpha),a',s_2) \in \Delta$ }{
     path=\useless{$s',s'',M$,\wsr}\;
     \Return{$t.path.t_2$}{}
     }
   }
   \ForAll{$s \in \mathcal{S}$}{
     \If{$(s \neq s_1 \vee s \neq s_2)\wedge (s,s_1) \in \wsr \wedge
       (s,s_2) \in \wsr $}{
       path1=\useless{$s_1,s,M$,\wsr}\;
       path2 = \useless{$s,s_2,M$,\wsr}\;
       \Return{path1.path2}{}
     
     }

}
}

\caption{Well-nested witness generation for MPDA}

\end{algorithm}

\begin{algorithm}[!h]\label{alg:holeuniquepath}
 
\SetKwProg{Return}{return}{\string;}{}
\SetKwProg{Fn}{Function}{\string:}{}
\SetKwFunction{Witnss}{Witness}
\SetKwFunction{useless}{Witness}
\SetKwFunction{timeLapse}{TimeLapse}
\SetKwFunction{useful}{UseFulPath}
\SetKwFunction{firable}{Firable}
\SetKwFunction{usefulTr}{UsefulTransition}
\SetKwFunction{HWitnss}{HoleWitness}

 \SetKwComment{Comment}{\textbackslash\textbackslash}{}
\Fn{\HWitnss{$\mu,M=(\mathcal{S},\Delta, s_0,\mathcal{S}_f,$ $n,\Sigma,\Gamma)$,\wsr,$AHS_i$,$HS_i$}}{
\KwResult{ A sequence of transitions for an accepting run}
\textbf{global} WitnessStacks = \{$St_i\mid i \in [n]$\};~\Comment*[h]{\textcolor{red}{Witness
  stacks for every stack i}}\\
$\mu_p = Parent(\mu)$;~\Comment*[h]{\textcolor{red}{Parent function
    returns the parent node of $\mu$ in the BFS exploration tree}}\\
$op_{\mu} = ParentOp(\mu)$;~\Comment*[h]{\textcolor{red}{ParentOp function
    returns the operation that extends  $Parent(\mu)$ to $\mu$ in the BFS exploration tree}}\\
\If{$op_{\mu} == ExtendByPop_i(\uparrow_i{\alpha}.wr_{pop}) \wedge wr_{pop} \in \wsr{}$}{
  $(i,s_1,s_2) = lastHole_i(\mu_p)$\;

  \If{$(s_i,\alpha,s_2) \in AHS_i \wedge (s_1,\alpha,s_2) = \downarrow_i(\alpha).wr_{push} \wedge wr_{push} \in \wsr$}{
    $push(St_i,\#)$\;
    $list = \useless(wr_{push})$\;
   $\forall t \in list, push(St_i,t)$\;
   $push(St_i,\downarrow_i(\alpha))$\;
   $list_{pop} = \useless(wr_{pop})$\;
   \Return{\HWitnss{$\mu_{p}$}.$\uparrow_{i}(\alpha).list_{pop}$}{}
    }
  \ElseIf{$(s_i,\alpha,s_2) \notin AHS_i \wedge (i,s_i,s_2) =  (s_i,\alpha,s_3).
    (i,s_3,s_2)\wedge (s_1,\alpha,s_3) \in AHS_i \wedge (i,s_3,s_2) \in HS_i \wedge (s_1,\alpha,s_3) = \downarrow_i(\alpha).wr_{push} \wedge wr_{push} \in \wsr$ }{
     $list = \useless(wr_{push})$\;
   $\forall t \in list, push(St_i,t)$\;
   $push(St_i,\downarrow_i(\alpha))$\;
   $list_{pop} = \useless(wr_{pop})$\;
   \Return{\HWitnss{$\mu_{p}$}.$\uparrow_{i}(\alpha).list_{pop}$}{}
  }
}
\If{$op_{\mu} == ExtendByHole_i$}{
  $list = \epsilon$\;
  \While{ $pop(St_i) \neq \#$}{
$    list = list.pop(St_i)$\;
    }
  \Return{\HWitnss{$\mu_{p}$}.$list$}{}
}

}

\caption{Non-well-nested witness generation for MPDA}

\end{algorithm}

\subsection*{An Illustrating Example for Witness Generation}\label{app:witnesseg}

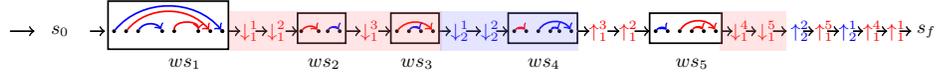
\begin{figure}[h]
\scalebox{.8}{
\begin{tikzpicture}[->,thick]
\tikzset
  {dest/.style={circle,draw,minimum width=0.005mm,inner sep=0mm}}
 \node[state, draw=white,initial,initial text={}] (s) at (-1.8,0) {$s_0$};
 \draw (-1,-0.3) rectangle (1,0.5);
 \draw[->] (-1.3,0) -- (-1.05,0);
 \node[dest] (s0) at (-.9,0) {};
  \node[dest] (s1) at (-.7,0) {};
 \node[dest] (s2) at (-0.5,0) {};
  \node[dest] (s3) at (-0.1,0) {};
 \node[dest] (s6) at (0.5,0) {};
 \node[dest] (s7) at (.1,0) {};
  \node[dest] (s4) at (0.7,0) {};
 \node[dest] (s5) at (0.9,0) {};
 \node[dest] (s5) at (0.9,0) {};
\put(0.1,-16.9){$ws_1$};
 \path(s1) edge[draw=red,bend left=50] node[above] {} node{}(s4);
\path(s2) edge[draw=blue,bend left=60] node[above] {} node{}(s3);
\path(s0) edge[draw=blue,bend left=50] node[above] {} node{}(s5);
 \path(s7) edge[draw=red,bend left=90] node[above] {} node{}(s6);
\draw[->] (1,0) -- (1.2,0);
\node[circle,dashed,minimum width=0.05mm,inner sep=0.5mm] (p1) at (1.35,0) {$\textcolor{red}{\downarrow^{1}_1}$};
\node[circle,dashed,minimum width=0.05mm, inner sep=0.5mm] (p2) at (1.8,0) {$\textcolor{red}{\downarrow^2_1}$};
\draw[->] (1.45,0) -- (1.65,0);

 \draw (2.15,-0.2) rectangle (2.95,0.3);
 \draw[->] (1.95,0) -- (2.15,0);
 \node[dest] (t0) at (2.21,0) {};
  \node[dest] (t1) at (2.5,0) {};
 \node[dest] (t2) at (2.65,0) {};
  \node[dest] (t3) at (2.85,0) {};
  \path(t0) edge[draw=red,bend left=50] node[above] {} node{}(t1);
\path(t2) edge[draw=blue,bend left=60] node[above] {} node{}(t3);
\node[circle,dashed,minimum width=0.05mm, inner sep=0.5mm] (p3) at (3.35,0) {$\textcolor{red}{\downarrow^3_1}$};
 \node[dest,draw=white] (dum2) at (4.5,0) {};
\put(66,-16.9){$ws_2$};
\put(110,-16.9){$ws_3$};
\put(170,-16.9){$ws_4$};
\put(240,-16.9){$ws_5$};

\draw[->] (2.95,0) -- (3.2,0);
\draw (3.7,-0.2) rectangle (4.5,0.3);
 \draw[->] (3.49,0) -- (3.7,0);
 \node[dest] (r0) at (3.8,0) {};
  \node[dest] (r1) at (4,0) {};
 \node[dest] (r2) at (4.2,0) {};
  \node[dest] (r3) at (4.4,0) {};
  \path(r0) edge[draw=red,bend left=50] node[above] {} node{}(r3);
\path(r1) edge[draw=blue,bend left=60] node[above] {} node{}(r2);

\node[circle,dashed,minimum width=0.05mm, inner sep=0.5mm] (p4) at (4.85,0) {$\textcolor{blue}{\downarrow^1_2}$};
\node[circle,dashed,minimum width=0.05mm, inner sep=0.5mm] (p5) at (5.35,0) {$\textcolor{blue}{\downarrow^2_2}$};
\draw[->] (4.5,0) -- (4.7,0);
 \draw[->] (4.95,0) -- (5.2,0);
 
\draw (5.65,-0.2) rectangle (6.8,0.3);
 \draw[->] (5.45,0) -- (5.65,0);
 \node[dest] (q0) at (5.75,0) {};
  \node[dest] (q1) at (5.95,0) {};
 \node[dest] (q2) at (6.15,0) {};
 \node[dest] (q3) at (6.35,0) {};
   \node[dest] (q4) at (6.5,0) {};
 \node[dest] (q5) at (6.7,0) {};
   \path(q0) edge[draw=red,bend left=50] node[above] {} node{}(q1);
\path(q2) edge[draw=blue,bend left=60] node[above] {} node{}(q5);
\path(q3) edge[draw=blue,bend left=60] node[above] {} node{}(q4);

\node[circle,dashed,minimum width=0.05mm, inner sep=0.5mm] (p9) at (7.15,0) {$\textcolor{red}{\uparrow^3_1}$};
\node[circle,dashed,minimum width=0.05mm, inner sep=0.5mm] (p9) at (7.65,0) {$\textcolor{red}{\uparrow^2_1}$};
 \draw[->] (6.8,0) -- (7,0);
\draw[->] (7.3,0) -- (7.5,0);

\draw (8,-0.2) rectangle (9.2,0.3);
 \draw[->] (7.75,0) -- (8,0);
 \node[dest] (d0) at (8.1,0) {};
  \node[dest] (d1) at (8.3,0) {};
 \node[dest] (d2) at (8.5,0) {};
 \node[dest] (d3) at (8.7,0) {};
   \node[dest] (d4) at (8.9,0) {};
 \node[dest] (d5) at (9.1,0) {};
   \path(d0) edge[draw=blue,bend left=50] node[above] {} node{}(d1);
\path(d2) edge[draw=red,bend left=60] node[above] {} node{}(d5);
\path(d3) edge[draw=red,bend left=60] node[above] {} node{}(d4);

\node[circle,dashed,minimum width=0.05mm, inner sep=0.5mm] (p6) at (9.5,0) {$\textcolor{red}{\downarrow^4_1}$};
\node[circle,dashed,minimum width=0.05mm, inner sep=0.5mm] (p7) at (9.95,0) {$\textcolor{red}{\downarrow^5_1}$};

\node[circle,dashed,minimum width=0.05mm, inner sep=0.5mm] (p8) at (10.5,0) {$\textcolor{blue}{\uparrow^2_2}$};
\node[circle,dashed,minimum width=0.05mm, inner sep=0.5mm] (p9) at (10.9,0) {$\textcolor{red}{\uparrow^5_1}$};
\node[circle,dashed,minimum width=0.05mm, inner sep=0.5mm] (p10) at (11.3,0) {$\textcolor{blue}{\uparrow^1_2}$};
\node[circle,dashed,minimum width=0.05mm, inner sep=0.5mm] (p9) at (11.7,0) {$\textcolor{red}{\uparrow^4_1}$};
\node[circle,dashed,minimum width=0.05mm, inner sep=0.5mm] (p9) at (12.1,0) {$\textcolor{red}{\uparrow^1_1}$};
\node[circle,dashed,minimum width=0.05mm, inner sep=0.5mm] (p10) at (12.6,0) {$s_f$};

\draw[->] (9.2,0)--(9.38,0);
\draw[->] (9.6,0)--(9.78,0);
\draw[->] (10.05,0)--(10.28,0);
 \draw[->] (10.55,0)--(10.73,0);
 \draw[->] (10.95,0)--(11.15,0);
\draw[->] (11.38,0)--(11.55,0);
\draw[->] (11.77,0)--(11.95,0);
\draw[->] (12.15,0)--(12.35,0);

\node[fill=red!50, rectangle, inner sep=0.5, fill opacity=0.2,fit= (p1)(dum2)](H1){};
\node[dest,draw=white] (dum3) at (6.76,0) {};

\node[fill=blue!50, rectangle, inner sep=0.5, fill opacity=0.2,fit= (p4)(dum3)](H1){};

\node[dest,draw=white] (dum4) at (9.8,0) {};

\node[fill=red!50, rectangle, inner sep=0.5, fill opacity=0.2,fit=
(p6)(p7)](H1){};

\end{tikzpicture}
}
 \caption{A run with 3 holes. The blue hole corresponds to the blue stack  and the red holes to the red stack. A final state is reached from $\textcolor{red}{\uparrow^1_1}$ on a discrete transition.}\label{run:holes}
\end{figure}

We illustrate the multistack case on an example.
Note that in figures illustrating examples, 
 we use colored uparrows and downarrows with subscript for stacks, and a superscipt $i$ representing the 
 $i$th push or pop of the relevant colored stack.

Assume that the path we obtain 
on back tracking is the reverse of Figure~\ref{run:holes}.  Holes arising from pending
pushes of stack 1 are red holes, and those from stack 2 are blue holes in the
figure.  We have two red holes\@: the first red hole has a left end point
$\textcolor{red}{\downarrow^1_1}$, and right end point $ws_3$.  The second red
hole has a left end point $\textcolor{red}{\downarrow^4_1}$, and right end point
$\textcolor{red}{\downarrow^5_1}$.  The blue hole has left end point
$\textcolor{blue}{\downarrow^1_2}$ and right end point $ws_4$.
\begin{enumerate}
  \item 	
  From the final configuration $s_f$, on backtracking, we obtain the pop
  operation ($\textcolor{red}{\uparrow^1_1}$).  By the fixed-point algorithm, this
  operation closes the first red hole, matching the first pending push
  $\textcolor{red}{\downarrow^1_1}$.  In the BFS exploration tree, the parent node
  has the red atomic hole consisting of just the
  $\textcolor{red}{\downarrow^1_1}$.  Notice also that, in the parent node, this
  is the only red hole, since the second red hole in Figure~\ref{run:holes} is
  closed, and hence does not exist in the parent node.  We use two witness
  stacks, a red witness stack and a blue witness stack to track the information with respect to
  the red and blue holes.  On encountering a pop transition closing a red
  hole, we populate the red witness stack with (i) a barrier signifying closure of a
  red hole, and (ii) the matching push transition
  $\textcolor{red}{\downarrow^1_1}$.
     
  \item
  Continuing with the backtracking, we obtain the pop operation
  $\textcolor{red}{\uparrow^4_1}$, which, by the fixed-point algorithm, closes the
  second red hole.  
  In the parent node, we have the atomic red hole consisting of just the
  $\textcolor{red}{\downarrow^4_1}$.  The red witness stack contains from bottom to top,
  $\# \textcolor{red}{\downarrow^1_1}$.  Since we encounter a closure of a red
  hole again, we push to the red witness stack, $\# \textcolor{red}{\downarrow^4_1}$.
  This gives the content of the red witness stack as $\# \textcolor{red}{\downarrow^1_1}
  \# \textcolor{red}{\downarrow^4_1}$ from bottom to top.  The next pop transition
  $\textcolor{blue}{\uparrow^1_2}$ is processed the same way, populating the
  blue witness stack with $\# \textcolor{blue}{\downarrow^1_2}$.
      
  \item
  Continuing with backtracking, we have the pop transition
  $\textcolor{red}{\uparrow^5_1}$.  Since this is not closing the second red
  hole, but only shrinking it, we push $\textcolor{red}{\downarrow^5_1}$ on top of
  the red witness stack (no barrier inserted).  This gives the content
  of the red witness
  stack as $\# \textcolor{red}{\downarrow^1_1} \#
  \textcolor{red}{\downarrow^4_1}\textcolor{red}{\downarrow^5_1}$.

  \item
  We next have the pop transition $\textcolor{blue}{\uparrow^2_2}$, which by the
  fixed-point algorithm, shrinks the blue hole.  The parent node has
  the blue hole
  with left end point $\textcolor{blue}{\downarrow^1_2}$, and ends with the
  atomic hole segment $\textcolor{blue}{\downarrow^2_2} ws_4$.  We push onto the
  blue witness stack, this atomic hole obtaining the witness stack contents (bottom to top) $\#
  \textcolor{blue}{\downarrow^1_2}\textcolor{blue}{\downarrow^2_2} ws_4$.  
  
  \item
  In the next step of backtracking, we are at a parent node using the create
  hole operation (creation of the second red hole).  We pop the contents of the
  red witness stack till we hit a $\#$, giving us the transitions
  $\textcolor{red}{\downarrow^5_1}\textcolor{red}{\downarrow^4_1}$ in the reverse
  order. 
  \item
  Next, on backtracking, we encounter the pop operation
  $\textcolor{red}{\uparrow^2_1}$ along with a well-nested sequence
  $ws^5$.  We retrieve from this
  information, $ws^5$, and using the Algorithm~\ref{alg:uniquepath}, obtain
  the sequence of transitions constituting $ws^5$.  The parent node has a hole segment with left
  end point $\textcolor{red}{\downarrow^1_1}$, followed by the atomic hole segment
  $\textcolor{red}{\downarrow^2_1} ws_2$.  We find the matching push transition as
  $\textcolor{red}{\downarrow^2_1}$, and push the last atomic hole segment to the
  red witness stack, obtaining witness stack contents $\# \textcolor{red}{\downarrow^1_1}
  \textcolor{red}{\downarrow^2_1} ws_2$.  The next pop operation
  $\textcolor{red}{\uparrow^3_1}$ leads us to the next parent having a hole with
  left end point $\textcolor{red}{\downarrow^1_1}$, and ending with the atomic
  hole $\textcolor{red}{\downarrow^3_1} ws_3$.  We push this to the red witness stack
  obtaining $\# \textcolor{red}{\downarrow^1_1} \textcolor{red}{\downarrow^2_1} ws_2
  \textcolor{red}{\downarrow^3_1} ws_3$ as the stack contents from bottom to top.
  
  \item
  Next, the backtracking leads us to the parent creating the blue hole.  We pop
  the blue witness stack retrieving $ws_4$ followed by the push transitions
  $\textcolor{blue}{\downarrow^2_2}$ and $\textcolor{blue}{\downarrow^1_2}$.
  The transitions of $ws_4$ are obtained from Algorithm~\ref{alg:uniquepath}.
      
  \item
  Continuing with the backtracking, we arrive at the transition which creates
  the first red hole.  At this time, we pop the red witness stack until we hit a
  barrier.  We obtain $ws_3$, and then we retrieve the transition
  $\textcolor{red}{\downarrow^3_1}$, followed by $ws_2$, and the push transitions
  $\textcolor{red}{\downarrow^2_1}$ and $\textcolor{red}{\downarrow^1_1}$.
  Transitions of $ws_3, ws_2$ are retrieved using Algorithm~\ref{alg:uniquepath}.
  
  \item
  Further backtracking leads us to the parent obtained by extending with the
  well-nested sequence $ws_1$.  We retrieve the transitions in $ws_1$ using
  Algorithm~\ref{alg:uniquepath}.  The last backtracking lands us at the root
  $[s_0]$ and we are done.
\end{enumerate}

\section{Details for Section~\ref{sec:timed}}\label{app:tmpda}
This part of the appendix is devoted to extending our algorithms 
for reachability and witness generation.  We start by defining 
timed multistack push down automata.  Then, Appendix \ref{app:timedreach}  
details the (binary) reachability and algorithms therein, whereas Appendix  
\ref{app:timedwit} describes the generation of a witness for \tmpda{}. 
 
\subsection*{Timed Multi-stack Pushdown Automata (\tmpda{})}
For $N\in \mathbb{N}$, we  denote the set of numbers $\{1,2,3 \cdots N\}$ as
$[N]$. $\mathcal{I}$ denotes the set of closed
  intervals $\{ I | I \subseteq \mathbb{R}_+\}$, such that the end points of the intervals belong to $\mathbb{N}$. $\mathcal{I}$ also contains a special interval $[0,0]$. 
We start by defining the model of timed multi-pushdown automata.
\begin{definition}
    A Timed Multi-pushdown automaton (\tmpda{}~\cite{AGK16})  is a tuple 
    $M=(\mathcal{S},\Delta, s_0,\mathcal{S}_f,$ $ \mathcal{X},n,\Sigma,\Gamma)$ 
        where, 
        $\mathcal{S} $ is a finite non-empty set of locations,
        $\Delta$ is a finite set of transitions,
        $s_0\in \mathcal{S}$ is the initial location,
        $\mathcal{S}_f \subseteq \mathcal{S}$ is a set of final locations,
        $\mathcal{X} $ is a finite set of real valued variables known as clocks,
        $n$ is the number of (timed) stacks,
        $\Sigma $ is a finite input  alphabet,
        and $\Gamma$ is a finite stack alphabet which contains $\bot$. 
        A transition $t \in \Delta$ can be represented as a tuple
        $(s,\varphi,\text{\op{}},a,R,s')$, where, $s,s' \in \mathcal{S}$ are
        respectively, the source and destination locations of the transition $t$,
        $\varphi$ is a finite conjunction of closed
         guards of the
         form $x \in I$ represented as, $(x\in I' \wedge y \in I'' \dots)  $ for
        $x,y \in \mathcal{X}$  and $I',I'' \in \mathcal{I}$, $R \subseteq \mathcal{X}$
        is the set of clocks that are reset, $a \in \Sigma$ is the
        label of the transition, and \op{} is one of the following
        stack operations
            (1) \nop, or no stack operation,
            (2) $\push(\downarrow_i{\alpha})$ which pushes $\alpha \in \Gamma$
            onto stack $i \in [n]$,
            (3) $\pop(\uparrow^I_{i}{\alpha})$ which pops stack $i$ if the top
            of stack $i$ is $\alpha \in \Gamma$ and the time elapsed from the push
            is in the interval $I \in \mathcal{I}$.
      \end{definition}

A \emph{configuration} of \tmpda{} is a tuple
$(s,\nu,\lambda_1,\lambda_2,\ldots,\lambda_n)$ such that, $s \in \mathcal{S}$ is
the current location, $\nu\colon \mathcal{X} \rightarrow \mathbb{R}$ is the current clock
valuation and $\lambda_i \in (\Gamma\times \mathbb{R})^*$
represents the current content of $i^{th}$ stack as well as the \emph{age} of
each symbol, i.e., the time elapsed since it was pushed on the stack.
A pair $(s, \nu)$, where $s$ is a location and $\nu$ is a clock valuation is called a \emph{state}.  

The semantics of the \tmpda  is  defined as follows: a run $\sigma$ is a sequence of  alternating time elapse and discrete 
transitions  from one configuration to another. The time elapses are non-negative real numbers, and, on  discrete transitions,  
the valuation $\nu$  of the current configuration is  checked  
to see if the  clock constraints are satisfied; likewise, on a pop transition, the age 
of the symbol popped is checked. Projecting out the operations of a single stack 
from $\sigma$ results in a well-nested  sequence.
 A run is accepting if it starts from the initial state with all clocks set to $0$, and reaches a final state with all stacks empty.
The language accepted by a \tmpda{} is  defined as the set of 
timed words generated by the accepting runs of the \tmpda.
Since the reachability problem for \tmpda{} is  
Turing complete (this is the case even without time), we consider under-approximate reachability.   
 
A sequence of transitions is said to be \textbf{complete} if each push has a 
matching pop and vice versa.
A sequence of transitions is said to be \textbf{well-nested}, denoted $ws$, if 
it is a sequence of \nop-transitions, or a concatenation of well-nested 
sequences $ws_1ws_2$, or a well-nested sequence surrounded by a matching push-pop 
pair $\push({\downarrow}_i{\alpha}) ~ws ~\pop({\uparrow}^I_{i}{\alpha})$.
If we visualize this by drawing edges between pushes and their corresponding pops, well-nested sequences have no crossing edges, 
as in \includegraphics[scale=0.2]{ws1.pdf} and \includegraphics[scale=0.2]{ws2.pdf}, where we have two stacks, depicted with red 
and violet edges.  We emphasize that a well-nested sequence can have
well-nested edges from any stack.  In a sequence $\sigma$, a push
(pop) is called a \textbf{pending} push (pop)  if its matching pop
(push) is not in the same sequence $\sigma$. For \tmpda{} every
sequence also carries total time elapsed during the sequence, this is
helpful to check stack constraints, and it is sufficient to store time
till the maximum stack constraint, i.e., the maximum constant value
that appeared
in the stack constraints.

\subsection*{Tree Width of Bounded Hole \tmpda{}}
 We will capture the behaviour(any run $\rho$) of $K-$hole bounded
 multistack pushdown systems as a graph $G$ where, every node
 represents a transition $t \in \Delta$ and the edge between the nodes can be of three types.

 \begin{itemize}
     \item Linear order($\preccurlyeq$) between the transition which gives the order
       in which  the transitions are fired. We will use
       $\preccurlyeq^+$ to represent transitive closure of $\preccurlyeq$.
       \item Timing relations $\curvearrowright^{c\in I}~ \in ~\preccurlyeq^+
         \forall c \in \mathcal{X}$  and $I \in \mathcal{I}$ such that, $t_1
         \curvearrowright^{c\in I} t_2$ if and only if
         the clock constraint $c \in I$ is checked in the transition $t_2$ and $t_1
         \preccurlyeq^+ t_2$ has
         the latest reset of clock $c$ with respect to $t_2$.
       
     \item The other type of edges represent the push pop relation
       between two transitions. Which means, if a transition $t_1$ have
       a push operation in any one of the stack $i$  and transition $t_2$
       has pop transition of the stack $i$ which matches with the push
       transition at $t_1$, then we
       have an edge $t_1 \curvearrowright^s t_2$  between them, which will
       represent the stack edge.
     \end{itemize}

     To prove that the tree width of the class of graph $G$ is bounded,
     we will use coloring game \cite{AkshayGKS17} and show that we need bounded number of
     colors to split any graph $g \in G$ to atomic tree terms.

     Eve will start from the right most node of the graph by coloring
     it. The last node of the graph can be any one of the following,
     \begin{itemize}
     \item End point of a well-nested sequence
       \item Pop transition $t_{pp}$ of stack $i$ , such that, the push
         $t_{ps}$ is coming from nearest \emph{hole} of stack $i$.
       \end{itemize}
       \begin{enumerate}
\item If the end point colored is the end point of a well-nested sequence
then Eve can remove the well-nested sequence by adding another color to
the first point of the well-nested sequence. But, there may be some
transitions $t$ in the well-nested sequence with clock  constraints $c\in I$ such
that, the recent reset of the clock $c$, with respect to $t$  is in the left of the well
nested sequence. In order to remove the well-nested sequence she have
 to color the nodes which represent the transitions with recent reset points of the clocks
$c\in \mathcal{X}$. This step require at
most $|\mathcal{X}|$ colors. Now, she can split the
graph in two parts, one of them will be well-nested with two end
points colored. Also, the  clock constraint edges, which are coming from
the left of the well-nested sequence are hanging in the
left, are colored. There can be at most $|\mathcal{X}|$ hanging colored points
possible in the left of the well-nested sequence. The other part will be the remaining graph with the
right most point colored along with the colored recent reset points on
the left of right most colored point.
which are also the hanging points of the previous partition.

If we look at the well-nested part with hanging clock edges, using just one more color we can
split it to atomic tree terms \cite{AkshayGKS17}. 

But the other part still remains a graph of class $G$ so Adam will
choose this partition for Eve to  continue
the coloring game.

\item If the last point of the graph $G$ is a pop point $t_{pp}$ as discussed
  earlier, then the corresponding push $t_{ps}$ can come from a open
  \emph{hole} or a closed \emph{hole}.
  \begin{itemize}
  \item If it is coming from a closed
  \emph{hole} then,  Eve will add color to the corresponding push $t_{ps}$ along
  with the transition $t_{q}$ such that, $t_{ps} \preccurlyeq^+ t_{q}$
  and $t_{ps} \-- t_{q}$ is a well-nested sequence, which forms a atomic
  hole segment $(\uparrow ws)$ where, $\uparrow$ represents the push pop edge $t_{ps}
  \curvearrowright^s t_{pp}$ and $ws$ represents the well-nested sequence
  $t_{ps} \-- t_{q}$. But just as we discussed in previous scenario of
  removing well-nested sequence, there may be
  some clock constraint $c \in \mathcal{X}$ in the well-nested sequence  $ws$ such
  that the transition with the recent resets are from the left of
  $(\uparrow ws)$ and
  without coloring them Eve can not remove the
  $(\uparrow ws)$. Similarly, there may be some clock resets inside $\uparrow
  ws$ from
  which there are clock constraint edges are going to the right of
  $\uparrow ws$. Eve has to color all those points inside the $\uparrow ws$
  which corresponds to those clock reset points in $\uparrow ws$. So, she have
  to color at most $2|\mathcal{X}|$ reset points to remove the stack
  edge $t_1 \curvearrowright t_2$ along with the well-nested sequence
  $t_{ps} \-- t_{q}$($\uparrow ws$), which makes the closed \emph{hole} open with
  colors in both ends of hole and at most $|\mathcal{X}|$ colors in the left
  of the hole and at most $|\mathcal{X}|$ colored hanging points inside the
  \emph{hole}. This operation requires  $2+2|\mathcal{X}|$ more colors. Please note that, the right end of the
  \emph{hole} which got colored after removal of $t_{ps}\--t_{q}$ is another
  push of the \emph{hole}, because \emph{hole} are defined as a sequence
  $(\uparrow ws)^+$. 
  \item If the push is coming from open \emph{hole} then the push
    transition $t_{ps}$ must be colored from previous operation as
    discussed above, hence Eve will add another color  $t_{q'}$ to mark
    the next well-nested sequence $t_{ps} \-- t_{q'}(ws')$ in the right of $t_{ps}$. But,
    similar to above section here also there may be some
    clock resets of clock $i\in \mathcal{X}$ inside the $ws'$ which is being
    checked in the right of the $ws'$. These reset points can be at
    most $|\mathcal{X}|$ and needs $|\mathcal{X}|$ colors. Now, Eve can remove the stack edge
    $t_{pp} \curvearrowright t_{ps}$ along with the well-nested
    sequence $ws'$. This operation  widens the \emph{hole}. Note that at any
    point of the game, hanging clock reset points inside the \emph{hole}
    and in left side of hole is bounded by $|\mathcal{X}|$. This
    operation requires at most $1+|\mathcal{X}|$ colors but subsequent application of
    this operation can reuse colors. 
  \end{itemize}
  In both the above operations, we can split the graph in two parts,
  one with a stack edge $t_{pp} \curvearrowright t_{ps}$ and a well-nested sequence, with at most $|\mathcal{X}|$ hanging points for each clock in
  the left of the $t_{pp}$ and at most $|\mathcal{X}|$ colors inside the
  $ws$. which require at most $1$ color to split into
  atomic tree terms without any extra colors. On the remaining part Eve
  will continue playing the game from right most point.
\end{enumerate}
Here, we claim that at any point of time of the coloring game, there will
be  $2K + (2K+1)|\mathcal{X}| +2 $ active
colors for $K \geq 1$ and $ K \in \mathbb{N}$. Every step of the game
splits the graph in two part, and one part always can be split into
atomic tree terms with at most $2|\mathcal{X}|+3$ colors. The remaining part will require $2 +2|\mathcal{X}|$
colors for every open \emph{hole} in the left of the right most point
of the graph. As the number of open \emph{hole} is bounded by $K$, so we can not have more
than $K$ open \emph{holes} in the left of any point. So, $2K + 2K|\mathcal{X}|$ colors to mark the
\emph{holes}, $1 + |\mathcal{X}|$ for the right most point and recent reset
points with respect to the right most point, $1+|\mathcal{X}|$ for coloring the
well-nested sequence after a matched push and the possible reset
points inside the well-nested sequence, but we will need to color such
well-nested sequence once at any point of time, which gives a total color of
$2K(|\mathcal{X}|+1) + 2(|\mathcal{X}|+1) = (2K+2)(|\mathcal{X}|+1)$.

\section{Reachability in \tmpda{}}
\label{app:timedreach}
In this section, we discuss how the BFS tree exploration extends 
in the timed setting. To begin, we talk about how a list at any 
node in the tree looks like.

\subsection*{Representation of Lists for BFS Tree}
Each node of the BFS tree stores a list of bounded length. A list is a 
 sequence of states $(s,\nu)$ separated by time
 elapses ($t$), representing a  $K$-hole bounded run in a concise
 form. The simplest kind of list is a single state $(s,\nu)$
or a well-nested sequence $(s,\nu,t,s_i,\nu_i)$ with time
elapse $t$.  Note that because of time constraints we need to store
total time elapsed to reach one state from another. This is why we are
keeping a time stamp between two states. Recall, the hole in \mpda{} is
defined as a tuple $(i,s,s')$. For \tmpda{} we need to store
total time elapsed in the hole as well, so it can be represented as a tuple
$H=(i,s,\nu,s',\nu',t_h)$, where, $t_h$ is the time elapse in the hole and
$(s,\nu), (s',\nu')s'$ being the end states of the hole. Also, the maximum possible
value of time stamp is bounded by the maximum integer value in the
constraints (both pop and clock). So, the total possible values that
the variable $t_i$ can take is also bounded.
Let $H, t$ represent respectively holes (of some stack) and time elapses. 
A list with holes has the form $(s_0,\nu_0).t.(H)^*(H.t.(s',\nu'))$. For example, a list with 3 holes 
of stacks $i, j, k$ is

{\center
{\small $[(s_0,\nu_0),t_1$,\textcolor{red}{$(i,s_1,\nu_1,s_2,\nu_2,t_2)$},$t_3$,\textcolor{red}{$(j,s_3,\nu_3,s_4,\nu_4,t_4)$},$t_5$,\textcolor{red}{$(k,s_5,\nu_5,s_6,\nu_6,t_6)$},$t_7,(s_7,\nu_7)]$}
}

\subsection*{Algorithms for \tmpda{}}\label{app:algotimed}

\begin{algorithm}[!h]
  
  {
    \scriptsize
        \SetKwProg{Return}{return}{\string;}{}
        \SetKwProg{Fn}{Function}{\string:}{}
        \SetKwFunction{NumberOfStacks}{NumberOfStacks}
        \SetKwFunction{NumberOfHoles}{NumberOfHoles}
        \SetKwFunction{pushpopclosure}{WellNestedReachTimed}
        \SetKwFunction{concatclosure}{ConcatClosure}
        \SetKwFunction{addpop}{AddPopTimed$_i$}
        \SetKwFunction{addhole}{AddHoleTimed$_i$}
        \SetKwFunction{pushclosure}{AtomicHoleSegmentTimed$_i$}
        \SetKwFunction{pushclosureC}{HoleSegmentTimed$_i$}
        \SetKwFunction{addcomplete}{AddWellNested}
        \SetKwFunction{isempty}{IsEmptyTimed}
        \SetKwComment{Comment}{\textbackslash\textbackslash}{}
        \Fn{\isempty{$M=(\mathcal{S},\Delta, s_0,\mathcal{S}_f,$ $
            \mathcal{X},n,\Sigma,\Gamma) ,K$}}{
        \KwResult{True or False}
            \wsrt{} := \pushpopclosure{$M$};
            \Comment{\textcolor{red}{Solves binary reachability for
                pushdown system}}
            
            \If{some $(s_0,\nu_0,t,s_1,\nu_1)\in$ \wsrt{} with $s_1 \in \mathcal{S}_f$ }{
                \Return{False}{}
              }
                \ForAll{$ i \in [n]$}{
               $AHST_i := \emptyset$\;

 \ForAll{$(s,\phi,{\downarrow}_{i}(\alpha),\rho,a,s_1)\in \Delta$,
   $\nu \models \phi$, and $\nu_1 = \rho[\nu]$  }{

    \ForAll{$ (s_1,\nu_1,t,s',\nu') \in$ \wsrt}{
       $AHST_i:= AHST_i \cup (i, s,\nu,\alpha,s',\nu',t)$\;
    }
  }
   $Set_i :=  \{(s,\nu,t,s',\nu') \mid \exists \alpha (i, s,\nu,\alpha,s',\nu',t) \in AHS_i\}$\;
 $HS_i := \{(i,s,\nu,s',\nu',t)\mid (s,\nu,t,s',\nu') \in \Transitive{$Set_i$} \}$\;
             
        }
        $\mu:= [s_0,\nu_0]$\;
        $\mu.\NumberOfHoles := 0$\;

         $\sol_{new} := \{\mu\}, \sol_{old}:=\emptyset$\;

        \While{$\sol_{new} \setminus \sol_{old} \neq \emptyset$}{
          $\sol_{diff} := \sol_{new} \setminus \sol_{old}$\;
         $ \sol_{old}:= \sol_{new}$\;
         \ForAll{ $\mu' \in \sol_{diff}$}{

            \If{$\mu'.\NumberOfHoles < K$}{

                \ForAll{$ i \in[n]$}
                {
                    $\sol_{h}$ :=
                    \addhole($\mu', HST_i$);~\Comment*[h]{\textcolor{red}{Add
                      hole for stack i}}\\
                  \ForAll{$\mu_2 \in \sol_{h}$}
                  {
                   
                        $\sol_{new} := \sol_{new} \cup \mu_2$\;}
                      
                }
            }

            \If{$\mu'.\NumberOfHoles > 0$}
            {
              \ForAll{$i \in [n]$}
              {
                    $\sol_p$:=
                    \addpop($\mu', M, AHST_i, HST_i$,\wsrt{});~\Comment*[h]{\textcolor{red}{Add
                        pop for stack i}
                    }\\
                    \ForAll{$\mu_3 \in \sol_p$}
                    {
                       \If{$\mu_3.last \in \mathcal{S}_f$ and
                         $\mu_3.\NumberOfHoles = 0$}
                       {
                         \Return(\Comment*[h]{\textcolor{red}{If reached
                             destination state}}){False}{}
                       }
                     
                         $\sol_{new}:=\sol_{new} \cup \mu_3 $\;
                         
                    }
                }
            }
           
          }

          }
       
      \Return{$True$}{}
  } 
  }  
  \caption{Algorithm for Emptiness Checking of hole bounded \tmpda{} }\label{alg:isemptytimed}
    
\end{algorithm}   

  \begin{algorithm}[!h]\label{alg:tlapse}
      \SetKwProg{Return}{return}{\string;}{}
      \SetKwProg{Fn}{Function}{\string:}{}
      \SetKwFunction{min}{Min}
      \SetKwFunction{states}{States}%
       \Fn{\states{$M=(\mathcal{S},\Delta, s_0,\mathcal{S}_f,$ $ \mathcal{X},n,\Sigma,\Gamma)$ }}{
         \KwResult{$F$}
         $F := \{(s,\nu)\mid \forall s \in \mathcal{S} \wedge
         \forall c \in \mathcal{X}, \nu[c] \leq max(c)+1\}$\;
       
        }
        \Return{$F$}{}
       \caption{States}
      \end{algorithm}

      \begin{algorithm}[h]\label{alg:tlapse}
      \SetKwProg{Return}{return}{\string;}{}
      \SetKwProg{Fn}{Function}{\string:}{}
      \SetKwFunction{min}{Min}
      \SetKwFunction{TimeElapse}{TimeElapse}%
       \Fn{\TimeElapse{$(s_1,\nu_1)$}}{
         \KwResult{$Set$}
         $Set := \emptyset$\;
       $t := 0$\;
       \While{$t \leq$ c\textsubscript{max}}{
        $\forall i \in X: \nu_2[i] := \min(\nu_1[i] +t,c_i)$\;
        $Set := Set\cup (s_1,\nu_1,t,s_1,\nu_2)$ \;
        $t := t + 1$\;
        }
        }
        \Return{$Set$}{}
       \caption{Time Elapse}
      \end{algorithm}

The function \href{alg:tlapse}{TimeElapse}
returns the states which are reachable from the state $(s_1,\nu_1)$ via
time elapse. It
also stores the total time elapsed to reach the state. This function
is only useful for timed systems.

\begin{algorithm}[!h]\label{alg:pushpopctimed}
\SetKwProg{Return}{return}{\string;}{}
\SetKwProg{Fn}{Function}{\string:}{}
\SetKwFunction{pushpopclosure}{WellNestedReachTimed}
\SetKwFunction{closure}{ConcatClosureTimed}
\SetKwFunction{timeelapse}{TimeElapse}
\SetKwFunction{Transitive}{TransitiveClosureTimed}
\SetKwFunction{states}{States}
 \Fn{\pushpopclosure{$M=(\mathcal{S},\Delta, s_0,\mathcal{S}_f,$ $ \mathcal{X},n,\Sigma,\Gamma)$}}{

 \KwResult{\wsrt{} $:= \{(s,\nu,t,s',\nu')| (s',\nu') $ is
   reachable from $(s,\nu)$ by time elapse $t$ via a well-nested 
   sequence\}}
 F = \states{M}\;
 Set = $\{(s,\nu,p,s,\nu)\mid (s,\nu) \in F\}$\;
 \ForAll{$(s,\nu) \in F$}{
   $Set = Set \cup $\timeelapse{$(s,\nu)$}\;
   
  \ForAll{$(s,\varphi,\text{\nop{}},a,R,s') \in \Delta$ with $\nu \models \phi$}{
           $Set := Set\cup (s,\nu,0,s',R[\nu])$ 
         }
         }
 $\mathcal{R}_{tc}$= \Transitive{Set}\;
 \While{True}{
     $\wsrt{} :=\mathcal{R}_{tc}$\;
     \ForAll{$(s,\phi_1,{\downarrow}_{i}(\alpha),\rho_1,a,s_1)\in
       \Delta$ and $(s,\nu) \in F$ 
     with $\nu\models\phi_1$}{
         \ForAll{$ (s_1,\rho_1[\nu],t,s_2,\nu_2) \in \mathcal{R}_{tc}$}{
             \ForAll{$(s_2,\phi_2,{\uparrow}_{i}^{I}(\alpha),\rho_2,a,s')\in
               \Delta$
             with $\nu_2\models\phi_2$, $t\in I$}{

             $ \mathcal{R}_{tc} :=  \mathcal{R}_{tc} \cup (s,\nu,t,s',\rho_2[\nu_2])$\;
         }
         }
     }
     $ \mathcal{R}_{tc}$:= \Transitive{$ \mathcal{R}_{tc}$}\;
     \If{ $ \mathcal{R}_{tc} \setminus \wsrt{} = \emptyset$}{
         break\;
     }
 }
 }
 \Return{\wsrt{}}{}
 \caption{Well Nested Reach Timed}

\end{algorithm}

\begin{algorithm}[!h]\label{alg:addholetimed}
\SetKwProg{Return}{return}{\string;}{}
\SetKwProg{Fn}{Function}{\string:}{}
\SetKwFunction{NumberOfHoles}{NumberOfHoles}
\SetKwFunction{pushclosure}{PushClosureTimed$_i$}
\SetKwFunction{pushclosureC}{\ensuremath{PushClosureTimed^*}}
\SetKwFunction{addH}{AddHoleTimed$_i$}
\Fn{\addH{$\mu,HST_i$}}{
\KwResult{$Set$ = \{ $\mu|$ $\mu $ is a list of states and time
elapses\}}

$Set := \emptyset$\;
$(s,\nu) := \last(\mu)$\;
\ForAll{$(i,s,\nu,t,s',\nu')\in HST_i$}{
$\mu' = copy(\mu)$\;
  $ \trunc(\mu')$;
   \qquad /* $\trunc(\mu)$ is defined as $\remove(\last(\mu))$) */\\
   $\mu'.\append[(i,s,\nu, t, s',\nu'), 0, (s',\nu')]$\;
   $\mu'.\NumberOfHoles := \mu.\NumberOfHoles + 1$\;
   $Set := Set \cup \{\mu'\}$\;
}
}
\Return{$Set$}{}

\caption{Add Hole Timed}
\end{algorithm}

\begin{algorithm}[!h]\label{alg:extendwithPopTimed}
\SetKwProg{Return}{return}{\string;}{}
\SetKwProg{Fn}{Function}{\string:}{}
\SetKwFunction{pushcomplete}{PushCompleteTimed$_i$}
\SetKwFunction{pushclosure}{PushClosureTimed$_i$}
\SetKwFunction{addpop}{AddPopTimed$_i$} 
\Fn{\addpop{$\mu,M=(\mathcal{S},\Delta, s_0,\mathcal{S}_f, \mathcal{X},n,\Sigma,\Gamma),AHST_i,HST_i,\wsrt{}$}}{
\KwResult{$Set$ = \{ $\mu \mid \mu $ is a list of states and time
elapses\}}
$Set := \emptyset$\;
$[t_l,(s,\nu)] := \last(\mu)$\;
$\textcolor{red}{[t',(i,s_1,\nu_1,t,s_3,\nu_3),t'']} := lastHole_i(\mu)$\;
$t_3 :=$ The sum of the time elapses in the list $\mu$ between
$(s_2,\nu_2)_{R_i}$ and $(s,\nu)$\;
\ForAll{$ (i,s_1,\nu_1,t_1,s_2,\nu_2) \in HST_i$, 
    $(i,s_2,\nu_2,t_2,\alpha,s_3,\nu_3)\in AHST_i$, 
    $(s,\phi,R,{\uparrow}^{I}_{i}(\alpha),s') \in \Delta$ 
    with $t=t_1+t_2$, $\nu\models\phi$ and $t_2+t_3\in I$, and
    $(s',R[\nu],t_4,s'',\nu'')\in \wsrt{} $}{
    $\mu' = copy(\mu)$\;
  $\trunc(\mu')$\;
  $\mu'.\append([t_l \oplus t_4,(s'',\nu'')]$\;
  $\mu'.\mathsf{replace}(\textcolor{red}{[t',(i,s_1,\nu_1,t,s_3,\nu_3),t'']},$
        \mbox{\qquad\qquad}$\textcolor{blue}{[t',(i,s_1,\nu_1,t_1,s_2,\nu_2),t_2\oplus t'']})$\;
  $Set := Set \cup \{\mu'\}$\;
  \If{$t_1=0$ and $(s_1,\nu_1)=(s_2,\nu_2)$}{
$    \mu'' = copy(\mu)$\;
     $\trunc(\mu'')$\;
     $\mu''.\append([t_l \oplus t_4,(s'',\nu''))$\;
     $\mu''.\mathsf{replace}(\textcolor{red}{[t',(i,s_1,\nu_1,t,s_3,\nu_3),t'']}$,
     $\textcolor{blue}{(t' \oplus t \oplus t'')})$\;
     $\mu''.\NumberOfHoles = \mu.\NumberOfHoles-1$\;
     $Set := Set \cup \{\mu'\}$\;
  }
}

}
\Return{$Set$}{}

\caption{Extend with a pop Timed}

\end{algorithm}    

\newpage

    \section{Witness Generation for \tmpda{}}
    \label{app:timedwit}
    In this section, we focus on the important question of generating a witness for an accepting run whenever our fixed-point algorithm guarantees non-emptiness. Since we use fixed-point computations  to speed up our reachability algorithm, finding a witness, i.e., an explicit run witnessing reachability, becomes non-trivial. In fact, the difficulty of the witness generation depends on the system under consideration : while it is reasonably straight-forward for timed automata with no stacks, it is quite non-trivial when we have (multiple) stacks with non-well nested behavior.

  \begin{algorithm}[!h]\label{alg:uniquepathtimed}
 
\SetKwProg{Return}{return}{\string;}{}
\SetKwProg{Fn}{Function}{\string:}{}
\SetKwFunction{Witnss}{Witness}
\SetKwFunction{useless}{WitnessTimedWR}
\SetKwFunction{timeLapse}{TimeLapse}
\SetKwFunction{useful}{UseFulPath}
\SetKwFunction{firable}{Firable}
\SetKwFunction{usefulTr}{UsefulTransition}

\Fn{\useless{$s_1,s_2,\nu,M=(\mathcal{S},\Delta, s_0,\mathcal{S}_f,$ $ \mathcal{X},n,\Sigma,\Gamma)$,\wsrt}}{
\KwResult{ A sequence of transitions for an accepting run}

\If{$s_1 == s_2$}{
  \Return{$\epsilon$}{}
}
\If{$\exists  t= (s,\phi,R,\nop,s') \in \Delta \wedge \nu \models
  \phi \wedge \nu = R[\nu]$}{
  \Return{$t$}{}
  }
  \ForAll{$s',s'' \in \mathcal{S}$}{
  
   \If{     $((s_1 \neq s') \vee (s'' \neq s_2)) \wedge (s',s'')\in \wsrt$
     $\wedge \exists t=(s_1,\phi,R,\downarrow_i(\alpha),a,s') \in \Delta \wedge$
     $\exists t_2 =(s'',\phi',R',\uparrow_i(\alpha),a',s_2) \in
     \Delta \wedge \nu = R[\nu] = R[\nu'] \wedge \nu \models \phi
     \wedge \nu \models \phi'$ }{
     path=\useless{$s',s'',\nu,M$,\wsrt}\;
     \Return{$t.path.t_2$}{}
     }
   }
   \ForAll{$s \in M.\mathcal{S}$}{
     \If{$(s \neq s_1 \vee s \neq s_2)\wedge (s,0,s_1) \in \wsrt \wedge
       (s,0,s_2) \in \wsrt $}{
       path1=\useless{$s_1,s,\nu,M$,\wsrt}\;
       path2 = \useless{$s,s_2,\nu,M$,\wsrt}\;
       \Return{path1.path2}{}
     
     }

}
}

\caption{ Well-nested Timed  Witness Generation}

\end{algorithm}

\begin{algorithm}[!h]
 \scriptsize{
\SetKwProg{Return}{return}{\string;}{}
\SetKwProg{Fn}{Function}{\string:}{}
\SetKwFunction{Witnss}{Witness}
\SetKwFunction{useless}{WitnessTimedWR}
\SetKwFunction{timeLapse}{TimeLapse}
\SetKwFunction{useful}{UseFulPath}
\SetKwFunction{firable}{Firable}
\SetKwFunction{usefulTr}{UsefulTransition}
\Fn{\Witnss{$(s_1,\nu_1),t,(s_2,\nu_2),M=(\mathcal{S},\Delta, s_0,\mathcal{S}_f,$ $ \mathcal{X},n,\Sigma,\Gamma)$,\wsrt}}{
\KwResult{ A sequence of transitions for an accepting run}

   \ForAll{$t_1 \in [T]$}{
   midPath = \Witnss{$(s_1,\nu_1+t_1),t-t_1,(s_2,\nu_2),M,\wsrt$}~
   {\textcolor{red}{Progress Measure 1}}\;
   \If{midPath $\neq \emptyset$}{
    \Return{$ t_1 \cdot$midPath}{}
   }
   }
   \ForAll{$\delta=(s'',\phi',R',\nop{},a',s_2) \in M.\Delta$}{
   \If{$\delta.R'[\nu_1] \neq  \nu_1$ and $\nu_1 \models \delta.\phi'$)}{
    $s_3$ = $\delta.s_2$\;
    $\nu_3$ = $\delta.R'[\nu_1]$\;
    midPath2 = \Witnss{$(s_3,\nu_3),t,(s_2,\nu_2),M,\wsrt$}  {\textcolor{red}{Progress Measure 2}}\;
    \If{midPath2 $\neq \emptyset$}{
    \Return{ $\delta \cdot$midPath2}{}
  }
}
}
  \ForAll{$s \in M.\mathcal{S}$}{
   path = \useless{$s_1,s,\nu_1,M$,\wsrt}  {\textcolor{red}{Progress Measure 3}}\;
   \If{path $\neq \emptyset$}{
     midPath3 = \Witnss{$(s,\nu_1),t,(s_2,\nu_2),M$,\wsrt}\;
     \If{midPath3 $\neq \emptyset$}{
       \Return{path $\cdot$ midPath3}{}
     }
     }
}

}
}

\caption{Timed Pushdown Automata Witness Generation}\label{alg:witnesstimed}
\end{algorithm}

\noindent{\bf 0-holes}. We start discussing the witness generation in the case of timed automata. As described in the algorithm in section \ref{sec:algo}, non-emptiness is guaranteed 
if a final state $(s_f, \nu_f)$ is reached from the initial state
$(s_0, \nu_0)$ by computing the transitive closure of the
transitions. The transitive closure computation results in generating
a tuple $(s_0,\nu_0,t,s_f,\nu_f) \in \wsrt{}$ (Algorithm~\ref{alg:pushpopctimed}),  
for some time $0\leq t \in \mathbb{R}$. Notice however that, in the
Algorithms~ \ref{alg:pushpopctimed}, 
we do not keep track of the sequence of states that led to the final state, and 
this is why we need to 
reconstruct a witness. To generate a witness run, 
we consider a \emph{normal form} for any run in the underlying timed automaton, and check for the existence of a witness in the normal form. A run is in the normal form if it is a sequence of \emph{time-elapse}, \emph{useful}, and \emph{useless} transitions. Time-elapse transitions have already been explained earlier. 
A discrete transition $(s,\nu) \stackrel{}{\rightarrow} (s',\nu')$ is \emph{useful}
if $\nu \neq \nu'$ , that is, there is at least one clock $x$ such that $\nu'(x)=0$ and $\nu(x)\neq 0$. A discrete transition is \emph{useless} 
if $\nu=\nu'$. 
 
If a tuple $(s_0,\nu_0,t,s_f,\nu_f)$, $t\geq  0$ is generated by  Algorithm \ref{alg:pushpopctimed}, 
we know that the system is non-empty. Now, 
 we describe an algorithm to generate   
the witness run for obtaining $(s_0,\nu_0,t,s_f,\nu_f)$, 
 by associating a \emph{lexicographic progress measure} while exploring 
runs starting from $(s_0, v_0)$. Integral time elapses, useful transitions and useless transitions 
are the three entities constituting the progress measure, ordered lexicographically.  
\begin{itemize}
\item First we check if it is possible to obtain a witness run of the form 
$(s_0, \nu_0) \stackrel{t_1}{\rightsquigarrow}(s, \nu)\stackrel{t_2}{\rightsquigarrow}(s_f,\nu_f)$, where 
$\stackrel{t}{\rightsquigarrow}$ denotes a sequence of transitions whose total time elapse is $t$. 
In case $t_1, t_2>0$, with $t_1+t_2=t$, we can recurse on obtaining witnesses 
to reach $(s,\nu)$ from $(s_0, \nu_0)$, and $(s_f, \nu_f)$ from  
  $(s,\nu)$, with strictly smaller time elapses, guaranteeing progress to termination.   
\item In case $t_1=0$ or $t_2=0$, we move to the second component of our progress measure, namely useful 
transitions. Assume $t_2=0$. Then indeed, there is no time elapse 
in reaching $(s_f, \nu_f)$ from $(s,\nu)$, but only a sequence of discrete transitions. 
Let $\#_X(\nu)$ denote the number of non-zero entries in the valuation $\nu$. To obtain the witness, 
we look at a maximal sequence of useful transitions from $(s,\nu)$ of the form $(s,\nu) \stackrel{}{\rightarrow} (s_1, \nu_1) \stackrel{}{\rightarrow} \dots 
\stackrel{}{\rightarrow} (s_k, \nu_k)$ such that $\#_X(\nu) >  \#_X(\nu_1) >  \dots >  \#_X(\nu_k)$, where 
 $k \leq $ the number of clocks. When we reach some $(s_i, \nu_i)$ from where we cannot make a useful transition, we go for a useless transition. Since there is no time elapse, and no 
useful resets, the clock valuations do not change on discrete transitions. We are left with enumerating 
all the locations  to check the reachability to $s_f$ (or to some $s_j$, from where we can again have a maximal sequence 
of useful transitions). Indeed, if $(s_f, \nu_f)$ is reachable from $(s,
\nu)$ with no time elapse, there is a path having at most $|\mathcal{X}|$ useful transitions, interleaved with 
a sequence of useless transitions. 
\end{itemize}

Generation of witness for timed automata 
is given in Algorithm~\ref{alg:witnesstimed}.
Notice that when  $\kappa=(s_0,v_0,0,s_f,v_f)$, 
the progress measure is $m(\kappa)=\#_X(\nu_0)-\#_X(\nu_f)$. If $m(\kappa)=0$, then $\nu_0=\nu_f$, and 	the path takes only useless transitions.  
In this case, we consider the graph with nodes as states $(s, \nu)$, and there is an edge from $(s_1, \nu_1)$ to $(s_2, \nu_2)$ 
if there is a transition $(s_1, \varphi, R, s_2)$ such that $\nu_1 \models \varphi$ and $\nu_1[R]=\nu_1$, that is, for all $x \in R$, 
$\nu_1(x)=0$. If $m(\kappa) \neq 0$, then we take at least one useful transition. We can check 
if there exists a transition $(s_1, \varphi, R, s_2)$ such that $s_1$ is reachable from $s_0$, and 
$\nu_0 \models \varphi, \nu_0[R] \neq \nu_0$, and the tuple 
$\kappa'=(s_2, \nu_0[R], 0, s_f, \nu_f) \in \wsrt{}$. In this case, we have 
$m(\kappa') < m(\kappa)$ and we can conclude by induction. 

 The case of a timed pushdown system with a single stack is similar to the case of timed automata, 
except for the fact that a discrete transition may involve push/pop
operations. We use the same progress measures as in the
timed automaton case, using the notion of runs in normal form.

\noindent{\bf Getting Witness from Holes}. We
can extend the backtracking algorithm for witness generation for
\mpda{} to generate witness for \tmpda{} without much
modification. In timed settings we need to take care of the time
elapses within a hole and an atomic hole segment. When a hole is partitioned
to an atomic hole segment and a hole, the time must be partitioned
satisfying possible atomic hole segments and holes along with other
constraints.

\section{Details for Experimental Section}
   \subsection{Multi-Producer Consumer}
   \paragraph{Description}
Suppose in a factory two machines (called producers)
produce products `A' and `B' respectively in batches. Another
machine (called consumer) first consumes one `A' and
then one `B' and repeats this process. $A$'s are produced in batches
of  size $M>0$, and $B$'s are produced in batches of size $N>0$.
On any day,
the number of `A's produced is a non-zero multiple of $M$ and the
number of `B's produced is a non-zero multiple of $N$ ($M,N>
0$).  
To avoid wastage, the factory
aims  to consume all `A's and `B's generated by the producers each
day. At any point of time, the factory does one of the following:
 produce a batch of $M$ $A$s, or produce a batch of $N$ $B$s, or 
 consume  $A,B$ in order, alternately, one at a time.
 Given $M, N$, the problem we are interested is to check if  
  all the $A$s and $B$s which are produced in a day can be completely consumed
 by end of day. 
\paragraph{Modeling}
The factory is modeled using a parametric two-stack pushdown system,
one stack per product, $A, B$.  The parameters used are the sizes of the batches,  $M$
and $N$ respectively. The production and consumption of the `A's and `B's are modeled using
push and pop in the respective stack. For a given $M$ and
$N$, the language accepted by the system is empty if it is not possible to satisfy the
conditions mentioned above. Otherwise, there will be a run where, all the produced `A's and
`B's are consumed.
\paragraph{Language} The language accepted by the two-stack pushdown
system is given by $((a^M+b^N)^+ (\bar{a}\bar{b})^+)^+$, where  
 $a,b$ represent respectively,  the push on stack 1, 2. Likewise,
 $\bar{a}, \bar{b}$ represent respectively,  the pop on stack 1, 2. 
\paragraph{Hole Bound}For any parameter $M, N>0$, the shortest
accepting run of
the two stack pushdown system is not well-nested.
\begin{theorem}
  1-hole bounded runs are same as 0-hole bounded or 
  well-nested runs.
  \label{thm:1}
\end{theorem}
\begin{proof}
  Assume, there exists an accepting run with 1-hole bound. From the definition of
  hole, we know that from any point $p$ of the run, there exist only one
  open hole to the left of $p$. Let $\ell< r$ respectively 
  be the left and right end points of this hole.
  Any point $r < i < p$ will be a part of some 
  well-nested sequence (due to hole bound 1). Let us consider
  extending the run beyond point $p$. There are two possibilities :
  (1) extending with a well-nested sequence, or (2) adding a
  pop that matches the last push of the hole. In the latter case,
  we have ``shrunk'' the hole or extended the last well-nested
  sequence between $r$ and $p$. Since an accepting run
  has no pending pushes to be matched, if we proceed
  with (1), (2) above, we will eventually
  complete the  run as a well-nested sequence.  By definition of
  holes, such an accepting run is a well-nested sequence having 0
  holes.  
   \qed{}
\end{proof}
Hence, there can be no run with hole bound less than 2.
  For any $M$ and $N$ \trim{} can
generate the minimum possible run with 2 bound on
holes. Because, we need one hole each for producing `A's and `B's.
Notice that the in the smallest run the number of `A's and `B's is equal to $LCM(M,N)$.
The consumption of `A's and `B's happen in an
order one by one i.e., in a sequence where consumption of `A' and `B'
alternate. This increases the number of context changes and if we
count them the bound on
scope turns out to be $O(2\times LCM(M,N))$ and bound on context is $O(2\times LCM(M,N))$.

  The run of the automata is as follows:
  \textcolor{red}{t3--t4--t5--t3--t4--t5}--\textcolor{blue}{t12--t14--t12--t14--t12--t14}--t15--t16--t19--t20--t19--t20--t19--t20--t19--t20--t19--t20 

The length of the witness is 24. The red color represents the hole of
stack 1 and the blue color represents the hole of stack 2. As marked
we need 2 holes to accept this run in \trim{}. If we try to compute
the scope bound on the run above, we will see that the maximum scope
starts at $t3$ and ends in $t20$, the number of context changes
between this transitions is 15. And total context change in this run
is 17.
 
\subsection{ Concurrent Insertion in Binary Search Tree (BST)}\label{expt:bst}
As explained earlier, we consider the algorithm proposed
in~\cite{kung1980concurrent} which solves this problem for concurrent
implementations.  But, if the locks are not implemented properly then it is
possible for a thread to overwrite others as shown in~\cite{chaki2006verifying}.  We have modified this algorithm so
that it becomes buggy and then we tried to model it using \mpda{},
\trim{} was able to detect the bug.  In short, we need two stacks to simulate
two recursive threads of binary  tree insertion.  The problem
(reachability query) arises when both threads tries to update the same sub tree,
simultaneously which is a bad run of the system.  Then we note that to find such
bad runs at least 4 context changes are  required in the underlying \mpda{}, \trim{}
can find such bad runs with bound 2 on holes.

In more detail, assume the following implementation for the
\texttt{insert}~\cite{kung1980concurrent} operation in a BST:

{\scriptsize
\begin{lstlisting}[language=C]
BSTNode* insert(BSTNode* root,int value)
{
1    if(root == NULL){
2       return createNewNode(value);
3    }
4    
5    if(root->data == value){
6        return root;
7    } 
8    else if(root->data > value){
9        while(root->leftLock);
10        if(!root->left){
11            root->leftLock = true;
12            root->left = insert(root->left,value);
13            root->leftLock = false;
14        }
15        else{
16            root->left = insert(root->left,value);
17        }
18    }
19    else{
20        while(root->rightLock);
21        if(!root->right){
22            root->rightLock = true;
23            root->right = insert(root->right,value);
24            root->rightLock = false;
25        }
26        else{
27            root->right = insert(root->right,value);
28        }
29    }
30    
31    return root;
32    
}
\end{lstlisting}
}

 \begin{figure}[h]

    \centering
  \tikzstyle{level 1}=[level distance=2cm, sibling distance=3.5cm]
\tikzstyle{level 2}=[level distance=2cm, sibling distance=2cm]
\tikzstyle{bag} = [rectangle,draw=black]
\scalebox{0.7}{     
\begin{tikzpicture}[grow=down]
\node[bag] {10}
    child {
      node[bag] {5}
      child{
        node[bag] {2}
        edge from parent
      }
      child{
        node[bag] {6}
        edge from parent
        }
      edge from parent
    }
    child{
      node[bag] {15}
      child{
        node[bag] {13}
        edge from parent
      }
      child{
        node[bag] {20}
        edge from parent
        }
            edge from parent
            };
        \end{tikzpicture}
}        \caption{An example of BST}\label{fig:bst}

 \end{figure}
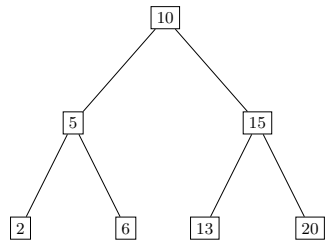

Consider the BST in Figure~\ref{fig:bst}.

Suppose 2 threads t1 and t2 are invoked simultaneously trying to
insert values 25 and 26 respectively and currently are at BSTNode with
value 20 (The rightmost node). To reach that node both of the threads
require to do recursive calls, which depend on the size of the
tree. In between a lot of context switching can
happen between the threads, here we present just the context changes
upon reaching the destination sub-tree i.e., 20.

{\scriptsize
\begin{lstlisting}[language=C]
  
a. t1:
          1. if(root == NULL)  //not true, will go to line 5.
          //switch

b. t2:
          1. if(root == NULL)  //not true, will go to line 5.
          //switch

c. t1:
          5. if(root->data == value){  //not true, will go to line 8.
          8. else if(root->data > value) //not true, will go to line 19.
          //switch

d. t2:
          5. if(root->data == value){  //not true, will go to line 8.
          //switch

e. t1:
          19    else{
          20        while(root->rightLock);  // lock is not held by anyone, so continue.
          21        if(!root->right){
          //switch
f. t2:
          8. else if(root->data > value) //not true, will go to line 19.
          19    else{
          20        while(root->rightLock);  // lock is not helpd by anyone, so continue.
          21        if(!root->right){
          22            root->rightLock = true;
          //switch

g. t1:
          22            root->rightLock = true;
          23            root->right = insert(root->right,value);
          //switch

h. t2: 
          23            root->right = insert(root->right,value);
          24            root->rightLock = false;
          //switch
        \end{lstlisting}
        }

As we can see in the above code section f,g and h that both t1 and t2 are entering into critical section without knowing the presence of each other.

Now we wish to check if the following set of instructions executed in one go:
\begin{lstlisting}[language=C]
20        while(root->rightLock);
21        if(!root->right){
22            root->rightLock = true;
\end{lstlisting}
\textbf{Observations:}

\begin{itemize}
  \item We need two stacks to simulate the recursive procedure
    \texttt{insert} on two threads.
    \item Note, that we can simulate the sequence of recursion call of each
      thread to destination sub-tree as a hole.
      \item Here the reachability query is whether there exists a run,
        where both the threads enter in the same sub tree and try to
        access same location.
\end{itemize}

\subsection{Double free in dm-target.c}\label{sbs:dmtarget} As mentioned
in~\cite{HagueL12}, to fix a memory-leak, in version 2.5.71 of the
Linux kernel, a double free was introduced to drivers/md/dm-target.c. When
registering a new target, memory was allocated for the target and
then a check made to see if the target was already known. If the
target exists, the allocated memory was freed and an flag (``exists'')
is set. Otherwise, the target was added to the target list as intended. Before
returning, the flag (``exists'') was checked and the object was freed
(again) if it was set. A run is faulty where, either an
item was removed from the empty list, the number of free calls was
greater than the number of allocations, or, the code exited normally,
but more memory was allocated than freed. Note, one of the stack is required to
track the size of the list which ensures that the number of allocations matches the
number of frees.

\subsection{$Maze_T$}
\paragraph{Description} We consider a parametrized example of a robot navigating a maze, picking
items, while visiting locations
in the maze respecting some time constraints, 
by extending the example 
 from~\cite{AGK16} to handle multiple stacks. 
 There are 9 locations in the maze as shown in
the Fig~\ref{fig:maze}. The global time $T$ to be spent in the maze,
is given as the parameter. 
The robot enters at  point 1, and exits via  point 9  after spending
$T >0$ time in the maze, respecting the following constraints.

\begin{enumerate}
\item   Each time it visits locations 3 and 5, it picks up one item from 
  these locations;
   \item On each visit to location 7, the robot drops there, the  items collected
   from  locations 3 and 5 in order one by one (first item from
   location 3 then item from location 5); 
\item  The time elapsed between picking up the  item from location 3
  and 
  dropping  it at location 
    7 must be within the interval \textcolor{red}{[4,6]};
 
 \item Likewise, the time elapse between picking up the  item at
   location 5 and  dropping it at  location
    7  must be within the interval
    \textcolor{red}{[1,4]};
    \item Before exiting the maze, the robot must have dropped  all items
      it has picked up; 
 \item The time difference between moving from location 8 to location
   2 and moving from  location 2 to location 4 must be within the interval
   \textcolor{red}{[2,3]};
   \item Similarly, the time difference between moving from location 2
     to location 4 and moving from location 4 to location 6 must be within 
     the interval \textcolor{red}{[2,3]};
 \item It can elapse time in
     the interval \textcolor{red}{[0,1]} at locations 3, 5 and, 6; and 
   \item it can not elapse time in any other location.
     \end{enumerate}
    
Thanks to conditions 1 and 2,  the robot must visit
locations 7,3 and, 5 equal number of times. We capture this succinctly
using 2 stacks; note that in case we know the global time $T$ a
priori, this maze can possibly be
modeled without any stacks, incurring a blow up in state space and we
need different modeling for different values of $T$. 
In the general case when $T$ is not fixed, we were able to model it using two
stacks and the value of $T$ can be changed without any modification of
the model.
\paragraph{Modeling} We model this maze using a timed multi-stack
pushdown automata (\tmpda{}) cf. Fig~\ref{fig:maze}.  2 clocks $x, y$
are used to capture the conditions 6--9; conditions 1--3 are captured
using the notion of age in the stacks. Indeed, at the cost of using
more clocks, incurring an increase in the state space,  we can capture
all time constraints without using  age constraints. The capability of
\trim{} to handle time on stacks aids in this succinct
representation. 
 \begin{wrapfigure}[13]{l}{.5\linewidth}
       \vspace{-.5cm}
  \centering
  \includegraphics[scale=.15]{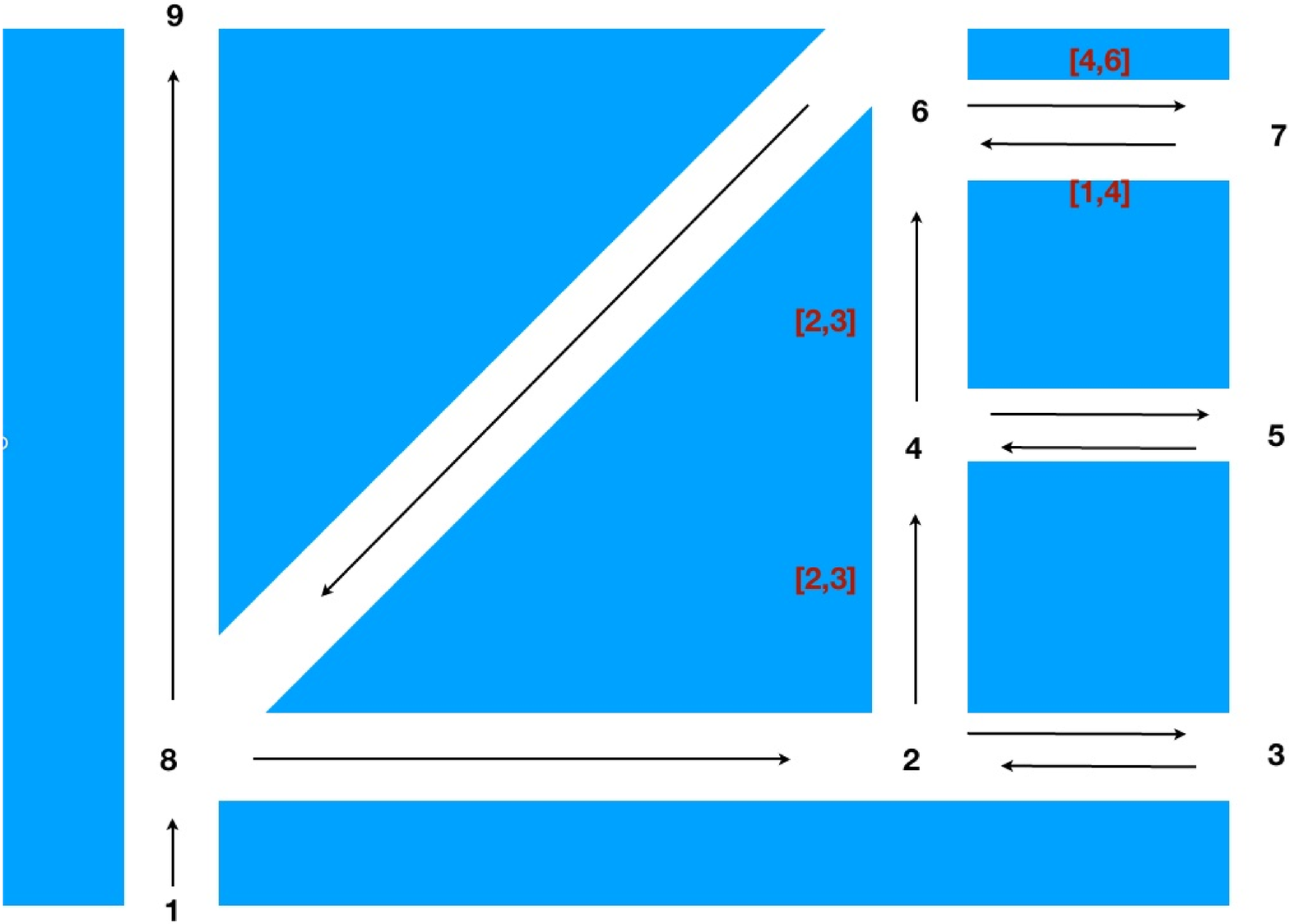}
  \caption{Maze}\label{fig:maze}
\end{wrapfigure}

To satisfy the above conditions, the robot must spend time within 2
  to 3 units  between the transitions from location 8 to location 2 and
   from location 2 to location 4. Similarly it must spend time within
   2 to 3 units between the transitions from location 2 to location 4
   and from location 4 to location 6. So, the
  minimum time it must spend is 4 units while travelling from location 8
  to location 6 in the maze.  Because the robot can only spend time in
  location 3,5 and 6, it will spend the time while picking items from
  location 3 and 5 but it can not spend more than 1 time unit in those locations. So, to
  satisfy the timing constraints, the robot will be forced to pick up two
  items by visiting locations 3 and 5 respectively. Then these items will be dropped in
  location 7 in specific order (first it will drop the location 3's item
  then it will drop location 5's item). We simulate the picking up
  using  push operations and dropping using pop operations. To satisfy the
  stack constraint mentioned in 3, the time difference between picking
  up an element in location 3 and dropping it at location 7 must be
  within the interval [4,6]. But, if the robot spend only 4 time units
  moving from location 8 to location 6, the second item collected from
  location 3 will only have age 3 when the robot reaches location
  6. So, the robot must spend 1 time unit in location 6 to satisfy the
  stack constraint. Moreover, if we project out the timing
  constraints, even if the robot picks one item from location 3 and 5
  respectively, it has to drop them in location 7, in the same order.
  This will not be a well-nested run.
  
Minimum value of $T$ that satisfies the above mentioned conditions is
5, with $T=5$, we get the following run,

  \[ t_1\-- t_2\-- t_4\-- TL(1)\-- t_5\-- t_4\-- TL(1)\-- t_5\-- t_6\-- t_9\-- TL(1)\-- t_{10}\-- t_9\-- TL(1)\-- t_{10}\-- t_{11}\-- \]
  \[TL(1)\--t_{14}\-- t_{15}\-- t_{14}\-- t_{15}\-- t_{16}\-- t_{19} \]
  Note that, $TL(\tau)$ represents time elapse of $\tau$
units and the $t_i$ represents transitions of the \tmpda{} cf. Fig~\ref{fig:maze}.

\paragraph{Hole bound} As there is no well-nested sequence possible in
the \tmpda{} for any $T>0$, from Theorem~\ref{thm:1}, there will not be any run with
0 or 1 hole bound. But, it is possible to find the run 
with $T=5$ as shown above has hole bound $2$. If we remove the timing
constraints from the stack as well as the clocks, then the robot will
not be forced to go inside the maze, in fact it can directly move to
location 9 from location 1 after visiting 8. The witness can be obtained with 0
hole-bound. The number of context changes in the above mentioned run of
timed maze is,  6 which is equal to the bound on scope.

\section{ Varying Parameters of \trim{}  }\label{sec:scalability}
We experimented by varying different
parameters like number of \emph{holes}, \cmax{} (for timed setting), number of locations, and number of transitions. Which gives us an excellent idea about the
performance of \trim{}. Here we are presenting the graphs generated by
varying the parameters for the example $L_{crit}$ by modifying the
\mpda{} \acritn{} to \acrit{} by
making the set of accepting states to empty. For the timed case we
call the \mpda{} \atcrit{}. Recall that we chose
\acrit{} because,
according to Algorithm~\ref{alg:isempty} if the language accepted
by the given \mpda{}(\tmpda{}) M is empty, the algorithm will generate all
possible lists $\mu_i$, recall by $\mu_i$ we represent the nodes of the BFS
exploration tree.  The number of such $\mu_i$ will depend on
  the number of open \emph{holes} K. As we already discussed that the number of such
  $\mu_i$ is finite; hence, the algorithm will terminate. So, when we
  run the algorithm with the input \acrit{} the
  algorithm will check all possible $\mu_i$ for an accepting run. In
  other words, this captures the worst case running time of the tool
  \trim{}. First we will show the performance of \trim{} in the un-timed
  settings and then we will redo the experiments on \atcrit{}.
  
\paragraph{\textbf{Un-timed:}} The graphs for un-timed system showing
the running time of \trim{} is in 
Figures~\ref{fig:graph-hole-l-crit-empty},~\ref{fig:graph-location-transitions-l-crit-empty}. And
the memory requirement for the same is shown in Figures~\ref{fig:graph-hole-l-crit-empty-memory},~\ref{fig:graph-location-transitions-l-crit-empty-memory}.

It is clearly seen in  Figure~\ref{fig:graph-hole-l-crit-empty}
that the time is exponential with the holes, and linear with the
number of locations as seen in Figure~\ref{fig:graph-location-transitions-l-crit-empty}.

\begin{figure}[h!]
 
    \begin{minipage}{.48\textwidth}
    \centering
  
  \begin{tikzpicture}[scale=.7]
  \begin{axis}[
    xlabel=Holes,
    ylabel=Run Time(s),
    xtick=data,
    enlarge y limits={upper, value=0.2}]
    \addplot[smooth,color=red,nodes near coords] plot coordinates {
       (4,0.3723)
      (5,0.6749)
      (6,2.5988)
      (7,4.7267)
      (8,15.681)
      (9,25.4822)
      (10,80.5287)

    };
     \addplot+[ybar,color=blue,fill] plot coordinates {
             (1,0.0110)
      (2,0.0371)
      (3,0.0891)
      (4,0.3723)
      (5,0.6749)
      (6,2.5988)
      (7,4.7267)
      (8,15.681)
      (9,25.4822)
      (10,80.5287)
    };
  \end{axis}

\end{tikzpicture}
  \caption{Performance of \trim{} with the number of \emph{holes}
    varying on input \acrit{}}\label{fig:graph-hole-l-crit-empty} 
\end{minipage}\hfill
 \begin{minipage}{.48\textwidth}
    \centering
    
  \begin{tikzpicture}[scale=.7]
  \begin{axis}[
    xlabel=Holes,
    ylabel=Memory(MB),
    xtick=data,
    enlarge y limits={upper, value=0.2}]
    \addplot[smooth,color=red,nodes near coords] plot coordinates {
       (1,3.9)
       (2,4)
       (3,4.7)
       (4,9.8)
      (5,17.8)
      (6,38.9)
      (7,75.6)
      (8,219.8)
      (9,377.7)
      (10,1127.6)

    };
    \addplot+[ybar,color=blue,fill] plot coordinates {
             (1,3.9)
       (2,4)
       (3,4.7)
       (4,9.8)
      (5,17.8)
      (6,38.9)
      (7,75.6)
      (8,219.8)
      (9,377.7)
      (10,1127.6)

    };
  \end{axis}

\end{tikzpicture}
  \caption{Memory consumption of \trim{} with the number of \emph{holes}
    varying on input \acrit{}}\label{fig:graph-hole-l-crit-empty-memory}
\end{minipage}
  \end{figure}

  \begin{figure}[h!]
    \begin{minipage}{.48\textwidth}
    \centering
\begin{tikzpicture}[scale=.7]
  \begin{axis}[
    xlabel=Locations,
    ylabel=Run Time(s),
    xtick=data,
    enlarge y limits={upper, value=0.2}]
    \addplot[smooth,color=red,nodes near coords] plot coordinates {
      (6,0.3723)
       (7,0.3853)
       (8,0.4308)
       (9,0.4718)
      (10,0.56005)
      (11,0.57922)
      (12,0.621077)
      (13,0.6925)
      (14,0.7468)
      (15,0.7707)

    };
    \addplot+[ybar,color=blue,fill] plot coordinates {
            (6,0.3723)
       (7,0.3853)
       (8,0.4308)
       (9,0.4718)
      (10,0.56005)
      (11,0.57922)
      (12,0.621077)
      (13,0.6925)
      (14,0.7468)
      (15,0.7707)
    };
  \end{axis}

\end{tikzpicture}
  \caption{Performance of \trim{} with the number of locations
    varying on input \acrit{}.}\label{fig:graph-location-transitions-l-crit-empty}
\end{minipage}\hfill
\begin{minipage}{.48\textwidth}
    \centering
\begin{tikzpicture}[scale=.7]
  \begin{axis}[
    xlabel=Locations,
    ylabel=Memory(MB),
    xtick=data,
    enlarge y limits={upper, value=0.2}]
    \addplot[smooth,color=red,nodes near coords] plot coordinates {
             (6,9.9)
       (7,10)
       (8,10.1)
       (9,10.2)
      (10,11.3)
      (11,12.7)
      (12,13.8)
      (13,14.0)
      (14,15.8)
      (15,16.9)

    };
    \addplot+[ybar,color=blue,fill] plot coordinates {
             (6,9.9)
       (7,10)
       (8,10.1)
       (9,10.2)
      (10,11.3)
      (11,12.7)
      (12,13.8)
      (13,14.0)
      (14,15.8)
      (15,16.9)

    };
  \end{axis}

\end{tikzpicture}
  \caption{Memory consumption of \trim{} with the number of locations
    varying on input
   \acrit{}.}\label{fig:graph-location-transitions-l-crit-empty-memory}
\end{minipage}
  \end{figure}
  
\paragraph{\textbf{Timed:}}  Figure~\ref{fig:graph-hole-time-l-crit-empty} shows 
  the result when we vary the number of \emph{holes} allowed in a
  given run. When we run the experiments by varying the maximum constant in the
\tmpda{}, we get the results in Fig.~\ref{fig:graph-cmax-time-l-crit-empty}.  Here we can
see the time increases polynomially with $c_{max}$  and
exponentially   with the number of \emph{holes} as expected. 

\begin{figure}[h]

\begin{minipage}{.48\textwidth}
  \centering

\begin{tikzpicture}[scale=.7]
  \begin{axis}[
    xlabel=Holes,
    ylabel=Run Time(s),
    xtick=data,
    enlarge y limits={upper, value=0.2}]
    
    \addplot[smooth,color=red,nodes near coords] plot coordinates {
      (2,9.05609)
      (3,46.0823)
      (4,404.218)
      (5,2726.01)

    };
     \addplot+[ybar,color=blue,fill] plot coordinates {
      (2,9.05609)
      (3,46.0823)
      (4,404.218)
      (5,2726.01)

    };
  \end{axis}

\end{tikzpicture}
  \caption{Performance of \trim{} with the number of \emph{holes}
    varying on \atcrit{} as input}\label{fig:graph-hole-time-l-crit-empty}

\end{minipage}\hfill
 \begin{minipage}{.48\textwidth}
  \centering

\begin{tikzpicture}[scale=.7]
  \begin{axis}[
    xlabel=Holes,
    ylabel=Memory(MB),
    xtick=data,
    enlarge y limits={upper, value=0.2}]
    \addplot[smooth,color=red,nodes near coords] plot coordinates {
      (2,54.108)
      (3,244.548)
      (4,2436.689)
      (5,11415.860)

    };
     \addplot+[ybar,color=blue,fill] plot coordinates {
       (2,54.108)
      (3,244.548)
      (4,2436.689)
      (5,11415.860)

    };
  \end{axis}

\end{tikzpicture}
  \caption{Memory consumption of \trim{} with the number of \emph{holes}
    varying on \atcrit{} as input}\label{fig:graph-hole-time-l-crit-empty-memory}
\end{minipage}

\end{figure}

\begin{figure}[h]

 \begin{minipage}{.48\textwidth}
  \centering
  \begin{tikzpicture}[scale=.7]
  \begin{axis}[
    xlabel=\cmax{},
    ylabel=Run Time(s),
    enlarge y limits={upper, value=0.2}]
    \addplot[smooth,color=red,nodes near coords] plot coordinates {
      (4,9.0560)
      (6,24.0387)
      (8,60.6463)
      (10,132.648)
      (12,261.001)
      (14,506.806)

    };
     \addplot+[ybar,color=blue,fill] plot coordinates {
      (4,9.0560)
      (6,24.0387)
      (8,60.6463)
      (10,132.648)
      (12,261.001)
      (14,506.806)

    };
  \end{axis}

\end{tikzpicture}
\caption{Performance of \trim{} with \cmax{} varying on
 \atcrit{} as input}\label{fig:graph-cmax-time-l-crit-empty}
\end{minipage}\hfill
  \begin{minipage}{.48\textwidth}
  \centering
  \begin{tikzpicture}[scale=.7]
  \begin{axis}[
    xlabel=\cmax{},
    ylabel=Memory(MB),
    enlarge y limits={upper, value=0.2}]
    \addplot[smooth,color=red,nodes near coords] plot coordinates {
      (4,54.108)
      (6,117.076)
      (8,406.236)
      (10,811.124)
      (12,1624.132)
      (14,3252.700)

    };
     \addplot+[ybar,color=blue,fill] plot coordinates {
        (4,54.108)
      (6,117.076)
      (8,406.236)
      (10,811.124)
      (12,1624.132)
      (14,3252.700)

    };
  \end{axis}

\end{tikzpicture}
\caption{Memory consumption of \trim{} with \cmax{} varying on
\atcrit{}
  as input}\label{fig:graph-cmax-time-l-crit-empty-memory}

\end{minipage}
\end{figure}

Finally, 
when we run the experiments by varying the maximum constant
(\cmax{}) in the \tmpda{}, we get the results as showed in  Fig.~\ref{fig:graph-cmax-time-l-crit-empty}. We can 
see increasing the \cmax{} increases
the time as a polynomial function. However, when we run experiments by varying the number of clocks, we get the
result showed in Fig.~\ref{fig:graph-clock-time-l-crit-empty} and the Graph in
indicates that time grows exponentially when the number of clocks
increased. For both cases memory consumption is shown in
Fig.~\ref{fig:graph-cmax-time-l-crit-empty-memory} and
Fig.~\ref{fig:graph-clock-time-l-crit-empty-memory} respectively.

We also checked the performance by varying the number of locations of
\tmpda{}. As expected, the time increased polynomially with the increase
of locations in the \tmpda{} c.f. Graph in
Fig.~\ref{fig:graph-location-time-l-crit-empty} and for memory in Fig.~\ref{fig:graph-location-time-l-crit-empty-memory}
\begin{figure}[h]
 
  \begin{minipage}{.48\textwidth}
   \centering
    \begin{tikzpicture}[scale=.7]
  \begin{axis}[
    xlabel=Clocks,
    ylabel=Run Time(s),
    enlarge y limits={upper, value=0.2}]
    \addplot[smooth,color=red,nodes near coords] plot coordinates {
      (2,9.0560)
      (3,14.1044)
      (4,21.9928)
      (5,39.1363)
      (6,87.4203)
      (7,264.72)

    };
     \addplot+[ybar,color=blue,fill] plot coordinates {
    (2,9.0560)
      (3,14.1044)
      (4,21.9928)
      (5,39.1363)
      (6,87.4203)
      (7,264.72)

    };
  \end{axis}

\end{tikzpicture}
\caption{Performance of \trim{} with the number of clocks  varying on
  \atcrit{} as input.}\label{fig:graph-clock-time-l-crit-empty}
\end{minipage}\hfill
  \begin{minipage}{.48\textwidth}
   \centering
    \begin{tikzpicture}[scale=.7]
  \begin{axis}[
    xlabel=Clocks,
    ylabel=Memory(MB),
    enlarge y limits={upper, value=0.2}]
    \addplot[smooth,color=red,nodes near coords] plot coordinates {
      (2,54.108)
      (3,55.124)
      (4,65.248)
      (5,84.440)
      (6,153.572)
      (7,401.572)

    };
     \addplot+[ybar,color=blue,fill] plot coordinates {
    (2,54.108)
      (3,55.124)
      (4,65.248)
      (5,84.440)
      (6,153.572)
      (7,401.572)

    };
  \end{axis}

\end{tikzpicture}
\caption{Memory consumption of \trim{} with the number of clocks  varying on
  \atcrit{} as input.}\label{fig:graph-clock-time-l-crit-empty-memory} 
\end{minipage}

\end{figure}

\begin{figure}[h]

 \begin{minipage}{.48\textwidth}
   \centering
   \begin{tikzpicture}[scale=0.7]
  \begin{axis}[
    xlabel=Locations,
    ylabel=Run Time(s),
    enlarge y limits={upper, value=0.2}]
    \addplot[smooth,color=red,nodes near coords] plot coordinates {
           (6,9.0560)
      (7,10.2249)
      (8,11.6579)
      (9,12.7217)
      (10,14.913)
      (11,15.6878)

    };
     \addplot+[ybar,color=blue,fill] plot coordinates {
         (6,9.0560)
      (7,10.2249)
      (8,11.6579)
      (9,12.7217)
      (10,14.913)
      (11,15.6878)

    };
  \end{axis}

\end{tikzpicture}
\caption{Performance of \trim{} with the number of locations  varying
  on \atcrit{} as input.}\label{fig:graph-location-time-l-crit-empty} 
  \end{minipage}\hfill
   \begin{minipage}{.48\textwidth}
   \centering
   \begin{tikzpicture}[scale=0.7]
  \begin{axis}[
    xlabel=Locations,
    ylabel=Memory(MB),
    enlarge y limits={upper, value=0.2}]
    \addplot[smooth,color=red,nodes near coords] plot coordinates {
           (6,54.108)
      (7,55.192)
      (8,61.212)
      (9,74.732)
      (10,84.176)
      (11,93.660)

    };
     \addplot+[ybar,color=blue,fill] plot coordinates {
          (6,54.108)
      (7,55.192)
      (8,61.212)
      (9,74.732)
      (10,84.176)
      (11,93.660)

    };
  \end{axis}

\end{tikzpicture}
\caption{Memory consumption of \trim{} with the number of locations  varying
  on \atcrit{} as input.}\label{fig:graph-location-time-l-crit-empty-memory} 
  \end{minipage}
 \end{figure}
\subsection{Varying parameters of $L_{bh}$}

We also did the same for $L_{bh}$ example, we constructed the \mpda{}
$M_{L_{bh}}$ which accepts the language $L_{bh}$ and then modified
$M_{L_{bh}}$ to make its set of final states empty, so that \trim{}
generates all possible list $(\mu)$ and returns emptiness. As we
discussed earlier this pushes \trim{} to its worst case
complexity. Details can be found in the Tables~\ref{tab-l-bh-untimed-empty-holes},~\ref{tab-l-bh-untimed-empty-locations} and in
the Graphs~\ref{fig:graph-hole-l-bh-empty},~\ref{fig:graph-location-transitions-l-bh-empty}.

\begin{figure}
  \begin{minipage}{.48\textwidth}
  \centering
  \begin{tikzpicture}[scale=.7]
  \begin{axis}[
    xlabel=Holes,
    ylabel=Run Time(s),
    enlarge y limits={upper, value=0.2}]
    \addplot[smooth,color=red,nodes near coords] plot coordinates {
      (0,0.0134)
      (1,0.0156)
      (2,0.2066)
      (3,1.6664)
      (4,11.6436)
      (5,55.2659)

    };
     \addplot+[ybar,color=blue,fill] plot coordinates {
            (0,0.0134)
      (1,0.0156)
      (2,0.2066)
      (3,1.6664)
      (4,11.6436)
      (5,55.2659)
    };
  \end{axis}

\end{tikzpicture}
  \caption{Performance of \trim{} with the number of \emph{holes}
    varying as in
    Table~\ref{tab-l-bh-untimed-empty-holes}}\label{fig:graph-hole-l-bh-empty}
\end{minipage}\hfill
\begin{minipage}{.48\textwidth}
    \centering
\begin{tikzpicture}[scale=.7]
  \begin{axis}[
    xlabel=Locations,
    ylabel=Run Time(s),
    enlarge y limits={upper, value=0.2}]
    \addplot[smooth,color=red,nodes near coords] plot coordinates {
     
       (7,0.2099)
       (8,0.2549)
       (9,0.2634)
      (10,0.3059)
      (11,0.3265)
      (12,0.3797)
      (13,0.4059)
      (14,0.4234)

    };
    \addplot+[ybar,color=blue,fill] plot coordinates {
       (7,0.2099)
       (8,0.2549)
       (9,0.2634)
      (10,0.3059)
      (11,0.3265)
      (12,0.3797)
      (13,0.4059)
      (14,0.4234)
      
    };
  \end{axis}

\end{tikzpicture}
  \caption{Performance of \trim{} with the number of locations  varying as in Table~\ref{tab-l-bh-untimed-empty-locations}}\label{fig:graph-location-transitions-l-bh-empty}
\end{minipage}
  \end{figure}
  \subsection{Varying parameters on non-empty \mpda{}}

   \subsubsection{L\textsup{crit}} This example is one of the most simple
    examples that we worked on, and the parameters of this example,
    like c\textsub{max}, Clocks can easily be changed to see the impact on
    running time. Below we present some tables generated by varying
    the parameters.
    \subsubsection{Varying Maximum Constant(c\textsub{max})}
    Table~\ref{tab:anbmcndm-max-const} shows the run time and memory
    consumption by the algorithm when the maximum constant(\cmax{}) is
    varied. Note, that the automaton is non-empty so it has a run, and
    as soon as it finds a path from the start state to the final state
    it stops and produces an accepting run using witness algorithm as
    described in the main paper.
 \subsubsection{Varying Holes} Changing \emph{holes} should increase time
exponentially but, in this example as the automaton has a run with
only 2 \emph{holes} so it terminates when ever it finds the final
state. Even if we allow more \emph{holes} the program terminates
before introducing more \emph{holes}, which is why time
saturates after 3 \emph{holes}. Look at Table \ref{tab:anbmcndm-holes}
for details.
\subsubsection{Increasing Clocks} Though the
  increasing number of clocks in this current model makes no sense,
  due to the lack of proper benchmarks, we had to find some way to
  understand the scalability of our algorithm so, we randomly increased
  the number of clocks with a fixed value of clock constraint, so that we
  get some idea about how the algorithm scales when the clocks are
  increased for a given example.
Increasing clocks increases the number of states exponentially. So, every time we add a clock, the number of state valuation pairs multiplies according to the maximum constant of that clock, which in turn increases the time to compute transitive closure. On the other hand, the number of lists depends on the state valuation pairs, which are reachable from the previous state valuation pairs by some transitions. The addition of an extra clock does not change the number of reachable pairs in this example. Hence, even if the number of state valuation pairs increased, the number of lists $\mu$ generated remain
 same s.f. Table \ref{tab:anbmcndm-clock}
\begin{figure}
  \begin{minipage}{.48\textwidth}
  \centering
\begin{tikzpicture}[scale=0.7]
  \begin{axis}[
    xlabel=Holes,
    ylabel=Run Time(s),
    enlarge y limits={upper, value=0.2}]
    \addplot[smooth,color=red,nodes near coords] plot coordinates {
      ( 2 , 9.9652)
      (3,  23.2813)
      (4, 30.4787)
      (5,  30.4565)
      (6, 30.7647)
      (7, 30.2986)
      (8, 30.5274)

    };
     \addplot+[ybar,color=blue,fill] plot coordinates {
      ( 2 , 9.9652)
      (3,  23.2813)
      (4, 30.4787)
      (5,  30.4565)
      (6, 30.7647)
      (7, 30.2986)
      (8, 30.5274)

    };
  \end{axis}

\end{tikzpicture}
\caption{Performance of \trim{} with varying number of holes on $\mathcal{A}_{crit}$
     (as in Table~\ref{tab:anbmcndm-holes})}
\end{minipage}\hfill
\begin{minipage}{.48\textwidth}
   \centering
\begin{tikzpicture}[scale=.7]
  \begin{axis}[
    xlabel=c\textsubscript{max},
    ylabel=Run Time(s),
    enlarge y limits={upper, value=0.2}]
    \addplot[smooth,color=red,nodes near coords] plot coordinates {
      ( 4 , 0.9326)
      (5, 1.86942)
      (6,3.6386)
      (7, 6.3043)
      (8,9.6519)
      (9,15.9538)
      (10,24.3861)

    };
     \addplot+[ybar,color=blue,fill] plot coordinates {
     ( 4 , 0.9326)
      (5, 1.86942)
      (6,3.6386)
      (7, 6.3043)
      (8,9.6519)
      (9,15.9538)
      (10,24.3861)

    };
  \end{axis}

\end{tikzpicture}
\caption{Performance of \trim{} with the c\textsubscript{max}
    varying on $\mathcal{A}_{crit}$ (as in Table~\ref{tab:anbmcndm-max-const})}
  \end{minipage}
  
\end{figure}

\begin{figure}[h!]
 \begin{minipage}{.48\textwidth}
  \centering
\begin{tikzpicture}[scale=.7]
  \begin{axis}[
    xlabel=Clocks,
    ylabel=Run Time(s),
    enlarge y limits={upper, value=0.2}]
    \addplot[smooth,color=red,nodes near coords] plot coordinates {
      ( 2 ,  0.276977)
      (3,   0.362076)
      (4,  0.753226)
      (5,  2.84098)
      (6, 16.4543 )
      (7, 118.46)
      (8, 917.355)

    };
     \addplot+[ybar,color=blue,fill] plot coordinates {
       ( 2 ,  0.276977)
      (3,   0.362076)
      (4,  0.753226)
      (5,  2.84098)
      (6, 16.4543 )
      (7, 118.46)
      (8, 917.355)

    };
  \end{axis}

\end{tikzpicture}
  \caption{Performance of \trim{} with varying number of clocks
    on \acritn{}   (as in Table~\ref{tab:anbmcndm-clock})}
\end{minipage}\hfill
\begin{minipage}{.48\textwidth}

  \centering
\begin{tikzpicture}[scale=.7]
  \begin{axis}[
    xlabel={LCM(M,N)},
    ylabel=Run Time(s),
    enlarge y limits={upper, value=0.2}]
    \addplot[smooth,color=red,nodes near coords] plot coordinates {

      (12,  172.542)
      (45,  716.35)
      (66, 957.338 )
      (84, 1194.28)
      (168, 1922.49)

    };
     \addplot+[ybar,color=blue,fill] plot coordinates {
       ( 6 ,  130.85)
      (10,   144.732)
      (12,  172.542)
      (45,  716.35)
      (66, 957.338 )
      (84, 1194.28)
      (168, 1922.49)

    };
  \end{axis}
\end{tikzpicture}
  \caption{\aprodp{}  with varying parameters `M' and `N'
on \trim{} shows that the run time is proportional to the LCM(M,N)
(as in Table~\ref{tab:para-prodcon})}
  \end{minipage}
\end{figure}

\subsection{Multi Producer Consumer Problem}
Multi producer consumer problem is a parameterized example (the size
of the production batches of `A' and `B' denoted as `M' and `N' respetively). We run \trim{} with
different parameter values of \aprodp{} and we show the results in the
table~\ref{tab:para-prodcon}.

    \begin{table}[t]
    \resizebox{\textwidth}{!}{
    \centering
    \begin{tabular}{c|c|c|c|c|c|c|c}
      \toprule
M & N & States
& Transitions
& Length(Witness)
(transitions)
      & Time(ms) & Time(Witness)(ms) & LCM(M,N) \\
      \hline
      3 & 2 & 7 & 11 & 24 & 130.85 & 0.27 & 6 \\

\hline
9 & 5 & 22 & 25 & 180 & 716.35 & 4.944 & 45 \\
\hline
10 & 5 & 23 & 26 & 40 & 144.732 & 0.854 & 10 \\
\hline
  11 & 6 & 25 & 28 & 264 & 957.338 & 6.002 & 66 \\
  \hline
  12 & 6 & 26 & 29 & 48 & 172.542 & 1.008 & 12 \\
  \hline
  12 & 7 & 27 & 30 & 336 & 1194.28 & 8.071 & 84 \\
  \hline
      24 & 7 & 32 & 34 & 672 & 1922.49 & 10.767 & 168 \\
      \hline
      
      \bottomrule
      
\end{tabular}
}
\caption{Experimental results on the parameterized \aprodp{} example}
\label{tab:para-prodcon}
\end{table}

\section{Scalability}
We also tried to run scalability test of \trim{} by checking maximum
number of holes \trim{} can handle within a given memory limit. Here we
used the \mpda{} \acrit{}, and \atcrit{}.
Table~\ref{tab-l-crit-untimed-empty-holes-scale}, and
Table~\ref{tab:lcrit-bounded-ram}  shows the
scalability of \trim{}  on \acrit{}, and \atcrit{} with memory limit of 8GB.

 \begin{table}[h!]
\centering
\resizebox{\textwidth}{!}{
  \begin{tabular}{|c|c|c|c|c|c|c|c|}
   \hline  Stacks & Location & Transitions &  Holes & Time(sec) & Memory(MB)  \\
    \hline \hline
         2 & 6 & 9 &  0 & 0.015 & 3.889 \\
    \hline

     2 & 6 & 9 &  1 & 0.0110 & 3.892  \\
    \hline
     2 & 6 & 9 &  2 & 0.0371 & 4.096 \\
    \hline
     2 & 6 & 9 &  3 & 0.0891 & 4.788 \\
    \hline
    2 & 6 & 9 &  4 & 0.3723 & 9.960 \\
    \hline
         2 & 6 & 9 &  5 & 0.6749 & 17.776  \\
    \hline
    2 & 6 & 9 &  6 & 2.5988 &38.964  \\
    \hline
    2 & 6 & 9 &  7 & 4.7267 & 75.564  \\
    \hline
    2 & 6 & 9 &  8 & 15.681 & 219.844 \\
    \hline
    2 & 6 & 9 &  9 & 25.4822& 377.624  \\
    \hline
    2 & 6 & 9 &  10 & 80.5287 & 1127.9  \\
    \hline
     2 & 6 & 9 &  11 & 151.612 & 1937.5  \\
    \hline
     2 & 6 & 9 &  12 & 211.584 & 5468.2  \\
    \hline
     2 & 6 & 9 &  13 & 662.98 & 8GB(Killed)\\
    \hline
\end{tabular}}
\caption{\acrit{}  scalability with respect to holes and  memory limited to 8GB}\label{tab-l-crit-untimed-empty-holes-scale}
\end{table}

\begin{table}[h!]
  \centering
  \resizebox{\textwidth}{!}{
    \begin{tabular}{|c|c|c|c|c|c|c|c |}
      \hline
      
Clocks & Stack & Locations & Transitions &\cmax{} & Holes & Time(sec) &
                                                                      Memory(MB)  \\
      \hline\hline
      2 & 2 & 6 & 9 & 4 & 0 & 0.0802 & 5.7 \\
      \hline
      2 & 2 & 6 & 9 & 4 & 1 & 0.09 & 5.7 \\
      \hline
      2 & 2 & 6 & 9 & 4 & 2 & 67.126 & 40.6 \\
      \hline
      2 & 2 & 6 & 9 & 4 & 3 & 385.59 & 1980.220 \\
      \hline
      2 & 2 & 6 & 9 & 4 & 4 & 858.9 & 8GB(Killed)  \\
      \hline

\end{tabular} 
  }
  \caption{\atcrit{} scalability with respect to holes and memory
    limited to 8GB}
  \label{tab:lcrit-bounded-ram} 

\end{table}
\begin{figure}[h!]
  \begin{minipage}{.48\textwidth}
  \centering
  \begin{tikzpicture}[scale=.7]
  \begin{axis}[
    xlabel=Holes,
    ylabel=Run Time(s),
    xtick=data,
    enlarge y limits={upper, value=0.2}]
    \addplot+[ybar,color=blue,fill] plot coordinates {
           (0,0.015)
       (1,0.0110)
       (2,0.0371)
       (3,0.0891)
      (4,0.3723)
      (5,0.6749)
      (6,2.5988)
      (7,4.7267)
      (8,15.681)
      (9,25.4822)
      (10,80.5287)
      (11,151.612)
      (12,211.584)
      (13,662.98)
    };
    \addplot[smooth,color=red,nodes near coords] plot coordinates {
       (10,80.5287)
      (11,151.612)
      (12,211.584)
      (13,662.98)

    };
     
  \end{axis}

\end{tikzpicture}
  \caption{Scalability of \trim{} with the number of \emph{holes}
    varying with memory limited to 8GB(8192MB) on \acrit{} (as in
    Table~\ref{tab-l-crit-untimed-empty-holes-scale})}\label{fig:graph-hole-l-crit-time-scale}
\end{minipage}\hfill
\begin{minipage}{.48\textwidth}
    \centering
\begin{tikzpicture}[scale=.7]
  \begin{axis}[
    xlabel=Holes,
    ylabel=Memory(MB),
    xtick=data,
    enlarge y limits={upper, value=0.2}]
    \addplot+[ybar,color=blue,fill] plot coordinates {
           (0,3.889)
       (1,3.392)
       (2,4.096)
       (3,4.788)
      (4,9.960)
      (5,17.776)
      (6,38.964)
      (7,75.564)
      (8,219.844)
      (9,377.624)
      (10,1127.9)
      (11,1937.5)
      (12,5468.2)
      (13,8192) 
     
    };
    \addplot[smooth,color=red,nodes near coords] plot coordinates {
      (10,1127.9)
      (11,1937.5)
      (12,5468.2)
      (13,8192)

    };
    
  \end{axis}

\end{tikzpicture}
  \caption{Memory consumption of \trim{} with  number of \emph{ holes}
    varying and
    memory limited to 8GB(8196MB) on \acrit{} (as in Table~\ref{tab-l-crit-untimed-empty-holes-scale})}\label{fig:graph-memory-l-crit-scale}
\end{minipage}
  \end{figure}

\begin{figure}
  \begin{minipage}{.48\textwidth}
  \centering
  \begin{tikzpicture}[scale=.7]
  \begin{axis}[
    xlabel=Holes,
    ylabel=Run Time(s),
    enlarge y limits={upper, value=0.2}]
    \addplot[smooth,color=red,nodes near coords] plot coordinates {
       
      (0,0.0802)
      (1,0.09)
      (2,67.126)
      (3,385.59)
      (4,858.9)

    };
     \addplot+[ybar,color=blue,fill] plot coordinates {
             (0,0.0802)
      (1,0.09)
      (2,67.126)
      (3,385.59)
      (4,858.9)
      
    };
  \end{axis}

\end{tikzpicture}
  \caption{ Scalability of \trim{} with the number of \emph{holes}
    varying as  with memory limited to 8GB on
    \atcrit{} (as in
    Table~\ref{tab:lcrit-bounded-ram})}\label{fig:graph-atcrtit-bounded-ram-time}
\end{minipage}\hfill
\begin{minipage}{.48\textwidth}
    \centering
\begin{tikzpicture}[scale=.7]
  \begin{axis}[
    xlabel=Locations,
    ylabel=Memory(MB),
    enlarge y limits={upper, value=0.2}]
    \addplot[smooth,color=red,nodes near coords] plot coordinates {
       (0,5.7)
      (1,5.7)
      (2,40.6)
      (3,1980.22)
      (4,8192)

    };
    \addplot+[ybar,color=blue,fill] plot coordinates {
       (0,5.7)
      (1,5.7)
      (2,40.6)
      (3,1980.22)
      (4,8192)      
    };
  \end{axis}

\end{tikzpicture}
  \caption{Memory consumption of \trim{} with the number of
    \emph{holes}   varying  with memory limited to 8GB on
    \atcrit{} (as in Table~\ref{tab:lcrit-bounded-ram})}\label{fig:graph-atcrit-bounded-ram-memory}
\end{minipage}
  \end{figure}

  \paragraph{Scalability in terms of states:} We also tried to check the scalability of \trim{} by
  increasing the number of states of \mpda{}. We were able to handle
  4000 state \mpda{} with 0 holes, in 9181 sec and 4.7GB of memory.

\section{Tables}

    \begin{table}[h!]
    \centering
    \resizebox{\textwidth}{!}{
      \begin{tabular}{|c|c|c|c|c|c|c|c|c|}
        \hline
Clocks & Stack & Locations & Transitions & $c_{max}$ & Holes &
                                                                   Time(sec)
        & Memory(KB)&Empty(Y/N) \\
      \hline\hline
      2 & 2 & 6 & 9 & 4 & 2 & 0.9326 & 28804&N \\
      \hline
      2 & 2 & 6 & 9 & 5 & 2 & 1.8694 & 54056 & N \\
      \hline
      2 & 2 & 6 & 9 & 6 & 2 & 3.6386 & 103904 & N \\
      \hline
      2 & 2 & 6 & 9 & 7 & 2 & 6.3043 & 204560 & N \\
      \hline
      2 & 2 & 6 & 9 & 8 & 2 & 9.6519 & 303012 &  N \\
      \hline
      2 & 2 & 6 & 9 & 9 & 2 & 15.9538 & 404908 &  N \\
      \hline
      2 & 2 & 6 & 9 & 10 & 2 & 24.3861 & 811040 &  N \\
      \hline
    \end{tabular}
      
  }
  \caption{$L_{crit}$ With Changing Maximum Constant}~\label{tab:anbmcndm-max-const}
  \end{table}
\begin{table}[h!]
  \centering
  \resizebox{\textwidth}{!}{
    \begin{tabular}{|c|c|c|c|c|c|c|c|c|}
      \hline
Clocks & Stack & Locations & Transitions & $c_{max} $& Holes & Time(sec) &
                                                                      Memory(KB)& Empty(Y/N) \\
      \hline\hline
      2 & 2 & 6 & 9 & 8 & 2 & 9.9652 & 203396&N \\
      \hline
      2 & 2 & 6 & 9 & 8 & 3 & 23.2813 & 414100 & N \\
      \hline
      2 & 2 & 6 & 9 & 8 & 4 & 30.4787 & 443064 & N \\
      \hline
      2 & 2 & 6 & 9 & 8 & 5 & 30.4565 & 443192 & N \\
      \hline
      2 & 2 & 6 & 9 & 8 & 6 & 30.7647 & 443236 & N \\
      \hline
      2 & 2 & 6 & 9 & 8 & 7 & 30.2986 & 443120 & N \\
      \hline
      2 & 2 & 6 & 9 & 8 & 8 & 30.5274 & 443188 & N \\
      \hline
\end{tabular} 
  }
  \caption{$L_{crit}$ With Changing Holes}
  \label{tab:anbmcndm-holes}
\end{table}

\begin{table}[h!]
  \centering
  \resizebox{\textwidth}{!}{
    \begin{tabular}{|c|c|c|c|c|c|c|c |c|}
      \hline
      
Clocks & Stack & Locations & Transitions &$c_{max}$ & Holes & Time(sec) &
                                                                      Memory(KB) & Empty(Y/N) \\
      \hline\hline
      2 & 2 & 6 & 9 & 4 & 2 & 0.276977 & 7336 & N \\
      \hline
      3 & 2 & 6 & 9 & 4 & 2 & 0.362076 & 8268 & N \\
      \hline
      4 & 2 & 6 & 9 & 4 & 2 & 0.753226 & 12880 & N \\
      \hline
      5 & 2 & 6 & 9 & 4 & 2 & 2.84098 & 29020 & N \\
      \hline
      6 & 2 & 6 & 9 & 4 & 2 & 16.4543 & 104832 &  N \\
      \hline
      7 & 2 & 6 & 9 & 4 & 2 & 118.46 & 397944 &  N \\
      \hline
      8 & 2 & 6 & 9 & 4 & 2 & 917.355 & 1580424   & N \\
      \hline
\end{tabular} 
  }
  \caption{$L_{crit}$ With Changing Clocks}
  \label{tab:anbmcndm-clock}

\end{table}

\begin{table}[h!]
\centering
\resizebox{\textwidth}{!}{
  \begin{tabular}{|c|c|c|c|c|c|c|c|c|}
   \hline  Stacks & Location & Transitions &  Holes & Time(sec) & Memory(KB) &  Empty(Y/N) \\
    \hline \hline
     2 & 7 & 13 &  0 & 0.0134& 4128 &  Y \\
    \hline
     2 & 7 & 13 &  1 & 0.0156 & 4200 & Y \\
    \hline
     2 & 7 & 13 &  2 & 0.2099 & 7300 &  Y \\
    \hline
    2 & 7 & 13 &  3 & 1.6664 & 33408 &  Y \\
    \hline
         2 & 7 & 13 &  4 & 11.6436 & 297656 &  Y \\
    \hline
    2 & 7 & 13 &  5 & 55.2659 &1408160  &  Y \\
   
    \hline
\end{tabular}}
\caption{$M_{L_{bh}}$ Varying Holes}\label{tab-l-bh-untimed-empty-holes}
\end{table}

\begin{table}[h!]
\centering
\resizebox{\textwidth}{!}{
  \begin{tabular}{|c|c|c|c|c|c|c|c|c|}
   \hline  Stacks & Location & Transitions &  Holes & Time(sec) & Memory(KB) &  Empty(Y/N) \\
    \hline \hline
     2 & 7 & 13 &  2 & 0.2099& 7300 &  Y \\
    \hline
     2 & 8 & 14 &  2 & 0.2549 & 9804 & Y \\
    \hline
     2 & 9 & 15 &  2 & 0.2634 & 9904 &  Y \\
    \hline
    2 & 10 & 16 &  2 & 0.3059 & 9996 &  Y \\
    \hline
         2 & 11 & 17 &  2 & 0.3255 & 10600 &  Y \\
    \hline
    2 & 12 & 18 &  2 & 0.3797 &10164 &  Y \\
   \hline
    2 & 13 & 19 &  2 & 0.4059 &15736 &  Y \\
    \hline
    2 & 14 & 20 &  2 & 0.4234 &15876 &  Y \\
    \hline
\end{tabular}}
\caption{$M_{L_{bh}}$ Varying Locations}\label{tab-l-bh-untimed-empty-locations}
\end{table}
 \begin{table}[h!]
\centering
\resizebox{\textwidth}{!}{
  \begin{tabular}{|c|c|c|c|c|c|c|c|c|}
   \hline  Stacks & Location & Transitions &  Holes & Time(sec) & Memory(KB) &  Empty(Y/N) \\
    \hline \hline
     2 & 6 & 9 &  4 & 0.3723 & 9960 &  Y \\
    \hline
     2 & 7 & 10 &  4 & 0.3853 & 9996 & Y \\
    \hline
     2 & 8 & 11 &  4 & 0.4308 & 10092 &  Y \\
    \hline
    2 & 9 & 12 &  4 & 0.4718 & 10160 &  Y \\
    \hline
         2 & 10 & 13 &  4 & 0.5600 & 11336 &  Y \\
    \hline
    2 & 11 & 14 &  4 & 0.57922 &12763  &  Y \\
    \hline
    2 & 12 & 15 &  4 & 0.62017 & 13808 &  Y \\
    \hline
    2 & 13 & 16 &  4 & 0.692599 & 14028 &  Y \\
    \hline
    2 & 14 & 17 &  4 & 0.746832& 15716 &  Y \\
    \hline
    2 & 15 & 18 &  4 & 0.770715 & 16918 &  Y \\
    \hline
\end{tabular}}
\caption{\acrit{} Varying Locations-Transitions}\label{tab-l-crit-untimed-empty-location-transitions}
\end{table}
\begin{table}[h]
  \centering
  \resizebox{\textwidth}{!}{
    \begin{tabular}{|c|c|c|c|c|c|c|c|c|}
      \hline
Clocks & Stacks & Location & Transitions & \cmax{} & Holes & Time(sec) & Memory(KB) & Empty(Y/N) \\
      \hline\hline
      2 & 2 & 6 & 9 & 4 & 2 & 9.05609 & 54108  & Y \\
      \hline
      2 & 2 & 6 & 9 & 6 & 2 & 24.0387 & 117076 & Y \\
      \hline
      2 & 2 & 6 & 9 & 8 & 2 & 60.6463 & 406236  & Y \\
      \hline
      2 & 2 & 6 & 9 & 10 & 2 & 132.648 & 811124  & Y \\
      \hline
      2 & 2 & 6 & 9 & 12 & 2 & 261.001 & 1624132 & Y \\
      \hline
      2 & 2 & 6 & 9 & 14 & 2 & 506.806 & 3252700 & Y \\
      \hline
    \end{tabular} }
  \caption{\atcrit{} Varying \cmax{}}\label{tab:l-crit-cmax}
\end{table}

\begin{table}[h]
  \centering
  \resizebox{\textwidth}{!}{
    \begin{tabular}{|c|c|c|c|c|c|c|c|c|}
      \hline
Clocks & Stacks & Location & Transitions & \cmax{} & Holes & Time(sec) & Memory(KB) &  Empty(Y/N) \\
      \hline\hline
      2 & 2 & 6 & 9 & 4 & 2 & 9.05609 & 54108 & Y \\
      \hline
      3 & 2 & 6 & 9 & 4 & 2 & 14.1044 & 55124 &  Y \\
      \hline
      4 & 2 & 6 & 9 & 4 & 2 & 21.9928 & 65248 & Y \\
      \hline
      5 & 2 & 6 & 9 & 4 & 2 & 39.1368 & 84440 &  Y \\
      \hline
      6 & 2 & 6 & 9 & 4 & 2 & 87.4203 & 153572 &  Y \\
      \hline
      7 & 2 & 6 & 9 & 4 & 2 & 264.71 & 401572 &  Y \\
      \hline
    \end{tabular}
     }
  \caption{\atcrit{} Varying Clocks}
  \label{tab:l-crit-clock}
\end{table}
\begin{table}[h!]
  \centering
  \resizebox{\textwidth}{!}{
    \begin{tabular}{c|c|c|c|c|c|c|c|c}
      \toprule
Clocks & Stacks & Location & Transitions & \cmax{} & Holes & Time(sec) & Memory(KB) &  Empty(Y/N) \\
      \hline
      2 & 2 & 6 & 9 & 4 & 2 & 9.05609 & 54108 & Y \\
      \hline
      2 & 2 & 7 & 10 & 4 & 2 & 10.2249 & 55192  & Y \\
      \hline
      2 & 2 & 8 & 11 & 4 & 2 & 11.6579 & 61212 &  Y \\
      \hline
      2 & 2 & 9 & 12 & 4 & 2 & 12.7217 & 74732 &  Y \\
      \hline
      2 & 2 & 10 & 13 & 4 & 2 & 14.913 & 84176 &  Y \\
      \hline
      2 & 2 & 11 & 14 & 4 & 2 & 15.6878 & 93660 &  Y \\
      \bottomrule
    \end{tabular}
    
     }
  \caption{\atcrit{} Varying Locations}\label{tab:l-crit-location}
\end{table}

\end{document}